\documentclass[12pt,letterpaper]{article}
\usepackage{setspace}
\usepackage[margin=1in]{geometry}

\usepackage{hyperref}
\hypersetup{
	colorlinks=true,
    linkcolor=blue,
	urlcolor=cyan,
        citecolor=blue,
	pdfauthor={Author One, Author Two, Author Three},
	pdftitle={A simple article template},
	pdfsubject={A simple article template},
	pdfkeywords={article, template, simple},
	pdfproducer={LaTeX},
	pdfcreator={pdflatex}
}

\RequirePackage{crop}
\RequirePackage{graphicx}
\RequirePackage{caption}
\RequirePackage{amsmath}
\RequirePackage{array}
\RequirePackage{color}
\RequirePackage{xcolor}
\RequirePackage{amssymb}
\RequirePackage{flushend}
\RequirePackage{stfloats}
\RequirePackage[figuresright]{rotating}
\RequirePackage{chngpage}
\RequirePackage{totcount}
\RequirePackage{fix-cm}

\RequirePackage{natbib}
\RequirePackage{wrapfig}%

\RequirePackage{amsthm}%

\RequirePackage[countmax]{subfloat}
\usepackage{anyfontsize}%
\RequirePackage{multirow}
\RequirePackage{footnote}
\RequirePackage{url}
\RequirePackage{amsmath}
\RequirePackage{mathrsfs}
\RequirePackage{algorithm}%
\RequirePackage{algorithmicx}%
\RequirePackage{algpseudocode}%
\RequirePackage{listings}%

\RequirePackage{tikz}

\usepackage{graphicx}
\usepackage{enumerate}
\usepackage{natbib}
\usepackage{xcolor}
\usepackage[xcolor,pdftex]{changebar}
\usepackage{url} 
\usepackage[utf8]{inputenc}
\usepackage{pgf, tikz}
\usepackage{adjustbox}
\usepackage{xcolor}
\usepackage{ulem}

\usepackage{pgfpages}
\usepackage{subcaption}
\usepackage{graphics}

\usepackage{amsmath}
\usepackage{amsthm}
\usepackage{amsfonts}
\usepackage{amssymb}

\theoremstyle{plain}
\newtheorem{theorem}{Theorem}
\newtheorem{corollary}[theorem]{Corollary}
\newtheorem{lemma}[theorem]{Lemma}

\theoremstyle{definition}
\newtheorem{definition}[theorem]{Definition}

\newcommand{\E}{\mathbb{E}}
\newcommand{\Sp}{\mathbb{S}}
\newcommand{\Hy}{\mathbb{H}}
\newcommand{\M}{\mathcal{M}}
\newcommand{\Mf}{\mathfrak{M}}

\newcommand{\D}{\mathcal{D}}
\newcommand{\OO}{\mathcal{O}}

\newcommand{\vect}[1]{\boldsymbol{#1}}
\newcommand{\norm}[1]{||{#1}||}
\newcommand{\R}{\mathbb{R} }
\newcommand{\W}{\mathbb{W} }
\newcommand{\Q}{\mathbb{Q} }
\newcommand{\C}{\mathbb{C} }
\newcommand{\al}[1]{\begin{align*}#1\end{align*}}

\newcommand{\acos}{\mathrm{acos}}

\newcommand{\Sec}{\mathrm{sec}}

\DeclareMathAlphabet{\mathpzc}{OT1}{pzc}{m}{it}

\makeatletter
\newcommand*{\triangleline}{\mathpalette\@triangleline\relax}
\newcommand*{\@triangleline}[2]{{%
    \sbox0{\m@th$#1\bigtriangleup$}%
    \dimen@\fontdimen8
       \ifx#1\displaystyle\textfont\else
       \ifx#1\textstyle\textfont\else
       \ifx#1\scriptstyle\scriptfont\else
       \scriptscriptfont\fi\fi\fi 3
   \ooalign{\hfil\hbox{\vrule\@width\dimen@\@height\ht0\@depth-.5\dimen@}\hfil\cr\box0}%
}}
\makeatother

\DeclareMathOperator*{\argmin}{arg\,min}
\DeclareMathOperator*{\logit}{\text{logit}}
\DeclareMathOperator*{\expit}{\text{expit}}

\let\hat\widehat
\let\tilde\widetilde
\let\bar\overline

\title{Asymptotically Normal Estimation of Network Curvature
}
\author{Steven Wilkins-Reeves\\Department of Statistics\\University of Washington\\\texttt{stevewr@uw.edu} \and Tyler H. McCormick\\Department of Statistics\\Department of Sociology\\University of Washington\\\texttt{tylermc@uw.edu}}

\date{
}

\begin{document}

\maketitle
	
	\begin{abstract}
		Network data, commonly used throughout the physical, social, and biological sciences, consist of nodes (individuals) and the edges (interactions) between them. One way to represent network data's complex, high-dimensional structure is to embed the graph into a low-dimensional geometric space. The curvature of this space, in particular, provides insights about the structure in the graph, such as the propensity to form triangles or present tree-like structures. We derive an estimating function for curvature based on triangle side lengths and the length of the midpoint of a side to the opposing corner. We construct an estimator where the only input is a distance matrix and also establish asymptotic normality. We next introduce a novel latent distance matrix estimator for networks and an efficient algorithm to compute the estimate via solving iterative quadratic programs. We apply this method to the Los Alamos National Laboratory Unified Network and Host dataset and show how curvature estimates can be used to detect a red-team attack faster than naive methods, as well as discover non-constant latent curvature in co-authorship networks in physics.  The code for this paper is available at \url{https://github.com/SteveJWR/netcurve}, and the methods are implemented in the R package \url{https://github.com/SteveJWR/lolaR}.
	\end{abstract}

\section{Introduction} \label{sec:Introduction}
Networks or graphs, $G = (V,E)$, are widely used in multiple scientific fields. These objects characterize the relations between a set of $n$ nodes (or vertices) $V = \{1,2,\dots, n\}$ via a set of edges $E \subset \{(i,j)| i \not= j, i \in V, j \in V \}$ representing the connections between the nodes. Social networks, where nodes often represent individuals, are a common application of these models (e.g., \citet{Borgatti2009NetworkSciences}). Networks are also commonly seen in models of biological and physical sciences, such as nodes representing cells in neuroscience \citep{Bassett2018OnNeuroscience} and particles \citep{Papadopoulos2018NetworkGrains}. Complex and high-dimensional structure is an inherent feature of network data, which poses challenges in modeling and representation. 

One common modeling approach represents the graph through an embedding into a lower-dimensional geometric space, where both the properties of the geometric space and the positions of points within it provide insights into graph structure.  Models that use this embedding representation rely on latent distance matrices \citep{Hoff2002latent, Handcock2007Model-basedNetworks, hoff2007modeling, Smith2017TheData}, where the distances in the low-dimensional manifold are inversely proportional to the propensity to form a connection.  We refer to this broad class as latent distance models.  
The geometry of the underlying manifold has substantial implications for the types of connections we expect to see in the network.  Manifolds with positive curvature (hyperspheres) tend to encourage triangles to close and produce more group structure, whereas negatively curved manifolds make it easier to form trees with long paths.  

Our focus is on estimating the curvature of a manifold based on a set of (noisy) distance measurements.  Our results rely on the fundamental, but profound, observation that the properties of triangles depend on the curvature of the manifold on which they're embedded.  In particular, we leverage the fact that the distance between a vertex in a triangle and the midpoint of the side opposite of that vertex is varies based on the curvature of the manifold.    \cite{lubold2023identifying} study a similar problem of hypothesis testing the geometric class of various network models. However, we offer an alternative approach to estimating curvature with several distinct advantages. Notably, our method requires only four distances, derived from lengths on a triangle to estimate curvature, making it a local approach that can be used to develop tests for constant curvature. Furthermore, our estimating function is smooth, allowing a single equation to estimate the curvature and derive interpretable asymptotic results, unlike \cite{lubold2023identifying}, which requires an eigenvalue equation that is generally not smooth and does not lead to interpretable asymptotics.

Our contributions include the following. We develop a smooth estimating equation to estimate curvature using a noisy distance matrix, leveraging triangles and their median lengths.  We call the triangle median the distance from one vertex to the midpoint of the other two (not to be confused with the statistical median). These results are general and applicable to any noisy distance matrix, though in this paper we apply them specifically to social networks. We next consider various aspects of working with surrogate midpoints when a true midpoint is not observed in the data. We establish upper and lower bounds for the curvature when collecting a set of distances that does not contain the midpoint of another pair of points. We also establish ``good conditions" under which surrogate midpoints form arbitrarily close to the actual midpoint of other points. Next, we turn to the specifics of the latent distance model. In this case, we present a curvature estimator and demonstrate that, under the typical assumptions used to fit latent distance models, it is asymptotically normal.  We show that we can further improve estimation using a constrained estimator that reflects the triangle inequality among distance constraints.  Lastly, we demonstrate that our estimator is a basis for the development of new methodology in sociology and cybersecurity by testing for changes in curvature.

The remainder of the paper continues as follows. First, we include a literature review in Section~\ref{sec: lit_review}. Next, in Section~\ref{sec:methods}, we introduce our aforementioned methodological contributions. We next illustrate the efficacy of these methods through a simulation study. Then, we discuss downstream statistical tasks such as testing whether the curvature of a noisy distance matrix is constant in Section~\ref{sec:application_constant_curvature} and detecting changepoints in Section~\ref{sec:application_changepoint}. We further elaborate on these with applications to co-authorship networks in physics and an application in cybersecurity.

\subsection{Literature Review} \label{sec: lit_review}
The use of distance matrices for data analysis is prevalent across numerous fields. Originating from applications in psychometrics, multidimensional scaling (MDS) \citep{torgerson1952multidimensional} pioneered the use of distance-based methods and has been explored in various domains, including the analysis of protein shapes \citep{havel1985evaluation}, image classification \cite{tenenbaum2000global}, and natural language processing \cite{kusner2015word}. Notably, \cite{Hoff2002latent} introduced this idea in a model of social network formation, which has since been expanded in numerous ways, such as model-based clustering \cite{Handcock2007Model-basedNetworks}, multi-view networks \cite{Salter-Townshend2017LatentData}, and dynamic networks \cite{kim2018review}. Some models use mixtures of a block structure to model only at the individual level within a cluster of the network \citep{Fosdick2016MultiresolutionModels, Lok2021ModelingBlockmodel}. Latent distance models have been applied to problems like modeling social influence \citep{Sweet2020AInfluence}, social media relationships of politicians \citep{Lok2021ModelingBlockmodel}, and neuron connectivity \citep{Aliverti2019SpatialClustering}, among others.

We focus on the properties of the geometric space underlying a latent distance model, particularly the notion of curvature. The sectional curvature of a latent space is broadly defined as the deviation from a flat (Euclidean) space via the growth of the circumference of small circles as a function of their radius. An important class of manifolds are those that are simply connected and have constant curvature. A classical result from \cite{Killing1891UeberRaumformen} characterizes these as the spherical (positive curvature), Euclidean (0 or flat curvature), and hyperbolic (negative curvature) spaces. Though importantly, these do not represent the entire class of such manifolds.

Classically, the choice of the embedding space was at the discretion of the analyst. Notably, latent spherical and Euclidean spaces were used in \cite{Hoff2002latent}. However, other metric spaces, particularly spherical and hyperbolic spaces, have been found to better represent many network data types \cite{Smith2017TheData}. The authors also provide a simulation-based approach to compare the eigenspectrum of the graph Laplacian to models under spherical, hyperbolic, and Euclidean geometry. They show how the latent embedding space influences observed properties of the network, notably the degree distribution and clustering of the network, which can influence the behavior of network contagion processes (i.e., SIR models) \citep{Volz2011EffectsDynamics}.

Our work bears the closest resemblance to that of \cite{lubold2023identifying}, which discusses hypothesis testing of the latent space among a class of models and the estimation of the related distance matrix. Our work is distinct in several ways. Firstly, our method of estimating the curvature of the latent space is novel and useful for deriving interpretable asymptotic results. This is due to the fact that our method allows for an estimating equation approach to identifying curvature, which leads to desirable properties. Importantly, their approach tests whether the geometry can be embedded globally in each of the canonical spaces, whereas we provide a local approach derived from triangle distances. Furthermore, we also provide an improved latent distance estimator which allows for the construction of an asymptotically normal distance matrix, based on cliques (fully connected subgraphs) in a network. Our approach for curvature estimation is modular and can be applied to general distance matrices. Consequently, we illustrate how to test for constant curvature within an embedding space. 

An alternative definition of graph curvature worth discussing includes the Ollivier-Ricci curvature \citep{Ollivier2007RicciSpaces} and extensions such as Haantjes-Ricci curvature \citep{Saucan2020ANetworks} and Forman-Ricci curvature \citep{Leal2018Forman-RicciHypergraphs}. These definitions of curvature are derived from metrics arising from graph distances (i.e., integer-valued shortest path distances) rather than distances on a smooth latent space. As such, it is not apparent what these estimates will converge to (or if they converge at all) when a network is studied as a random object, with the exception of \citet{vanderHoorn2020Ollivier-RicciGraphs}, who study a problem where connections are governed by a small radius on a latent space. The authors study the convergence of a modified Ollivier-Ricci curvature to the Ricci curvature of the underlying space in random geometric graphs under the limit of the connection radius shrinking to $0$. However, these discrete curvatures have also been applied to various settings, such as financial network instability \citep{Sandhu2016RicciRisk, Samal2021NetworkInstability}, network sampling \citep{Barkanass2022GeometricNetworks}, cancer detection in gene regulatory networks \citep{Sandhu2015GraphNetworks}, functional neuroscience \citep{Farooq2019NetworkConnectivity}, and community detection \citep{Sia2019Ollivier-RicciNetworks, Ni2019CommunityFlow}.

\section{Methods} \label{sec:methods}

We begin by formally introducing our environment, including defining the properties of manifolds covered by our method.  Next, we propose an estimator of curvature based on noisy distance measurements from triangle midpoints. We will begin with an ideal estimator when the true midpoint is measured and then follow up with a study of using points that are nearly midpoints, which we call surrogate midpoints, in their place.  These methods are general and apply in any setting where we have measured noisy distances.  We then turn to our setting--social networks--and describe in detail how to construct distance estimators, and subsequent estimates of curvature, for the latent distance model.

\subsection{Geometric Environment} \label{sec: geometric environment}

We begin by defining the geometric environment. We assume points lie on a Riemannian manifold $\M^p$ of dimension $p$, equipped with a corresponding metric tensor $g$. The metric tensor can be used to define the sectional curvature of the manifold at a point $q \in \M^p$. For our purposes, we assume that this manifold is connected and that the curvature is both upper and lower bounded. In our problem, we exclusively work with distances and thus consider the metric space induced by the Riemannian manifold $\Mf = (\mathcal{M}^p, d)$. We include the related definitions of the metric tensor and the distance on the manifold in the appendix in Section~\ref{sec: riemannian geometry definitions}.

We further assume the manifold is a member of the class of simply connected Riemannian manifolds with constant sectional curvature $(\kappa)$.  This assumption is consistent with work on the latent distance model in social networks \citep{Hoff2002latent}. These include the classical Euclidean $\E^p$ $(\kappa = 0)$, spherical $\Sp^p(\kappa)$ $(\kappa > 0)$, and, more recently, hyperbolic space $\Hy^p(\kappa)$ $(\kappa < 0)$. The celebrated Killing-Hopf theorem states that these are the only manifolds of this type \citep{Killing1891UeberRaumformen,hopf1926clifford}. Therefore, we refer to these manifolds spaces as the \textbf{canonical manifolds}, and we introduce common representations of these manifolds later in this subsection.

To identify the curvature, we rely on a simple geometric insight, relating the side lengths of a triangle and the length of the triangle's median; the line segment connecting a vertex to the midpoint of the opposite side.

In general, we require the manifold to satisfy two properties.
\begin{itemize}
    \item (A1) (\textbf{Algebraic Midpoint Property}) $\Mf$ satisfies the algebraic midpoint property. For any $x, y \in \M^p$, there exists a point $z$ such that $d_{\M^p}(z, x) = d_{\M^p}(z, y) = \frac{1}{2}d_{\M^p}(x, y)$. 
    \item (A2) (\textbf{Locally Euclidean}) For all $q \in \M^p$, there exists some $\delta > 0$ and some functions $c_p, C_p$ such that for all $\epsilon \leq \delta$: 
    $$ c_p(\epsilon) \leq \frac{\text{Vol}(B_{\M^p}(\epsilon, q))}{\text{Vol}(B_{\E^p}(\epsilon, 0))} \leq C_p(\epsilon) \text{ and } \lim_{\epsilon \to 0} \frac{\text{Vol}(B_{\M^p}(\epsilon, q))}{\text{Vol}(B_{\E^p}(\epsilon, 0))} = 1.$$ 
\end{itemize}

Here, $B_{\M^p}(\epsilon, q)$ is the $\epsilon$-ball on $\M^p$ centered at a point $q$, which we abbreviate to $B(\epsilon, q)$, and $B_{\E^p}(\epsilon, 0)$ is the $\epsilon$-ball in Euclidean space. These conditions will hold under our mild assumptions on the manifold. See Section~\ref{sec:manifold_assumptions} in the supplementary materials for details. Next, we consider an explicit set of representations of the canonical manifold.

\textbf{Canonical Manifolds.} 

Each of the canonical manifolds can be represented using a set of positions with real-valued vectors and a corresponding distance function, allowing for closed-form computation of the distances. We include definitions for Euclidean, spherical, and hyperbolic spaces for completeness. We emphasize the difference between the intrinsic and extrinsic geometry here. Though each of these canonical manifolds is embedded in $\R^p$, only the Euclidean space uses the standard $2$-norm to construct the distances. The curved canonical manifolds $\Sp^p$ and $\Hy^p(\kappa)$ can be embedded in $\R^{p + 1}$ along with a properly defined metric. In the spherical example, we compute the length according to the path length on the surface of the sphere, rather than the Euclidean distance through the sphere. We next highlight each of these models.

The \textbf{Euclidean manifold}: $\E^p$ can be described using a set of points in $\mathbb{R}^p$ with the standard $2$-norm.
\[
    d_{\E^p}(x, y) = \sqrt{\sum_{k = 1}^p (x_k - y_k)^2}
\]

The \textbf{spherical manifold}: $\Sp^p$ with curvature $\kappa > 0$ is equivalent to the sphere of radius $r = \frac{1}{\kappa^2}$. We express this using a set of coordinates $x \in \R^{p + 1}$, such that $ \sum_{k = 0}^{p}x_k^2 = 1$. The distance on the sphere can be computed using the quadratic form $B_{\Sp^p}$ defined below:
\[
    B_{\Sp^p}(x, y) = \sum_{k = 0}^{p}x_k y_k , \quad  d_{\Sp^p}(x, y) = \frac{1}{\sqrt{\kappa}}\arccos(B_{\Sp^p}(x, y)).
\]

The \textbf{hyperbolic manifold}: $\Hy^p$ with curvature $\kappa < 0$ can be constructed using the hyperboloid model (Minkowski model), which corresponds to a set of points $x \in \R^{p + 1}$, such that $x_{0}^2 - \sum_{k = 1}^{p}x_k^2 = 1, x_0 > 0$. An analogous quadratic form $B_{\Hy^p}$ exists for the hyperbolic embedding and can be used to compute the distances:
\[
    B_{\Hy^p}(x, y) = x_{0}y_{0} - \sum_{k = 1}^{p}x_k y_k , \quad d_{\Hy^p}(x, y) = \frac{1}{\sqrt{-\kappa}}\text{arccosh}(B_{\Hy^p}(x, y)).
\]

\subsection{Identifying Curvature } 
\label{sec: identifying_curvature} 

A number of methods exist to verify whether a particular set of distances can be embedded in a space of constant curvature, typically based on the zeros of eigenvalues of a transformation of the distance matrix. These include methods by \citet{Schoenberg1935RemarksHilbert} and \citet{Begelfor2005TheManifolds}, as well as Cayley-Menger determinants \citep{Blumenthal1943DistributionN-Space}. \citet{lubold2023identifying} previously used such a criterion to identify whether a set of distances could be globally embedded in a particular curvature space. In this work, we take a different approach where we identify curvature based on a minimal set of points. We also estimate a specific curvature, rather than performing a test on the sign of the curvature, as presented by \citet{lubold2023identifying}.

We rely on a simple geometric observation. Consider a set of three points that form a triangle and the length of the median. This length reveals the sectional curvature of the manifold, where a smaller distance corresponds to a more negatively curved space and a larger distance corresponds to a more positively curved space. This is visualized in Figure~\ref{fig:triangle_midpoid_distances} for Euclidean, spherical, and hyperbolic triangles.

\begin{figure}
    \centering
    \includegraphics[width = 0.45\textwidth]{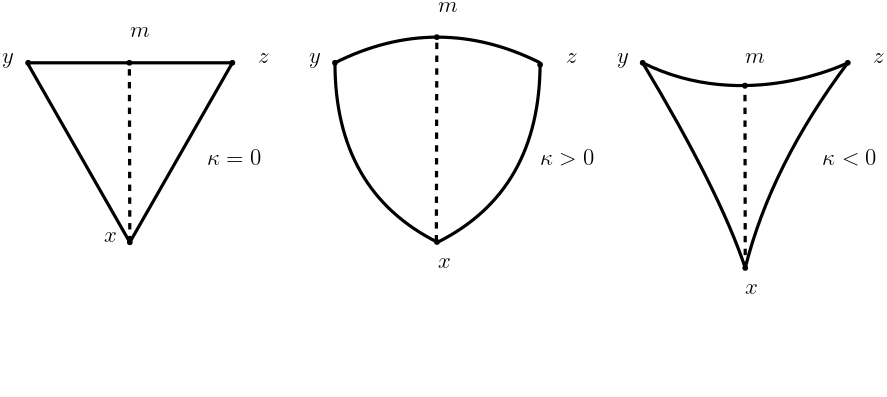}
    \caption{Midpoint distances and curvature of the space with equilateral triangles. The length of the triangle median $d_{xm}$ is an increasing function of the curvature $\kappa$ for fixed other triangle side lengths. }
    \label{fig:triangle_midpoid_distances}
\end{figure}

This use of midpoints helps to identify the curvature uniquely, when it may not be identified in other situations.  For example, consider $4$ points placed equidistant from each other.  This set of distances can be either represented in $\E^3$ as the tetrahedron, or using equidistant points on the sphere $\Sp^2(\kappa)$ for $\kappa > 0$, making it impossible to identify the curvature of the space from this collection of distances alone.

We instead take a more direct approach to eliciting curvature by leveraging distances in a triangle and the triangle median (which, recall, is the distance from one vertex to the midpoint of the other two and distinct from the statistical median). This allows us to identify curvature with only four points in total. We now formalize this intuition. For any three points $x, y, z$ which lie in an unknown canonical manifold of dimension $p \geq 2$ with constant sectional curvature $\kappa$, $(x, y, z)$ can be isometrically embedded in a submanifold of dimension $2$. We illustrate this through the use of submanifolds that contain their \textit{geodesics}, the paths on a manifold that minimize the path length between two points, and thus define a distance.  Simply stated, even if the manifold's dimension $p > 2$, the curvature can be identified through the totally geodesic submanifold containing the triangle.

\begin{definition} 
    A submanifold $\tilde \M \subset \M$ is \textit{totally geodesic} if every geodesic in $\tilde \M$ is also a geodesic in $\M$. 
\end{definition}

Some simple examples include the Euclidean plane, within the three-dimensional Euclidean space ($\mathbb{E}^2$ is a totally geodesic submanifold of $\mathbb{E}^3$). However, this is not always the case; consider the two-dimensional sphere $\mathbb{S}^2(\kappa)$, which also resides within three-dimensional Euclidean space but does not contain all of its geodesics, as these geodesics in $\mathbb{E}^3$ pass through the center of the sphere. This distinction highlights the differences between intrinsic and extrinsic notions of distance, as our object of study is the former. We will use the fact that a totally geodesic submanifold contains all points along the geodesic, including the midpoint.

\begin{lemma} \label{lem:submanifold_of_dim_2}
    If any $x,y,z \in \M^p(\kappa)$ where $p\geq 2$ and $x,y,z$ are not co-linear. Then  $x,y,z, m \in \M^2(\kappa)$ where $\M^2(\kappa)$ is a totally geodesic submanifold of dimension 2 with constant sectional curvature $\kappa$ and $m$ is the midpoint of points $y$ and $z$. 
\end{lemma}

The intuition behind this lemma is that, regardless of the ambient dimension of the latent manifold $p$, we can determine the curvature from a two-dimensional submanifold. This submanifold is constructed from the geodesics of a given triangle. As we will see in Theorem~\ref{thm:midpoint_curve_equation}, the side lengths and the length of the triangle's median are sufficient to identify the curvature of the manifold. The proof is straightforward and in Section~\ref{sec:proof_submanifold_of_dim_2}. The main implication here is that the totally geodesic submanifold allows us to look at the distances in a subspace of dimension $2$ which will be useful in the following theorem for identification. The three points $x,y,z$ will fall into one of these sub-manifolds, as well as $m$ which lies on the geodesic between $y$ and $z$.  Since geodesics determine the distance, and geodesics on the submanifold are the same as geodesics on the manifold.

We now use this fact to derive an equation which will relate the curvature $\kappa$ to the set of distances between the points $(x,y,z,m)$, which we denote $\vect{d}^{\triangleline} = (d_{xy},d_{xz},d_{yz}, d_{xm})$. 

\begin{theorem}[Midpoint Curvature Equation] \label{thm:midpoint_curve_equation}
    Suppose that points $x,y,z \in \M^p(\kappa)$ an unknown Riemannian manifold of dimension $p \geq 2$ of constant sectional curvature $\kappa$.  Let $m$ denote the midpoint between $y,z$.  The following equation holds for $\kappa \in \R$. 
    \begin{equation} \label{eq: midpoint_curvature_equation}
    \begin{aligned}
      g(\kappa, \vect{d}^{\triangleline}) &= \text{Re}\bigg[ \frac{2\cos(d_{xm}\sqrt{\kappa})}{\kappa} - \frac{\Sec(\frac{d_{yz}}{2}\sqrt{\kappa})(\cos(d_{xy}\sqrt{\kappa}) + \cos(d_{xz}\sqrt{\kappa}))}{\kappa} \bigg]  = 0
    \end{aligned}
    \end{equation}
    Where $d_{jk}$ denoted the distance between points $j,k$ and $\text{Re}[]$ denotes the real part of the equation. 
\end{theorem}

The proof first leverages the fact that by Lemma~\ref{lem:submanifold_of_dim_2}, we can construct a submanifold of dimension 2 that contains the midpoint of points on a triangle. We then use this to derive an equation that relates the side lengths of the triangle, the triangle median length, and the curvature of the space. For cases when $\kappa < 0$, we take the real part of equation~\eqref{eq: midpoint_curvature_equation}, which is equivalent to replacing the trigonometric functions with their hyperbolic analogues. Though this does require that $p \geq 2$, this covers most manifolds of interest. The proof is found in the supplementary materials in Section~\ref{sec:proof_midpoint_curve_equation}. It will also be convenient to express the length of the triangle median $d_{xm}(\kappa; \vect{d}^{\triangle}(x, y, z))$ as a function of the curvature $\kappa$ and a set of triangle side lengths $\vect{d}^{\triangle}(x, y, z)$.

We remark on the a similar criterion used by \cite{gu2018learning} for identifying the number of components of positive and negative curvature in a product-space embedding. ~\cite{gu2018learning} derive the parameter $\theta_T$ from Toponogov's theorem (which can be found in the appendix Theorem~\ref{thm:toponogov}) and its sign can be used to identify the sign of the curvature of the manifold. In fact, it is straightforward to show $\lim_{\kappa \to 0} g(\kappa, \vect{d}^{\triangleline}) = \theta_T(\vect{d}^{\triangleline})$.
If the sectional curvature of $\M^p$ is positive, then $\theta_T(\vect{d}^{\triangleline}) > 0$, and if the curvature is negative, $\theta_T(\vect{d}^{\triangleline}) < 0$. In Euclidean space, $\theta_T(\vect{d}^{\triangleline}) = 0$, and this reduces to the parallelogram law. Our method is distinct as it directly identifies the curvature using this same set of distances, rather than the sign of the curvature. 

If the true manifold that generates the distances is non-constant, then the value $\kappa$ is the curvature of the 2-dimensional canonical manifold which can isometrically embed a set of triangle distances $\vect{d}^{\triangleline}$.

Our method differs from those of \citet{Schoenberg1935RemarksHilbert} and the Cayley-Menger determinants \citep{Blumenthal1943DistributionN-Space} because it requires midpoint information, which theirs do not. Their methods do not uniquely identify curvature but only determine if a space of constant curvature can embed an given distance matrix. In the appendix, we illustrate the smoothness of this equation, which is also a desirable property for plug-in estimators (Section~\ref{sec: Smoothness of the Estimating Equation}).

\subsection{Estimating Curvature} \label{sec: Estimating Curvature}

In this subsection, we describe how to estimate the curvature from a noisy estimate of a set of distances using Theorem~\ref{thm:midpoint_curve_equation}. We begin by introducing an estimator based on triangle distances and its median length. In practice, we are often given an estimate of a distance matrix between an arbitrary set of $K$ points $\vect{D} \in \R^{K \times K}$. In this setting, it is not guaranteed that there is a midpoint of two other points among the observed points in the distance matrix. We  illustrate how one can bound the curvature in this setting. Lastly, we introduce a result that characterizes the formation of points arbitrarily close to the midpoint of other points.

We first consider an idealized scenario. In this setting, we suppose that we are given an estimate of the triangle distances $\vect{\hat d}^{\triangleline}$. We return to the problem of estimating distances in our model in Section~\ref{sec: Latent Distance Estimation}. We can use such a set of noisy or estimated distances to estimate the curvature ($\hat \kappa$) by solving for the value of $\kappa$ which is the solution to equation~\ref{eq: curvature estimate}.
\begin{equation} \label{eq: curvature estimate}
    g(\hat \kappa, \vect{\hat d}^{\triangleline}) = 0
\end{equation}

The advantage of this method is that the smoothness of $g$ allows for the derivation of explicit asymptotics for $\hat \kappa$, which is not possible for the method used by \cite{lubold2023identifying}. We present this result below in Theorem~\ref{thm:asymptotic_normality}.

\begin{theorem} \label{thm:asymptotic_normality}
    
    Suppose there exist points $x,y,z \in \M^{p}$. Let $m$ denote the midpoint between $y,z$ where these points are fixed.  Let $\vect{\hat d}^{\triangleline} = (\hat d_{xy},\hat d_{xz},\hat d_{yz}, \hat d_{xm})$ be the estimated distances and $\vect{d}^{\triangleline} = (d_{xy},d_{xz},d_{yz}, d_{xm})$ be their true, unknown counterparts.
    
    Assume we have a distance estimator $\vect{\hat d}$ such that
    \begin{itemize}
        \item (B1) $r(n)\Big(\vect{\hat d}^{\triangleline} - \vect{d}^{\triangleline}\Big) \to N(0, \Sigma)$
        \item (B2) $\kappa  < \left(\frac{\pi}{\max\{d_{xy},d_{xz},d_{yz}, d_{xm}\}}\right)^2$
        \item (B3) $\frac{d}{d\kappa'} d_{xm}(\kappa'; \vect{d}^{\triangleline}(x,y,z)) \bigg|_{\kappa' = \kappa} > 0$
    \end{itemize}
   
    Where $r(n)$ is the rate of convergence. 
    Let $\hat \kappa$ be the solution to $g(\kappa, \hat d) = 0$. Then 
    \begin{equation}
        r(n)(\hat \kappa - \kappa) \to_{d} N \left(0, \left(\frac{\partial g(\kappa,d)}{\partial \kappa}\right)^{-2}\left(\nabla_d g(\kappa,d)^\top \Sigma \nabla_d g(\kappa,d)\right)\right)
    \end{equation}
    where $\to_d$ refers to convergence in distribution. 
    If $(B1)$ is replaced by a consistency, i.e. $\norm{\vect{\hat d}^{\triangleline} - \vect{d}^{\triangleline}} = o_P(r(n))$, then $\hat \kappa - \kappa = o_P(r(n))$. 
\end{theorem}

The proof is found in the supplementary materials in Section~\ref{sec:asymptotic_normality} and is an application of the implicit function theorem together with the delta method. Assumption (B1) is mild as it only requires asymptotic normality of the distance estimator. Assumption (B2) is trivial since it simply requires that the true distances are not greater than the maximum allowable distances on the sphere of curvature $\kappa$. Assumption (B3) tends to hold unless the three points $x, y, z$ are collinear. For a more in-depth discussion of the non-decreasing property of the midpoint, see Lemma~\ref{lem:non_decreasing_midpoint_length} in the appendix. In Section~\ref{sec: Latent Distance Estimation}, we illustrate the asymptotic normality of a distance estimator based on cliques.  We next discuss the bias associated when a set of pairwise distances are observed, where no point is necessarily a midpoint of another pair of points.

In general, rather than distances from a triangle median, we may only observe distances in the form of a distance matrix $\vect{D} \in \mathbb{R}^{K \times K}$. There might not be a midpoint between two other points within this set. In this case, given a triangle, we can use a point nearby to the midpoint as a \textit{surrogate midpoint} $m'$ taking the role of a midpoint between $y, z$.

We let $\kappa'$ be the solution to equation~\eqref{eq: curvature estimate} where $d_{xm'}$ takes the place of $d_{xm}$. In this case, we can approximate the bias of the curvature estimate using a Taylor series expansion.
\al{
    0 &= g(\kappa', \vect{d}^{\triangleline '}) - g(\kappa, \vect{d}^{\triangleline}) \\
    &= \nabla_{\kappa}g(\kappa, \vect{d}^{\triangleline})(\kappa - \kappa') + \nabla_{d_{xm}}g(\kappa, \vect{d}^{\triangleline})(d_{xm} - d_{xm'}) + o(|\kappa - \kappa'| + |d_{xm} - d_{xm'}|) \\ 
    \implies |\kappa - \kappa'| &\approx (\nabla_{\kappa}g(\kappa, \vect{d}^{\triangleline}))^{-1}\nabla_{d_{xm}}g(\kappa, \vect{d}^{\triangleline})|d_{xm} - d_{xm'}| \\
    &\leq (\nabla_{\kappa}g(\kappa, \vect{d}^{\triangleline}))^{-1}\nabla_{d_{xm}}g(\kappa, \vect{d}^{\triangleline})d_{mm'}.
}
Therefore, the bias will scale approximately linearly as a function of $d_{mm'}$ for small values of $d_{mm'}$.

If four points are within a manifold of constant curvature $x, y, z, m' \in \M^p(\kappa)$, then using the curvature value $\kappa$, one can compute the distance from the midpoint $m$ to $m'$ by letting $m'$ take the place in equation~\eqref{eq: midpoint_curvature_equation} and solving for $d_{mm'}$. As such, we denote this $d_{mm'}(\kappa) := d_{m' m}(\kappa; \vect{d}^{\triangle}(y, z, m'))$ as a function of the curvature. Using the triangle inequality, we can then upper and lower bound the curvature by replacing $d_{xm}$ with upper and lower bounds in equation~\eqref{eq: midpoint_curvature_equation} and solving the corresponding equations. We illustrate this in Theorem~\ref{thm:curve_bounds}.

For a given $\kappa$, let $\vect{d^{\triangleline, +}}(\kappa) = (d_{xy}, d_{xz}, d_{yz}, d_{xm'} + d_{mm'}(\kappa))$ and $\vect{d^{\triangleline, -}}(\kappa) = (d_{xy}, d_{xz}, d_{yz}, d_{xm'} - d_{mm'}(\kappa))$.

\begin{theorem}[Curvature Bounds] \label{thm:curve_bounds}
    Let $x, y, z, m' \in \M^p(\kappa)$. Then let $d_{jk}$ denote the distance between points $j, k \in \{x, y, z, m'\}$. Let $\kappa_u$ and $\kappa_l$ denote the solutions 
    \[
        g(\kappa_u, \vect{d^{\triangleline, +}}(\kappa_u)) = 0 , \quad g(\kappa_l, \vect{d^{\triangleline, -}}(\kappa_l)) = 0
    \]
    then $\kappa_{l} \leq \kappa \leq \kappa_{u}$.
\end{theorem}

When the surrogate midpoint and the true midpoint are the same, then $d_{mm'} = 0$ and the upper and lower bounds converge. Similar to the curvature estimate $\hat \kappa$, given a noisy estimate of the distances, we can estimate the upper and lower bounds of the curvature. In Section~\ref{sec:application_constant_curvature}, we leverage these bounds to develop a test of constant curvature. We next address the formation of surrogate midpoints arbitrarily close to the midpoints of other pairs of points.

We next provide an outline involving how fast we can expect surrogate midpoints to form. In order to do so, we first introduce a useful definition. A subset $A \subset \M^p$ is \emph{geodesically convex} if the geodesic between any two points in $A$ is contained within $A$ itself. Here, convexity on a manifold will refer to geodesic convexity. In Theorem~\ref{thm:midpoint_convergence}, we let $h_{m,3}(m_1,m_2)$  denote the joint density function of a pair of midpoints with a shared endpoint and $h_{m,4}(m_1,m_2)$  denote the joint density of two midpoints without a shared endpoint. These two densities will be functions of the unknown manifold $\M$ and the distribution of the positions $Z'$, $G_{Z'}$.

\begin{figure}[htb!]
    \centering
    \includegraphics[width = 0.45\textwidth]{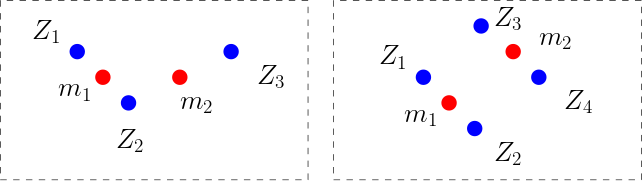}
    \caption{Left side illustrates two midpoints with a shared endpoint with joint density $h_{m,3}$, while the right side illustrates a joint density with no shared endpoints $h_{m,4}$. Endpoints, i.e., sampled positions of $Z_i$ are shown in blue while midpoints of the pairs are shown in red.}
    \label{fig:3 vs 4 point midpoints}
\end{figure}

\begin{theorem}\label{thm:midpoint_convergence}
    Suppose that $Z'_i \stackrel{iid}{\sim} H_{Z'}$ are $K$ points sampled iid from a distribution on a simply connected manifold $\M^p$. Denote this set of points $\{Z'_i\}_{i = 1}^K = \D_K$. Suppose there exists a convex region $A$ for which 
    \begin{itemize}
        \item (C1) $h(z) \geq \alpha > 0, \text{ for all } z \in A$
        \item (C2) $\dim(A) = p \geq 2$
        \item (C3) $h_{m,3}(m_1, m_2), h_{m,4}(m_1, m_2) \leq \alpha_m < \infty$
    \end{itemize}
    for all $z \in A$ where $f$ is the density function corresponding to $G_Z$. Let $m(y,z)$ denote the midpoint between two points $y$ and $z$. Define the statistic 
    \[
    \Phi(\D_{K}) := \min_{x,y,z \in \D: x \not= y, x \not= z, y \not= z} d\Bigl( m(y,z), x \Bigr)
    \]
    which is the minimum distance from an observed point to the midpoint of another pair of points. 
    Then 
    \begin{equation}
        \Phi(\D_{K}) = \OO_P(K^{-3/p}). 
    \end{equation}
\end{theorem}

Our approach to demonstrating the above result draws on the work of \citet{Cai2013DistributionsSpheres} and \citet{Brauchart2015RandomSeparation}, who discuss the convergence of the minimum distance between any two points sampled uniformly on a hypersphere. These authors show that \(\min_{i,j} d(Z_i, Z_j) = \OO_P(K^{-2/p})\), and they also derive an exact distribution for \(\min_{i,j} d(Z_i, Z_j)\) under the assumption of uniformity on the sphere. They use a technique that recursively computes the probability that a point is at least a radius \(\epsilon\) away from the previous \(K\) points. In our approach, at each placement of a new point, there are \(\binom{K}{2}\) midpoints instead of \(K\) current points, leading to the faster rate we observe.

Assumption (C1) ensures the existence of geodesically convex regions for which midpoints can form. Assumption (C2) ensures that this region has the same dimension as the ambient space. Lastly, Assumption (C3) is relatively mild as long as we have a smooth manifold and a continuous density. The proof is detailed in Section~\ref{sec:proof_midpoint_convergence} of the supplementary materials and relies on a result regarding medians of arbitrarily correlated random variables, which may be of independent interest (Theorem~\ref{thm:copula_theorem}).

\subsection{Reducing Bias and Variance in Curvature Estimation} \label{sec:practical_reference}
When given an estimate of a distance matrix and to later apply our method, as in many statistical problems, we are concerned with the bias and variance of our estimator. The variance, in general, will be driven by the shape of the triangle used to estimate the curvature, while the bias will be driven by the closeness of the surrogate midpoint to the true midpoint of a triangle. We first visualize this phenomenon and then follow up by providing some practical strategies for constructing good estimators.

We can visualize the theoretical variance by using our asymptotic result in Theorem~\ref{thm:asymptotic_normality}. In Figure~\ref{fig:theoretical_variance}, we plot the variance of the curvature estimate for a distance estimator with identity variance (i.e., $\vect{\hat d}^{\triangleline} \sim N(\vect{d}^{\triangleline 0}, I_{4 \times 4})$) for a variety of choices of $\kappa$. In this plot, we change the position of the vertex $x$ of the triangle (which we may also refer to as the reference point). The smallest variance reference points across all of these curvatures are the ones that form nearly equilateral triangles with the points $y, z$. Additionally, the variance of the estimator tends to be larger in a given reference location as the curvature $\kappa$ decreases. We plot all reference points $x$ within a ball of radius 2.  Secondly, we illustrate, for an equilateral triangle, the bias of the estimate of the curvature when moving the location of the surrogate midpoint $m'$. The spherical and hyperbolic spaces are shown using a projection where the distance to the center and relative angle are mapped onto Euclidean space. 

\begin{figure}[htb!]
    \centering
    \includegraphics[width=0.45\textwidth]{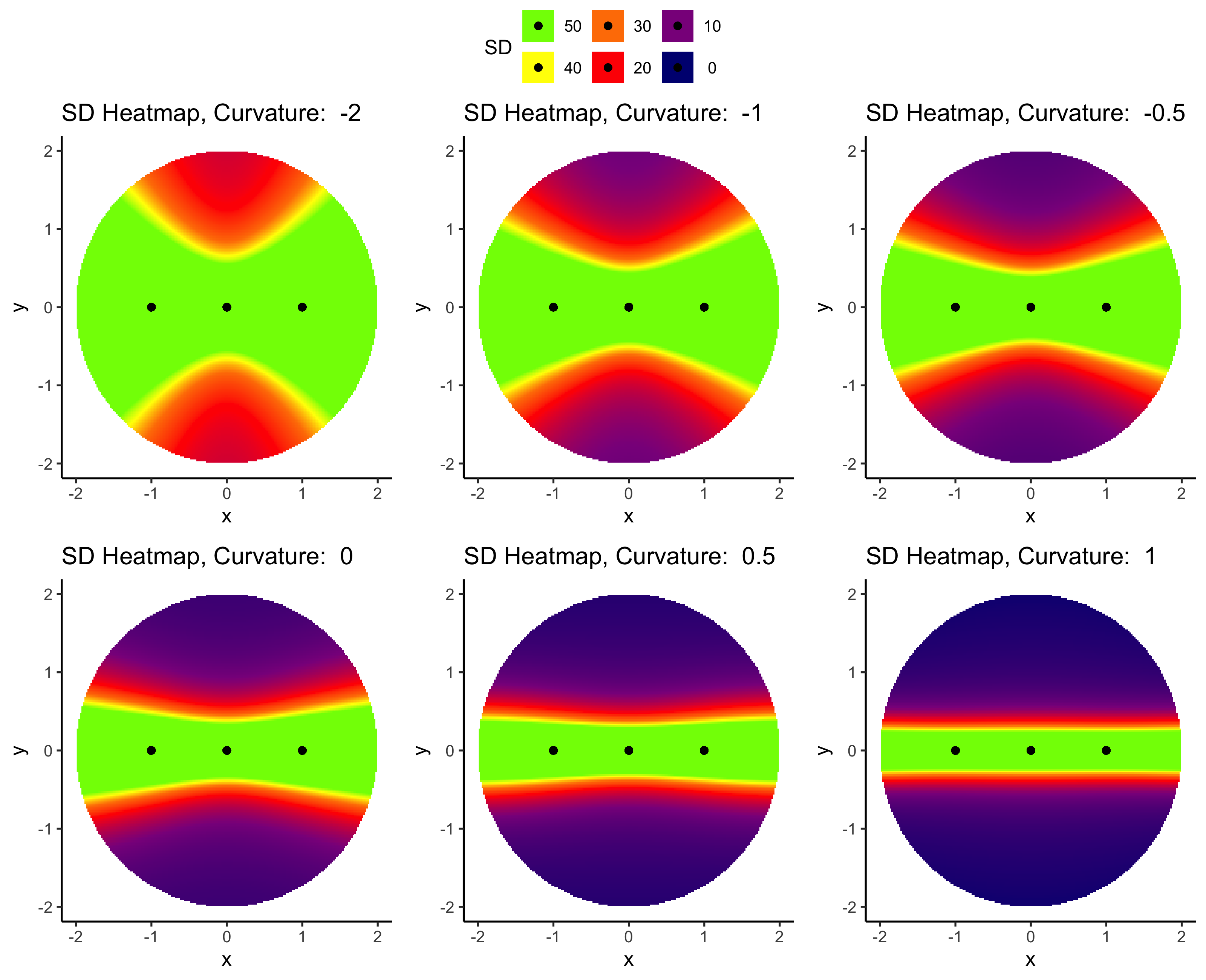}
    \caption{Variance as a function of $x$ position. True points $(y, m, z)$ in black.}
    \label{fig:theoretical_variance}
\end{figure}
\begin{figure}[htb!]
    \centering
    \includegraphics[width=0.45\textwidth]{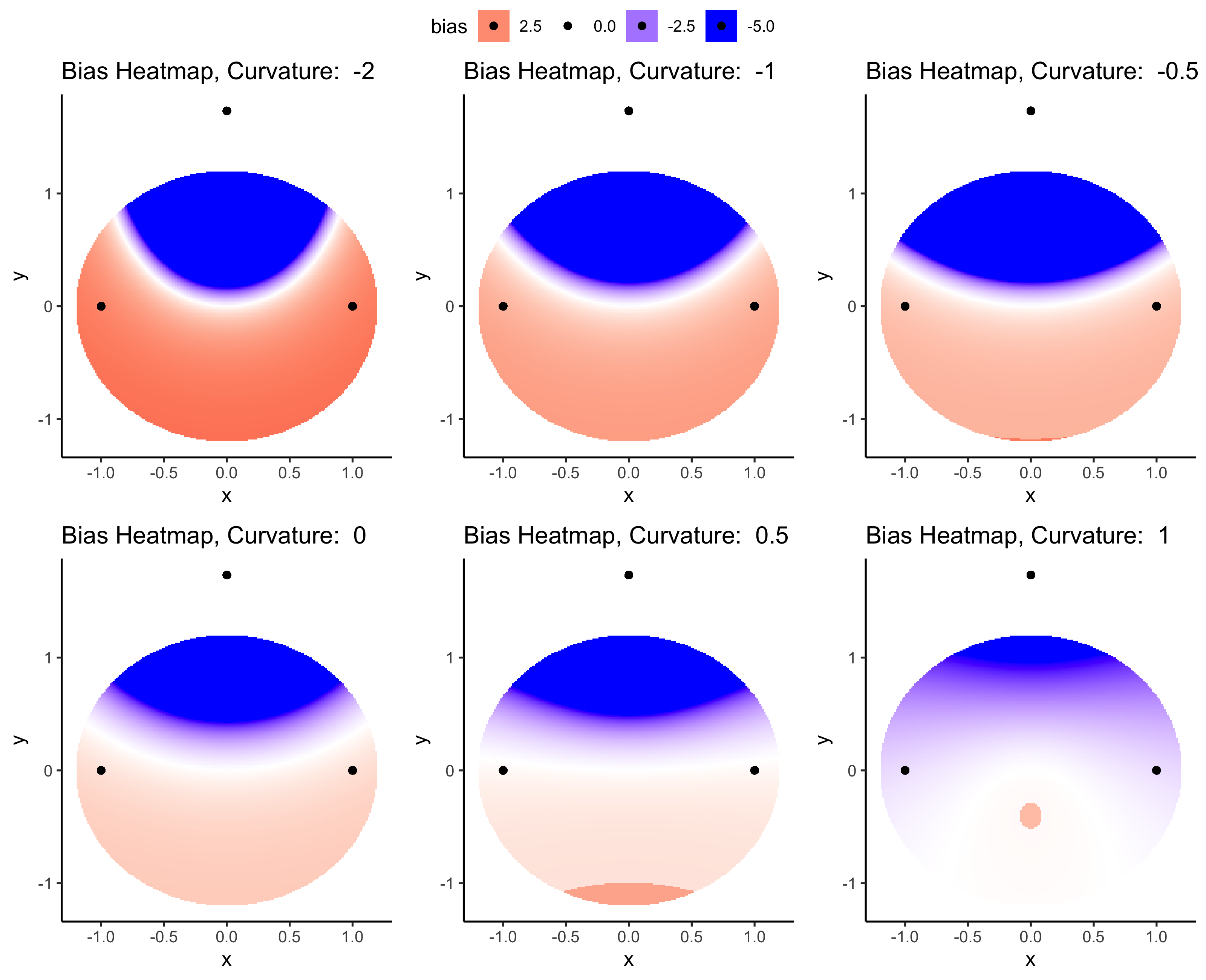}
    \caption{Bias as a function of surrogate midpoint $m'$ position. True points $(x, y, z)$ in black. For visualization purposes, the bias is capped at a magnitude of 5.}
    \label{fig:theoretical_bias}
\end{figure}

We next illustrate some practical choices to minimize the bias and variance of a curvature estimate from a distance matrix $\vect{\hat D}$.

\textbf{Choosing the Best Midpoint.} In practice, we would like to search over the space to find the best midpoint between two other points in the distance matrix. The midpoint between any two points $y, z$ is also known as the Fr\'echet mean, which is defined as follows:

\[
    m^* = \argmin_{m \in \M^p} d^2_{ym} + d^2_{zm}.
\]

 In practice we search over the space of candidate entries of a distance matrix $D$ to find the best midpoint available. We also want to ensure that the points $y, z$ are not too close to each other since this tends to lead to a high variance estimator. A reasonable option is to normalize this quantity by $d_{yz}$. Given a true midpoint $m^*$, then $d_{ym^*} = d_{zm^*} = \frac{1}{2}d_{yz}$, then 
\begin{equation} \label{eq:midpoint_obj}
    \frac{d^2_{ym^*} + d^2_{zm^*}}{d^2_{yz}} = \frac{1}{2}.
\end{equation}
We can also add a term $\frac{|d_{ym} - d_{zm}|}{d_{yz}}$ which aids in balancing the lengths of the distance to each point. 

Therefore, to compute the best surrogate midpoint set, we solve the following problem:
\begin{equation} \label{eq:surrogate_midpoint_problem}
    \hat y, \hat z, \hat m = \argmin_{y, z, m} \left( \frac{d^2_{ym} + d^2_{zm}}{d^2_{yz}} + \frac{|d_{ym} - d_{zm}|}{d_{yz}} \right).
\end{equation}

In cases where we are interested in measuring the curvature across multiple surrogate midpoint sets, we remove this set of points and solve equation~\eqref{eq:surrogate_midpoint_problem} for the remaining indices to construct a collection of surrogate midpoint set $\{(y^{(j)}, z^{(j)}, m^{(j)})\}_{j = 1}^J$.

\textbf{Selecting the Best Triangles.} Given a surrogate midpoint set $(y, m', z)$, we seek to find choices for the reference point $x$ which provide the lowest variance. A good rule of thumb is to search for triangles $x, y, z$ that are nearly equilateral.

We exploit this by considering a scaled version of the triangle inequality. Let $C_{\triangle} \in [1, 2]$ be a constant that determines the flatness of the allowed triangles $x, y, z$. Then we select only the $x$ such that
\begin{equation} \label{eq:triangle_tuning}
     d_{ij} + d_{jk} \geq C_{\triangle} d_{ik} \quad \forall (i, j, k) \in (x, y, z).
\end{equation}

If one believes that the curvature is constant across a surrogate midpoint set, we can estimate a single curvature by taking the median across the values of $x$. This tuning parameter allows us to pick the triangles closest to equilateral, which tend to give the best estimates of the curvature. Letting $C_{\triangle} = 1$ allows for all triangles, no matter how flat they are, and $C_{\triangle} = 2$ will only permit exact equilateral triangles. Setting $C_{\triangle}$ too large results in no triangles being found, and setting $C_{\triangle}$ too small will result in using triangles that are very flat, often suffering from high variance in the corresponding $\hat \kappa$ estimates. In practice, we find that a value of $C_{\triangle} \in [1.05, 1.7]$ is effective, and a good default choice is $1.3$.

From an estimated distance matrix $\vect{\hat D}$, we let  $\vect{X}_{\vect{\widehat D}, C_{\triangle}} = \{x : \text{ eq}~\eqref{eq:triangle_tuning}\}$ denote the set of reference points used for curvature estimation. We let $\hat \kappa_{x} = \kappa(\vect{\hat D}, \vect{\hat d}^{\triangleline}(x, y, z, m'))$, where $\vect{\hat d}^{\triangleline}(x, y, z, m')$ denotes the estimated distances of the triangle $(x, y, z)$ with surrogate midpoint $m'$. Let $\hat \kappa_{med} = \text{median}\{\hat \kappa_{x}\}$. Given $\vect{X}$, we can construct the median estimator $\hat \kappa_{med, \vect{X}}$ of the curvature in equation~\eqref{eq:median_kappa}.

\begin{equation} \label{eq:median_kappa}
    \hat \kappa_{med, \vect{X}} = \text{median}_{x \in \vect{X}} \kappa(\vect{\hat D}, y, z, m, x)
\end{equation}

We now turn to estimating the distance matrix $\vect{\hat D}$ in a latent distance model.


\subsection{Distance Estimation in Latent Distance Models for Social Networks} \label{sec: Latent Distance Estimation}
We consider the random undirected network corresponding to a graph $G = (V, E)$ where $|V| = n$, with adjacency matrix $A \in \{0,1\}^{n \times n}$ such that $A_{ij} = 1$ iff $(i, j) \in E$. We motivate our method through the latent distance model of network formation \cite{Hoff2002latent}. This is one specific choice of distance model that we will focus on, but it is not the only option.  In this model, locations $Z_i, Z_j$ are most generally equipped to some metric $d(\cdot, \cdot)$, forming a metric space, $\mathfrak{M}$. As before we consider the metric space generated from a manifold $\mathcal{M}^p$ of dimension $p$, where the probability of forming an edge is inversely proportional to the distance on the latent metric, $d_{\M}(Z_i, Z_j)$,

\begin{equation*}
  \Lambda\Biggl(P\biggl(A_{ij} = 1\biggr)\Biggr) \propto - d_{\M}(Z_i, Z_j)
\end{equation*}

where $\Lambda(\cdot)$ is a link function. Where convenient, we condense $d_{\M}(Z_i, Z_j)$ to $d_{ij}$ for brevity to denote the distance between two indices $i$ and $j$. A specific form we consider, is to include random effects (representing node level-gregariousness, or a degree correction), is as follows:

\begin{equation} \label{eq:ld_model}
    P(A_{ij} = 1|\nu, Z) = \exp\biggl(\nu_i + \nu_j - d_{\mathfrak{M}}(Z_i, Z_j)\biggr)
\end{equation}

where 
\[
    Z_i \sim_{iid} F_Z \quad \nu_i \sim_{iid} F_{\nu}
\]
with a corresponding measures $F_\nu$ having support on the non-positive real line $\text{supp}(F_\nu) \subset (\infty, 0]$. The latent position measure $F_Z$ has support on some the latent manifold $\M^p$. Here $iid$ refers to sampling identically and independently from the latent distribution. 

This model relates the latent distances of the generative model to the probability of a connection between two points. Although not discussed in this paper, simple generalizations can be derived for directed networks. This paper aims to estimate the curvature of $\M^p$. The $\exp(\cdot)$ link function is used as an example for ease of explanation in deriving the localization of positions in the latent space (this will be further discussed later in this section). However, in the supplementary materials (Sections~\ref{sec: Other Link Functions}, \ref{sec: Alternative Estimators of the Distance Matrix}), we discuss the extension to the $\expit(\cdot)$ model, which was the original link function proposed by \cite{Hoff2002latent}, as well as additional link functions. As we illustrate in Section~\ref{sec: Latent Distance Estimation}, this link function will be convenient for describing the formation and localization of latent positions within cliques, which are useful for estimating latent distances.

In this section, we illustrate a procedure for estimating a distance matrix based on latent positions in a latent distance model. We improve upon the results of \citet{lubold2023identifying} to construct an asymptotically normal distance estimator using \emph{cliques}, fully connected subgraphs of the graph $G$. This approach has two advantages. First, since cliques represent multiple nodes, we have multiple opportunities for connections between two cliques, meaning we can use the fraction of realized ties between two cliques as an estimate for the probability of connection between the two cliques. Second, cliques correspond to points near each other in space. As the size of the clique grows, the maximal distance between latent positions must be small (in fact, the maximum distance between nodes in a clique will converge exponentially fast with respect to the size of the clique). This means that we can consider cliques as "points" on the manifold that can be used to identify a distance matrix. We visualize this in Figure~\ref{fig:example_latent_position_localization}.

\begin{figure}[htb!]
    \centering
    \includegraphics[width=0.28\textwidth]{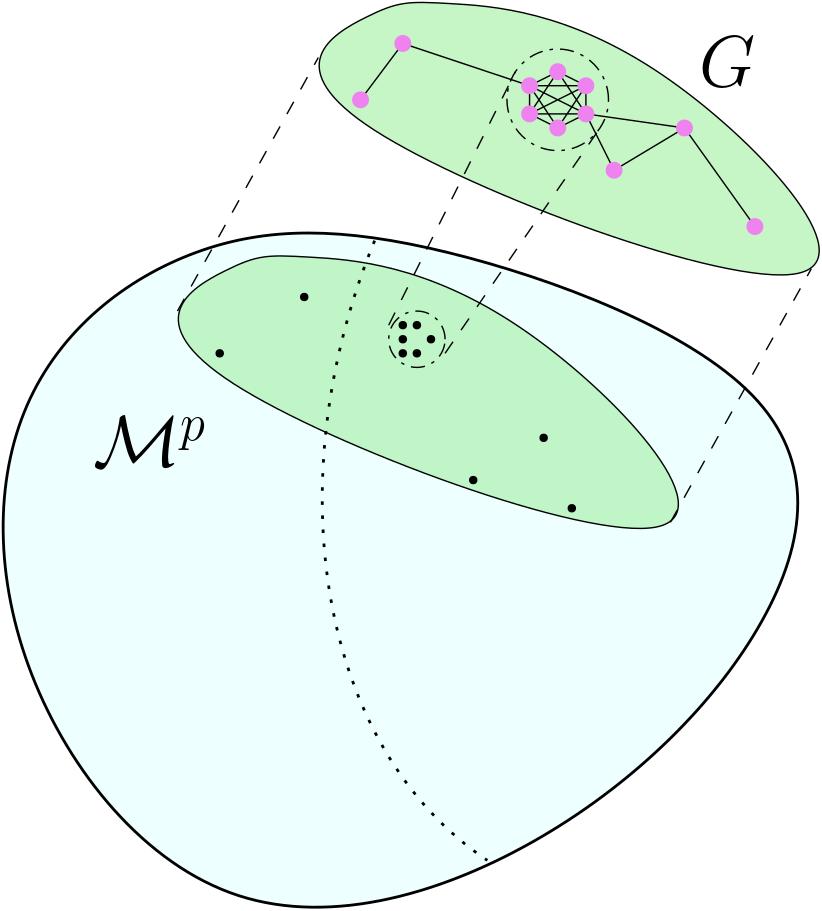}
    \caption{Illustration of the localization of the latent positions within a clique. Nodes are shown in magenta, while the position on the latent space is shown in black.}
    \label{fig:example_latent_position_localization}
\end{figure}

We discuss asymptotics for a fixed number of cliques ($K$) of size $\ell$ where the clique size grows, $\ell \to \infty$. In the appendix (Section~\ref{sec: additional discussion on the latent distance model}), we discuss the rate of clique formation as well as other alternatives to the clique-based estimators, which can be tailored to sparse graphs. Under this model, as $n \to \infty$, the expected number of cliques of size $\ell = \OO(n^{1/(p + 2) - \epsilon})$ for $\epsilon > 0$ grows to infinity as well. Since the formation of cliques requires nodes to be extremely close to one another in the latent space, the primary driver is the intrinsic dimension of the latent space $p$, rather than the curvature. A formal statement of this is found in the appendix (Section~\ref{sec: rates of clique formation}).

An important aspect we consider is the maximum radius of a set of positions conditional on a clique. Unless latent positions are nearly (or exactly) in the same location, large cliques are rare in latent distance models. Lemma~\ref{lem:clique_localization_rate} illustrates a rate of convergence of these latent locations relative to the size of a clique.  Additionally, the nodes which form cliques will also have random effects approaching $0$, Lemma~\ref{lem:randeff_consistency}. 

Under this model, the nodes within a clique have nearly the same latent position. For illustration purposes, we first assume that this holds and later illustrate the rate at which the diameter of the set of latent positions within a clique converges. Let $X, Y \in \{1, 2, \dots, n\}$ denote sets of nodes representing non-overlapping cliques. Then the average probability of connection can be used to identify the latent distance if we can also estimate the average of random effects ($\nu$). Let $|W|$ denote the size of $W \in \{X, Y\}$.

\begin{lemma}\label{lem:clique_localization_rate}
    
    Assume the latent distance model as in equation~\eqref{eq:ld_model} and let $C_\ell$ denote the event that a collection of nodes indexed by $i \in  \{1, \dots, \ell\}$ form a clique of size $\ell$. 

    Let $\mu_d := \E[d(Z_i,Z_j)]$ denote the average distance between any two sampled latent positions $Z_i$ and $Z_j$ which are sampled independently from $F_Z$.  If this latent density satisfies the following condition
    \begin{itemize}
        \item (D1) $F_Z$ admits a continuous density ($f_Z$) 
    \end{itemize}
    then for any $0 < \tilde \mu_d < \mu_d$,  
    \al{
        \max_{i,j} d(Z_i,Z_j) |C_\ell = o_P(\exp(-\tilde \mu_d\ell/p)).
    }
    
\end{lemma}
See the supplementary materials Section~\ref{sec:proof_clique_localization_rate} for the proof of this lemma. The main implication is that it is reasonable to treat the latent positions as a single point when nodes are within a clique. The assumption of $F_Z$ admitting a smooth continuous density ($f_z$) is extremely mild. In fact, one could derive even faster rates if the latent density contained point masses.

An similar result can be derived for the node-level random effects. 
\begin{lemma} \label{lem:rand_eff_clique}
    Assume the latent distance model as in equation~\eqref{eq:ld_model} and let $C_\ell$ denote the event that a collection of nodes indexed by $i \in \{1, \dots, \ell\}$ form a clique of size $\ell$. Suppose that
    \begin{enumerate}
        \item (E1) $F_{\nu}$ admits a continuous density  $(f_\nu)$ on $(-\infty, 0]$ and this density function is positive at $0$, $f_{\nu}(0) > 0$
    \end{enumerate}
    Let $\mu_{\nu} := E[\nu]$. Then for any $\mu_{\nu} < \tilde \mu_{\nu} < 0$
    \begin{equation}
        \min_{i} \nu_i|C_\ell = o_P\left(\exp\left(- |\tilde \mu_{\nu}|\ell\right)\right)
    \end{equation}
\end{lemma}
This theorem states that the nodes we find in a clique tend to have near-zero random effects. This will be useful as the random effects will converge to zero, allowing one to use between-clique connections to estimate distances. We next leverage each of these results to construct an estimator for a set of distances on the underlying manifold.

Due to the localization of the latent positions and random effects within a clique, one can estimate a set of distances using the average connection probability across cliques. Let $\mathpzc{X}, \mathpzc{Y} \subseteq V$ be subsets of vertices that denote nodes that form cliques respectively. We define the average probability of connection across cliques $\mathpzc{X}, \mathpzc{Y}$ by $p_{\mathpzc{X}\mathpzc{Y}}$. Based on the results of Lemma~\ref{lem:clique_localization_rate}, we note that the maximum distance between a set of two points within the same clique is $o_P(\exp(-\tilde \mu_d \ell))$, therefore we let $d_{\mathpzc{X}\mathpzc{Y}}$ denote the average distance in the latent space between latent positions of the cliques. For any $x,y \in \mathpzc{X}, \mathpzc{Y}$, $d_{xy} = d_{\mathpzc{X}\mathpzc{Y}} + o_P(\exp(-\tilde \mu_d \ell)$ and hence the pairwise distances between nodes in a clique are nearly identical. We can therefore relate the average connection probability to the average distance across cliques:
\al{
    p_{\mathpzc{X}\mathpzc{Y}} &:= \frac{1}{|\mathpzc{X}||\mathpzc{Y}|}\sum_{x \in \mathpzc{X}}\sum_{y \in \mathpzc{Y}}p_{xy} \\
    &= \frac{1}{|\mathpzc{X}||\mathpzc{Y}|}\sum_{x \in \mathpzc{X}}\sum_{y \in \mathpzc{Y}}\exp(\nu_x + \nu_y - d_{xy}) \\
    &= \exp(- d_{\mathpzc{X}\mathpzc{Y}})\frac{1}{|\mathpzc{X}||\mathpzc{Y}|}\sum_{x \in \mathpzc{X}}\sum_{y \in \mathpzc{Y}}\exp(\nu_x + \nu_y)\left(1 +  o_P(\exp(-\tilde \mu_d \ell))\right) \\
    d_{\mathpzc{X}\mathpzc{Y}}  &= -\log(p_{\mathpzc{X}\mathpzc{Y}}) + \gamma_\mathpzc{X} + \gamma_\mathpzc{Y} + o_P(\exp(-\tilde \mu_d \ell)) \\
    \text{Where }\gamma_\mathpzc{W} &:= \log\left(\frac{1}{|\mathpzc{W}|}\sum_{w \in \mathpzc{W}}\exp(\nu_w)\right) \quad \text{ for } \mathpzc{W} \in \{\mathpzc{X}, \mathpzc{Y}\}.
}

Since there are $\ell^2$ possible connections between a pair of cliques of size $\ell$, we can construct an asymptotically normal estimator of $p_{\mathpzc{X}\mathpzc{Y}}$ using the sample mean of the connections across cliques. By Lemma~\ref{lem:clique_localization_rate}, the distance between latent positions within a clique is negligible. Our goal is to develop an asymptotically normal estimate of $d_{\mathpzc{X}\mathpzc{Y}}$. We can achieve this by estimating $\gamma_\mathpzc{X}$ and $\gamma_\mathpzc{Y}$ at sufficiently fast rates (i.e., at least $o_P(\frac{1}{\ell}))$ so that these are negligible compared to the estimation of $p_{\mathpzc{X}\mathpzc{Y}}$, for which we can leverage a central limit theorem.

For $\mathpzc{W} \in \{\mathpzc{X},\mathpzc{Y}\}$, in order to estimate $\gamma_{\mathpzc{W}}$, we consider the average connection probability to any node in the network, conditional on a node being in a clique of size $\ell$: 
\al{
    P(A_{ik} = 1|C_{\ell}) &= \exp(\nu_i) \int \exp(\nu_k + d(z_i,z_k)) \, dF_{Z}(z_k|C_{\ell}) \, dF_{\nu}(\nu_k|C_{\ell})  \\
    \frac{P(A_{ik} = 1|C_{\ell})}{P(A_{jk} = 1|C_{\ell})} &= \exp(\nu_i - \nu_j) + o_P(\exp(-\tilde{\mu}_d \ell)).
}
By Lemma~\ref{lem:clique_localization_rate}  $d(Z_i, Z_j) = o_P(\exp(-\tilde{\mu}_d \ell))$. 
Estimation is straightforward by considering the density of connections, namely: $\hat{P}(A_{ik} = 1|C_{\ell}) = \frac{d_i}{n}$.  Since we only use the ratio to compute the estimates of $\exp(\nu_i - \nu_j)$, the total number of nodes $n$ is not necessary.  Given $\ell \ll n$, where $n$ is the number of nodes in the network, this ratio can be estimated easily using the degree ratio. We define $\Delta \nu_i := \nu_i - \nu_\mathpzc{W} $, allowing us to identify $\Delta \nu_i$, where $\nu_\mathpzc{W} = \max_{i \in \mathpzc{W}} \nu_{i}$
\al{
    \gamma_\mathpzc{W} &= \nu_\mathpzc{W} + \log\left(\frac{1}{|\mathpzc{W}|} \sum_{i \in \mathpzc{W}} \exp\left( \Delta \nu_i \right)\right).
}
The remaining question is to estimate $\nu_\mathpzc{W}$. Fortunately, as we have previously discussed in Lemma~\ref{lem:rand_eff_clique}, within a clique, the random effects approaches $0$ at an exponentially fast rate. Hence, an estimator for the random effect can be constructed by setting the largest degree node's random effect to $0$:
\begin{equation}
    \hat \gamma_\mathpzc{W} = \log\left( \frac{1}{| \mathpzc{W}|}\sum_{i \in \mathpzc{W}} \frac{d_i}{\max_{j \in \mathpzc{W}} d_{j}}\right). \label{eq: gamma estimation}
\end{equation}
This in turn can be used to estimate the distances: 
\begin{equation}
    \hat d_{\mathpzc{X}\mathpzc{Y}} = -\log(\hat p_{\mathpzc{X}\mathpzc{Y}}) + \hat \gamma_\mathpzc{X} + \hat \gamma_\mathpzc{Y} \label{eq: distance estimate}
\end{equation}
where $\hat p_{\mathpzc{X}\mathpzc{Y}} = \frac{1}{|\mathpzc{X}||\mathpzc{Y}|} \sum_{x \in \mathpzc{X}} \sum_{y \in \mathpzc{Y}} A_{xy}$. The asymptotic distribution is next described in Theorem~\ref{thm:asymptotic_distance}.

\begin{theorem}\label{thm:asymptotic_distance}
    Let $\hat d_{\mathpzc{X}\mathpzc{Y}}$ the distance estimator as per equation~\eqref{eq: distance estimate}. Suppose that $\mathpzc{X}, \mathpzc{Y}$ are cliques of size $\ell$ and $X \cap Y = \emptyset$. If the following conditions hold:
    \begin{itemize}
        \item (F1) $\ell =  o(\sqrt{n})$
        \item (F2) 
        $\lim_{\ell \to \infty} \frac{\sum_{x \in X} \sum_{y \in Y} (A_{xy} - p_{xy})^2I(|A_{xy} - p_{xy}| > \epsilon \ell^2 \sigma_{\ell})}{\sigma_{\ell}} = 0$ \\
        for all $\epsilon > 0 $
    \end{itemize}
    and we denote 
    \begin{align}
         p_{\mathpzc{X}\mathpzc{Y}} &= \frac{1}{\ell^2} \sum_{x \in X} \sum_{y \in Y} p_{xy}\\
         \sigma_{\ell} &= \frac{1}{\ell^2} \sum_{x \in X} \sum_{y \in Y} p_{xy}(1 - p_{xy}).
    \end{align}
    Then 
    \begin{equation}
        \sqrt{\ell^2 \frac{\sigma_{\ell}}{p_{\mathpzc{X}\mathpzc{Y}}}}\biggl(\hat d_{\mathpzc{X}\mathpzc{Y}} - d_{\mathpzc{X}\mathpzc{Y}}\biggr) \to_d N(0,1)
    \end{equation}
\end{theorem}

The condition (F1) simply states that the size of the cliques should grow at a rate slower than $\sqrt{n}$, the total size of the network, which is a very mild assumption. This is to ensure that the asymptotics associated with the estimation of $\gamma_\mathpzc{W}$ are negligible. Condition (F2) is the standard Lindeberg CLT condition. 
If the node-level probabilities of connection $p_{xy}$ is bounded away from $0$ and $1$, then this will hold (this may not hold if $d_{\mathpzc{X}\mathpzc{Y}} \to -\infty$).  The full proof is found in Section~\ref{sec:proof_asymptotic_distance}.  The main implication is that the asymptotic distribution is driven by the estimator $\hat p_{\mathpzc{X}\mathpzc{Y}}$. Since the random effects will approach $0$, for large cliques $p_{\mathpzc{X}\mathpzc{Y}} \approx \exp(-d_{\mathpzc{X}\mathpzc{Y}})$, thus allowing for simplified expressions for $\sigma_{\ell}$ and $p_{\mathpzc{X}\mathpzc{Y}}$.  

\textbf{Constrained Estimation.} Pairwise estimation of the distances according to equation~\eqref{eq: distance estimate} is one method of estimating a distance matrix, however, this does not restrict the final estimate to be a metric. When cliques are not connected to one-another, this may result in distance estimates which are $-\infty$.

For example, we construct an enumeration of the cliques of size $7$ or greater from the General Relativity co-authorship network of \citet{Leskovec2007GraphEvolution} in Figure~\ref{fig:GR_clique_subgraph}. Any pair of cliques which does not share an edge will inherently be estimated to have infinite distance, which prevents the estimation of curvature. 

\begin{figure}[htb!]
    \centering
    \includegraphics[width = 0.45\textwidth]{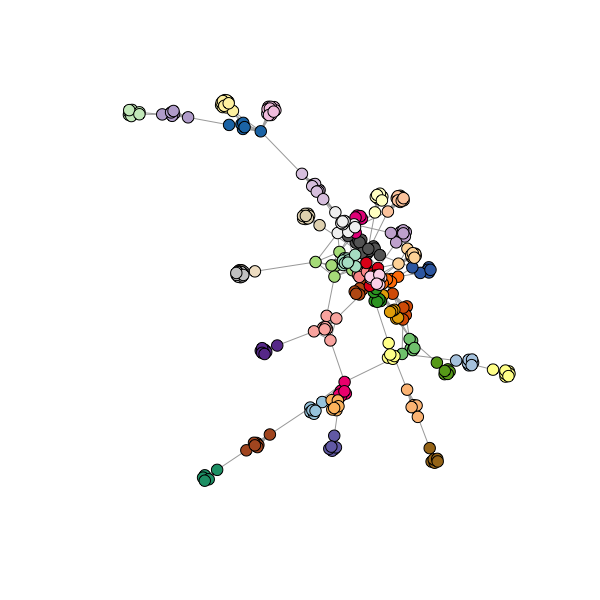}
    \caption{Clique subgraph of co-authorship network in ArXiv General Relativity. Cliques of size 7 and greater are shown by color. }
    \label{fig:GR_clique_subgraph}
\end{figure}

To address this challenge, we posit a similar estimation problem, while respecting the triangle inequality for each triplet. The following estimation problem is posed below.  Let $\mathcal{C}$ denote the set of cliques in an observed graph. Let $\mathpzc{X}\subset \mathcal{C}$ be a set of indices corresponding to a clique. We estimate the random effects within a clique $i \in \mathpzc{X}$ according to $\hat{\nu_i} = \log(\frac{d_i}{\max_{j \in \mathpzc{X}} d_{j}})$. Given a set of fixed effects, we can propose a maximum likelihood optimization problem for the distance matrix $D \in \R^{K \times K}$. As by Lemma~\ref{lem:clique_localization_rate} cliques have approximately a common latent position. From our latent distance model, we define the following likelihood function for the connections within our set of cliques $\mathcal{C}$,  $\mathcal{L}_{\mathcal{C}}$:
\al{
    \mathcal{L}_{\mathcal{C}}(D; \vect{\nu}) := \sum_{\mathpzc{X}, \mathpzc{Y} \in \mathcal{C}, i \in \mathpzc{X}, j \in Y} A_{ij}\left( \nu_i + \nu_j - d_{\mathpzc{X}\mathpzc{Y}}\right) + \sum_{\mathpzc{X}, \mathpzc{Y} \in \mathcal{C}, i \in \mathpzc{X}, j \in Y} (1 - A_{ij})\log\left(1 - \exp\left( \nu_i + \nu_j - d_{\mathpzc{X}\mathpzc{Y}}\right)\right). 
}   
We can now define the maximum likelihood optimization problem, after estimating a set of random effects $\vect{\hat 
\nu}$
\al{
    \vect{\hat D} &= \sup_{D \in \R^{K \times K}} \mathcal{L}_{\mathcal{C}}(D; \vect{\hat \nu}) \\
    \text{s.t. } d_{\mathpzc{X}\mathpzc{Y}} &\geq 0, \quad \text{Diag}(D) = 0, \quad \text{tr}(E_{s}^\top D) \geq 0 \quad  \forall s \in \mathcal{S}
}
where $\mathcal{C}$ is a list of clique indices.  $E_s$ are matrices which define the triangle inequalities and contain mostly $0$'s other than $3$ indices, $i,j,k$ for which $E[i,j] = E[i,k] = 1 = -E[j,k]$. We define $\mathcal{S}$ to be an enumeration of all such matrices $E_s$.  There are $3\binom{K}{3}$ such restrictions in total. The set of feasible distance matrices that satisfy these constraints form a polyhedron of interior dimension $\binom{K}{2}$, therefore therefore the constraints do not reduce the dimensionality of the space of distance matrices, but reduce the volume of the space to a smaller polyhedron. This optimization procedure is a \textit{convex} and in practice, we may use \texttt{CVXR} to solve this system \citep{Fu2020CVXR:Optimization}.  For a greater gain in computational speed, we use the \texttt{MOSEK} solver for the constrained optimization \citep{Andersen2000TheAlgorithm}.  

Though in its current form, the problem is numerically challenging to solve due to the sheer number of constraints. A natural option is to use the Newton method and approximate the likelihood using a second order Taylor expansion and solve this problem successively. Since the Hessian is diagonal, $\frac{\partial^2}{\partial d_{ij} \partial d_{kl}} \mathcal{L}_{\mathcal{C}}(D; \vect{\nu}) = 0 $ if $d_{ij} \not = d_{kl}$, this can be made into a more efficient quadratic program which can be solved faster in \texttt{CVXR}. Let $\tilde g(D;D_0;\vect{\nu})$ be the second order Taylor series approximation to $\mathcal{L}_{\mathcal{C}}$ about a matrix $D_0$. We can successively solve for $\vect{\hat D}_{t +1}$ using the following constrained optimization problem
\begin{align}
     \vect{\hat D}_{t + 1} &= \sup_{D \in \R^{K \times K}}\tilde g(D,\vect{\hat D}_{t}; \vect{\hat \nu}) \label{eq:newtonstep} \\
    \text{s.t.} d_{\mathpzc{X}\mathpzc{Y}} &\geq 0, \quad \text{Diag}(D) = 0, \quad \text{tr}(E_{s}^\top D) \geq 0 \quad  \forall s \in \mathcal{S} \nonumber
\end{align}
   
which can be iteratively computed until $\mathcal{L}_{\mathcal{C}}$ increases less than some threshold.  Each iteration is a linear constrained quadratic program, a common convex problem for which standard software has implemented faster solutions. We implement this in \texttt{CVXR}. For further details, see the supplementary materials in Section~\ref{sec:distance_estimation}.

In practice, running the optimization step for many iterations can be computationally costly. This is particularly problematic in a later application we discuss which involves a subsampling approach to testing for constant curvature (Section~\ref{sec:application_constant_curvature}).  It is well known in maximum likelihood estimation, one only needs to perform one Newton step for asymptotic efficiency \citep{le1956asymptotic}.  As a result, one can start with any consistent estimator of the distance matrix $D$ and apply a single Newton step from equation~\eqref{eq:newtonstep} and obtain the same asymptotic distribution, and therefore in practice, only one step is needed. In the appendix~\ref{sec: Alternative Estimators of the Distance Matrix} we include an alternative approach for estimating distance matrices under an alternative framework, without the need for cliques.  In practice, we can use the method in equation~\eqref{eq: distance estimate} to form an initial estimate, then modify the entries so that it is a metric. One example of an algorithm which can be used for this purpose is the Floyd-Warshall Algorithm. This algorithm modifies as few entries as possible, solving a problem also known as sparse metric repair \citep{gilbert2017if}. For full details, see the appendix Section~\ref{sec:distance_estimation}. We now present the entire procedure for estimating latent curvature from a network in Algorithm~\ref{alg:kappa_estimation}.

\begin{algorithm}
\begin{algorithmic}[1]
\Require $G$, Steps $T$, Triangle constant $C_{\triangle}$, e.g. $C_{\triangle} = 1.4$. Minimum clique size $\ell$. 
\State Find a set of $\mathpzc{C}$ where $K = |\mathpzc{C}|$ and $C \in \mathpzc{C}$ where $C \subset G$ such that $|C| \geq \ell$. 
\State Estimate the initial distance matrix estimate as per equation~\eqref{eq: distance estimate} and denote this $\vect{\hat D}^{(-1)}$
\State Apply the sparse metric repair algorithm (i.e. the Floyd-Warshall algorithm to modify $\vect{\hat D}^{(-1)}$ so that it is a metric and denote this $\vect{\hat D}^{(0)}$. 
\State Estimate the distance matrix $\vect{\hat D}$ by iterating equation~\eqref{eq:newtonstep} $T$ times.
\State Rank the surrogate midpoint sets and identify the best midpoint set $\hat y, \hat m, \hat z$ by applying equation~\eqref{eq:surrogate_midpoint_problem}. 
\State Select values of $x, \mathbf{X}$ such that equation~\eqref{eq:triangle_tuning} is satisfied. 
\State Estimate $\hat \kappa$ by taking the median over the set $\vect{X}$ as per equation~\eqref{eq:median_kappa}. \\
\Return $\hat \kappa $
\end{algorithmic}
\caption{Algorithm for estimation $\hat \kappa$.}\label{alg:kappa_estimation}
\end{algorithm}

In the appendix (Section~\ref{sec: Alternative Estimators of the Distance Matrix}), we also illustrate an approach to estimating distance matrices in a similar setting in the absence of cliques. We now continue with an empirical study of our estimators. 

\section{Simulation Study}\label{sec:simulations_consistency}

We construct a number of simulations to illustrate the performance of our curvature estimator under the full model complexity. This involves first sampling the parameters of a latent position cluster model to draw locations and variances independently, then sampling positions and random effects from the random latent position cluster model. This is to illustrate how midpoints may naturally align from random draws of the centers of the cluster model. We simulate from the following model:
\al{
    \mu_k &\sim F_{\mu}^{\M^p(\kappa)}(\vect{O}, R) \quad &k \in \{1, \dots, \lceil\sqrt{\rho} 50 \rceil \} \\
    \sigma^2_k &\sim F_{\sigma^2} \quad &\vect{\pi} \sim Dirichlet(\vect{2}) \\
    \rho &\in \{\frac{1}{\sqrt{2}}, 1, 2\} \quad &\kappa \in \{-2, -1, -0.5, 0, 0.5, 1\}
}
where $F_{\mu}^{\M^p(\kappa)}(\vect{O}, R)$ denotes the prior distribution on the latent positions, which is a uniform ball with radius $R$ for $\kappa > 0$, and two concentric balls, one with radius $R$ and the other with $R/2$, with equal probability for $\kappa \leq 0$. This setup facilitates forming midpoints since the volume of a ball grows exponentially as $\kappa$ decreases. In all simulations, we set $R = 2.5$. This process determines a random latent position cluster model (LPCM) similar to \cite{Handcock2007Model-basedNetworks}. The locations of $\mu_k$ as well as the relative sizes of $\sigma^2_k$ determine where cliques are most likely to form in the latent space. The parameter $\rho$ refers to the scale of the network simulations, allowing for the number of centers to grow with $\rho$, $\vect{\pi}$ refers to the vector of cluster probabilities in the mixture model, $\vect{\mu}$ and $\vect{\sigma^2}$ refer to the cluster mean scale parameters, and $\kappa$ the true curvature of the space. The parameters $\vect{\mu}, \vect{\sigma^2}, \vect{\pi}$ determine a latent position cluster model, where the positions, $Z_i \sim F_Z := F_Z(\vect{\mu}, \vect{\sigma^2}, \vect{\pi}; \kappa)$. The mixture components correspond to the heat kernels in spherical space (Von-Mises Fischer distribution), Gaussian distribution in Euclidean space, and the wrapped normal distribution in hyperbolic space \citep{Nagano2019ALearning}. In all cases, we let $p = 3$. The randomness in the latent position cluster model incorporates the fact that we are not assigning good midpoints exactly but finding them in the data each time we simulate a matrix. We let $F_{\sigma^2}$ denote the gamma distribution with shape parameter $1/16$.

We then sample the networks according to the draw of the latent position cluster model:
\al{
    n &= 5000\rho, \quad Z_i \sim F_{Z}, \quad \nu_i \sim F_{\nu, \rho} \\
    p_{ij} &= \exp(\nu_i + \nu_j - d(Z_i, Z_j)) \\
     A_{ij} &\sim \text{Bern}(p_{ij})
}
where $F_Z = F_Z(\vect{\pi}, \vect{\sigma^2}, \vect{\mu})$ is a particular draw of the random latent position cluster model. The random effects distribution $F_{\nu, \rho}$ is a trimmed normal distribution with mean $-3\rho$ and standard deviation $3\rho$ trimmed at $0$ and $\log(2/\sqrt{n})$ so that $\nu_i$ remain non-negative and to prevent isolated nodes. We let the minimum clique size used in our estimator be $\ell = (8 + 4\log(\rho))$, which tends to generate $30-40$ cliques. We repeat this 500 times for each scale and curvature setting. In practice, we find the cliques using the \texttt{maximal\_cliques} function in the \texttt{igraph} R package \citep{csardi2006igraph}. See Section~\ref{sec:graph_statistics} for additional graph statistics summaries over the simulations. We provide the values of the tuning parameter $C_{\triangle}$ in the appendix in Section~\ref{sec: values of C_delta}.

\subsection{Consistency of Simulation Curvature Estimates} \label{sec: Consistency of Simulation Curvature Estimates}
We now explore the consistency of our curvature estimate as a function of the curvature of the latent space. In each of the simulations, we limit the number of cliques used in the estimator to 35 (50 for $\kappa \leq -1$) for computational convenience for numerous simulations; though in real data applications, this number can be larger.

We plot the results in Figure~\ref{fig:Estimator_Latent_GMM_consistency}. We see that, as the clique size increases, there is a reduction in bias and variance. We note observe that as the curvature more negative, the estimator has greater variability. This is for two reasons. Firstly, as we saw in Figure~\ref{fig:theoretical_variance}, the variance of the estimate is simply larger when $\kappa$ is negative in nearly all regions of the space, a small variation in the length of a triangle median corresponds to a large variation in the corresponding curvature. Secondly, due to the vastness of the negatively curved spaces, we tend to have poorer quality midpoints form as well as fewer reference points $x$ which form nearly equilateral triangles.

\begin{figure}[htb!]
    \centering
    \includegraphics[width=0.45\textwidth]{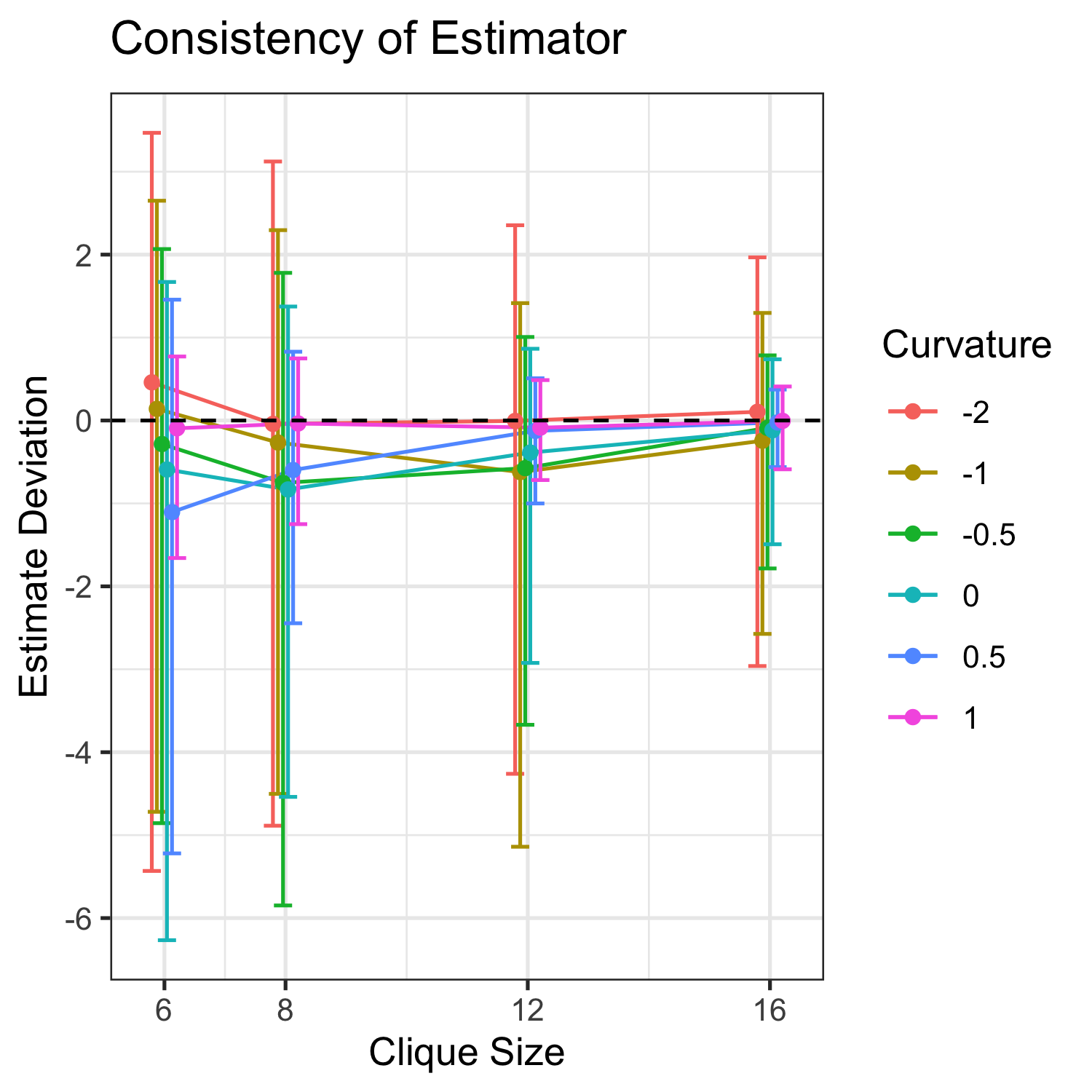}
    \caption{Consistency of Curvature Estimator. Upper error bars indicate the 0.95 quantile of simulations and lower indicate the 0.05 quantile. Central dots indicate the mean after trimming outliers larger than $\pm 100$ (0.107\% of the observations).}
    \label{fig:Estimator_Latent_GMM_consistency}
\end{figure}

\section{Applications: Testing and Detecting Differences in Curvature} 
\label{sec:application_constant_curvature}

We now return to the problem of testing whether the latent space is one of constant curvature, i.e. one of our canonical manifolds.  Formulated as a statistical test, we write this null hypothesis as follows,
$$ H_0: \mathcal{M} = \mathcal{M}(\kappa). $$
This test could provide a model diagnostic (e.g. testing whether a latent variable model that assumes a single constant manifold is appropriate) or provide meaningful insights into heterogeneity in the structure of the graph.

Our test of constant curvature will rely on the upper and lower bounds of the curvature, as developed in Section~\ref{sec: Estimating Curvature}. We do this by leveraging a set of midpoint complexes (i.e. $(y^{(j)}, m^{(j)}, z^{(j)})$), and computing a set of corresponding upper and lower bounds on the curvature at each set. 

Our test will require the sampling distribution of the distance matrix estimator $\vect{\hat D}$. We  consider a method of subsampling from the distribution of the cliques in order to approximate the sampling distribution of the distance matrix estimator. This is based on the subsampling approach of \citet{Politis1994LargeAssumptions}, which can be used to approximate the distribution of a random variable using subsampling under conditions weaker than the bootstrap, similar to the strategy in \citet{lubold2023identifying}.  This is highlighted in Algorithm~\ref{alg:subsample}.  For simplicity, we illustrate the algorithm where we use subsamples of size $\ell - 1$ of each of the cliques. 

Let $\mathcal{I} = \{\textbf{I}_{k}\}_{k = 1}^K$ denote a set of indices corresponding to the cliques, where $\textbf{I}_{k} \cap \textbf{I}_{k'} = \emptyset$ when $k \not = k'$.  Let $\hat \nu_{\textbf{I}_k}$ denote the set of estimated random effects corresponding to the clique $\textbf{I}_k$. 
\begin{algorithm}[htb!]

\caption{Sub-sampling Procedure to Approximate the Sampling Distribution of $ \vect{\hat D}$}\label{alg:subsample}
\begin{algorithmic}[1]
\Require $G$, $B \geq 1$ and $\{\textbf{I}_{k}\}_{k = 1}^K$
\For{$b \in \{1,2,\dots, B\}$}
    \For{$k \in \{1,2,\dots, K\}$} 
        Sample $|\vect{I}_k| - 1$ nodes without replacement from $\vect{I}_k$ and denote this set $\vect{I}^{(b)}_k$
    \EndFor
    \State Denote the set $\mathcal{I}^{(b)} = \{\textbf{I}^{(b)}_{k}\}_{k = 1}^K$
    \State Let $\mathbb{\nu}^{(b)} = \mathbb{\nu}_{\mathcal{I}^{(b)}}$ denote the corresponding random effects estimates. 
    \State Let $\vect{\hat D}^{(b)}_{0} = f_0(A, \mathcal{I}^{(b)}, \mathbb{\nu}^{(b)})$  Initial Estimate \label{alg_step: initial_estimate}
    \State Let $\vect{\hat D}^{(b)}_{1} = f_1(\vect{\hat D}^{(b)}_{0};A, \mathcal{I}^{(b)}, \mathbb{\nu}^{(b)})$ One-step Estimate \label{alg_step: one_step_estimate}
\EndFor \\
\Return $\{\vect{\hat D}^{(b)}\}_{b = 1}^B = \{\vect{\hat D}^{(b)}_{1}\}_{b = 1}^B $
\end{algorithmic}
\end{algorithm}
In Algorithm~\ref{alg:subsample} $f_0$ denotes the distance estimation by each of the distance separately as per equation~\ref{eq: distance estimate} and subsequent sparse modification of the distances using the Floyd-Warshall Algorithm so that $\vect{\hat D}_0$ is a metric. The subsequent step $f_1$ refers to applying a one-step estimation procedure of equation~\eqref{eq:newtonstep}.

We now operationalize our test of constant curvature. Consider a set of midpoint sets $y'^{(j)}, j \in \{1,2,\dots, J\}$ with corresponding reference points $\vect{X}^{(j)}$.  We test for whether the curvature is constant across the latent space across each of these midpoint sets
$$ H_0: \kappa_j = \kappa_{j'} \text{ for all } j, j' \in \{1,2,\dots, J\}$$
using Algorithm~\ref{alg:cc_test}. In order to choose the corresponding locations we utilize the selection of midpoint sets for $J$ collections of midpoints as illustrated in \ref{sec:practical_reference}.

\begin{algorithm}
\footnotesize
\caption{Constant Curvature Test}\label{alg:cc_test}
\begin{algorithmic}[1]
\Require $\{\vect{\hat D}^{(b)}\}_{b = 1}^B$ and $\{y^{(j)},z^{(j)},m'^{(j)},\vect{X}^{(j)}\}_{j = 1}^J$
\For{$b \in \{1,2,\dots, B\}$}
    \For{$j \in \{1,2,\dots, J\}$}
        \For{$x \in \{\vect{X}^{(j)}\}$}
            \State $\hat \kappa_{u,x}^{(b,j)} = \kappa_u(\vect{\hat D}^{(b)}; y^{(j)},z^{(j)},m'^{(j)}, x)$
            \State $\hat \kappa_{l,x}^{(b,j)} = \kappa_l(\vect{\hat D}^{(b)}; y^{(j)},z^{(j)},m'^{(j)}, x)$
        \EndFor 
        \State Compute $\hat \kappa_u^{(b,j)} = \text{median}_{x \in \vect{X}^{(j)}} \hat \kappa_{u,x}^{(b,j)}$ 
        \State Compute $\hat \kappa_l^{(b,j)} = \text{median}_{x \in \vect{X}^{(j)}} \hat \kappa_{u,x}^{(b,j)}$
    \EndFor 
    \State Compute $\hat \kappa_u^{(b)} = \text{min}_{j \in \{1,2,\dots, J\}} \hat \kappa_u^{(b,j)}$ 
    \State Compute $\hat \kappa_l^{(b)} = \text{max}_{j \in \{1,2,\dots, J\}} \hat \kappa_l^{(b,j)}$
\EndFor 
\State Let $\hat \kappa_{u,(m)}$ be the $m^{th}$ order statistic of $\{\hat \kappa_u^{(b)}\}_{b = 1}^B$
\State Let $\hat \kappa_{l,(m)}$ be the $m^{th}$ order statistic of $\{\hat \kappa_l^{(b)}\}_{b = 1}^B$
\State Let $m = \min\{m : \hat \kappa_{u,(m)} \geq \hat \kappa_{l,(B - m)}\}$ \\
\Return $p$-value: $\min\{2m/B, 1\}$
\end{algorithmic}
\end{algorithm}

The constant curvature test involves analyzing a sampling distribution of distance matrices ${\vect{\hat D}^{(b)}}_{b = 1}^B$ derived from Algorithm~\ref{alg:subsample}. It also uses collections of surrogate midpoint sets $\{y^{(j)}, z^{(j)}, m'^{(j)}\}_{j = 1}^J$ and the corresponding favorable triangle reference points $\{\vect{X}^{(j)}\}_{j = 1}^J$. These sets help estimate the upper and lower bounds of curvature. Practically, each set $\vect{X}^{(j)}$ is chosen using the method described in equation~\eqref{eq:triangle_tuning}. To reduce the variance of these upper and lower bound estimates across surrogate midpoint sets, we use the median across the reference points $\vect{X}^{(j)}$. This then gives us a collection of upper and lower bound estimates of the curvature across regions $\{\hat \kappa^{(j)}_l, \hat \kappa^{(j)}_u\}$. If the minimum of the upper bounds is less than the maximum of the lower bounds then this corresponds to a distance matrix that cannot be represented by a single curvature.  This process is repeated across the set $b \in {1,2,\dots, B}$ and the proportion of times these quantities cross can then be interpreted as a $p$-value for the constant curvature test.

In this test there exists an inherent trade-off: achieving tightly aligned midpoint sets provides better upper and lower bounds on curvature estimates, yet it is also essential to sample from regions of the latent space sufficiently distant to potentially reveal differences in curvature. The analysis primarily focuses on optimizing the choice of midpoint sets on the latent space, as the sampling adequacy in regions with varying curvature typically lies outside of the analyst's direct control.

\subsection{Simulations: Type 1 Error Control} \label{sec:simulations_test_type1_error}

Under the same simulation setup as in Section~\ref{sec:simulations_consistency} we can illustrate the coverage of the constant curvature test as a function of clique size.  For computational convenience, we restrict these to a maximum of size $12$.  We see in  Figure~\ref{fig:Estimator_Type_I_error} that this test tends to be overly conservative in small sample sizes, but returns to nearly nominal coverage when cliques are larger. This is due to the fact that a poorly aligning midpoints lead to more conservative bounds of $\kappa_l, \kappa_u$, however, better aligning midpoints are found in larger networks with more cliques.
\begin{figure}[htb!] 
    \centering
    \includegraphics[width=0.45\textwidth]{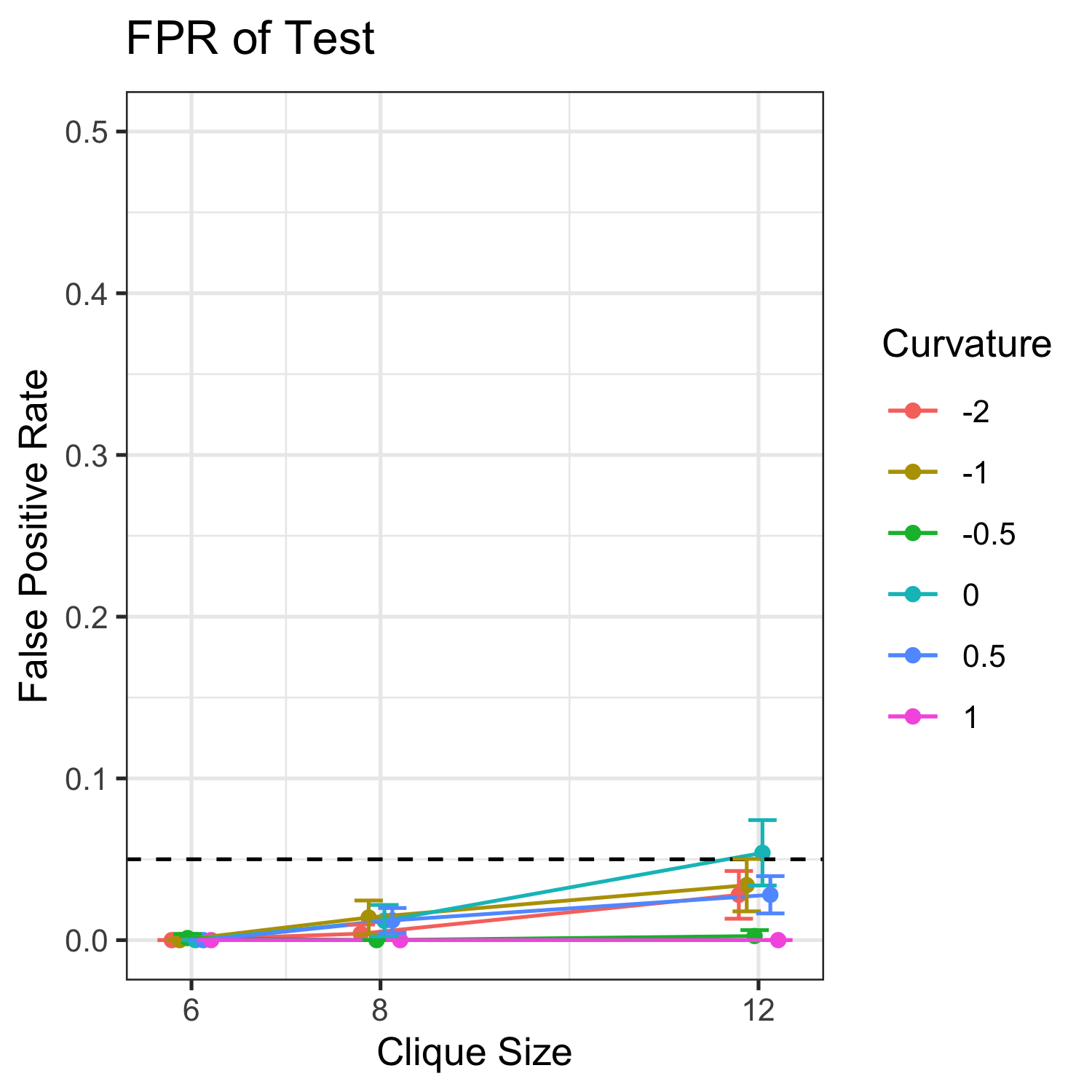}
    \caption{False positive rate for constant curvature test.} 
    \label{fig:Estimator_Type_I_error}
\end{figure}


\subsection{Multiplex Networks} \label{sec: multiplex network constant curvature tests}

In this next simulation, we consider a model of multiplex (or multiview) networks. Several methods exist for modelling multiplex networks via extensions of the latent distance model, however, most assume the same geometry latent space among views  \citep{Salter-Townshend2017LatentData, Macdonald2022latent}.  For example, \citet{Salter-Townshend2017LatentData} model the multiple relationships between individuals in the \citet{Banerjee2013TheMicrofinance} diffusion of microfinance dataset using Euclidean spaces.  We illustrate a simulated example where this is not the case, and how our method can be used to detect this. 

We first construct a latent position model for which multiple views are drawn from a common set of latent positions, however, these positions are common coordinates of spheres of curvature $(\kappa_1,\kappa_2) = (0.5,1.5)$ respectively. Additional details for the simulation are identical to the consistency simulations in Section~\ref{sec:simulations_consistency}. 

In this example, we simulate a multiplex network for which latent spherical positions are the same for this network set, however, they are embedded in two spheres of different radii and thus different curvatures. We simulate $200$ draws of these networks and test the curvature difference in curvature using Algorithm~\ref{alg:cc_test}. We plot the power of the test in Figure~\ref{fig:power_multiview}. In each view we compute the optimal surrogate midpoint set $y^{(j)}, z^{(j)}, m^{(j)}$ from each view's distance matrix $\vect{\hat D}^{(1)}$ or $\vect{\hat D}^{(2)}$ using equation~\ref{eq:surrogate_midpoint_problem} and subsample each distance matrix accordingly. In Figure~\ref{fig:power_multiview} we observe that the power of the test grows with the clique size. 

\begin{figure}[htb!]
    \centering
    \includegraphics[width = 0.45\textwidth]{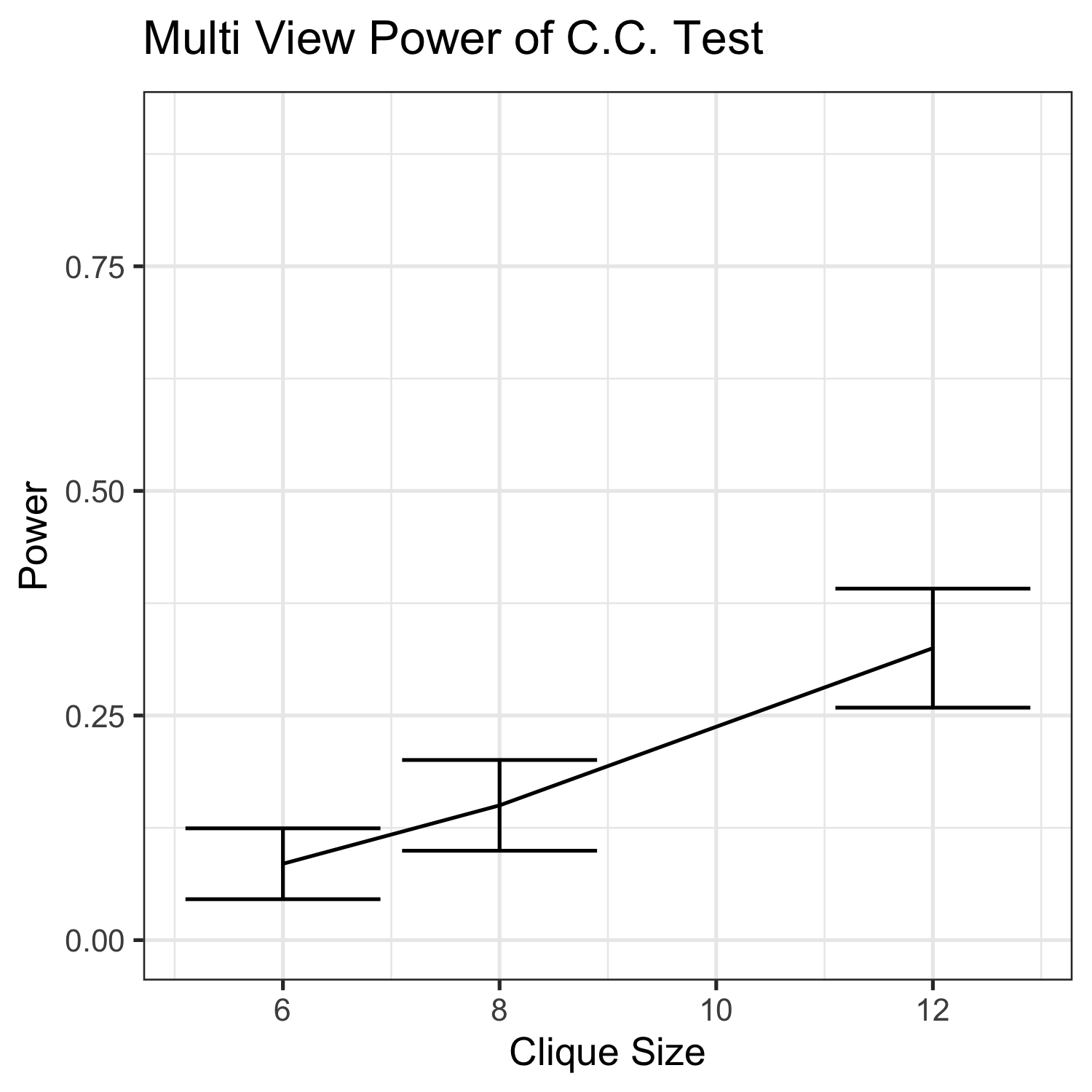}
    \caption{Constant curvature test power applied to the multiview network example.}
    \label{fig:power_multiview}
\end{figure}

\subsection{Noncanonical Manifolds} \label{sec: Noncanonical Manifolds}
We now demonstrate how our method can detect non-constant curvature in the latent space. Given that many prevalent latent distance models inherently assume constant curvature, it becomes crucial to confirm if this assumption aligns with the actual data observed.

We construct a latent space consisting of two adjacent spheres. For any two points in the same sphere the distance is straightforward to compute. For any two points $(x,y)$ in opposite spheres, the distance can be computed using the distance to the origin $(1,0,0)$ in each of the spheres. Since any geodesic must pass through the connecting point, i.e. the origin we can compute these distances as follows
\al{
 d_{\M}(x,y) &= d_1(x,o) + d_2(o,y).
}
This manifold was chosen as distances were straightforward to compute but come from a space without constant curvature. These spheres have curvature $\kappa_1 = 1, \kappa_2 = 1.5$ respectively. 

The latent cluster locations are sampled according to uniform distributions centered on opposite poles. We again simulate $200$ draws from the latent position cluster model and test for constant curvature by finding the best three non-overlapping sets minimizing equation~\eqref{eq:surrogate_midpoint_problem}. We plot the corresponding power in Figure~\ref{fig:power_ad_sphere} which once again,  increases along with the clique size.  

\begin{figure}[htb!]
    \centering
    \includegraphics[width = 0.45\textwidth]{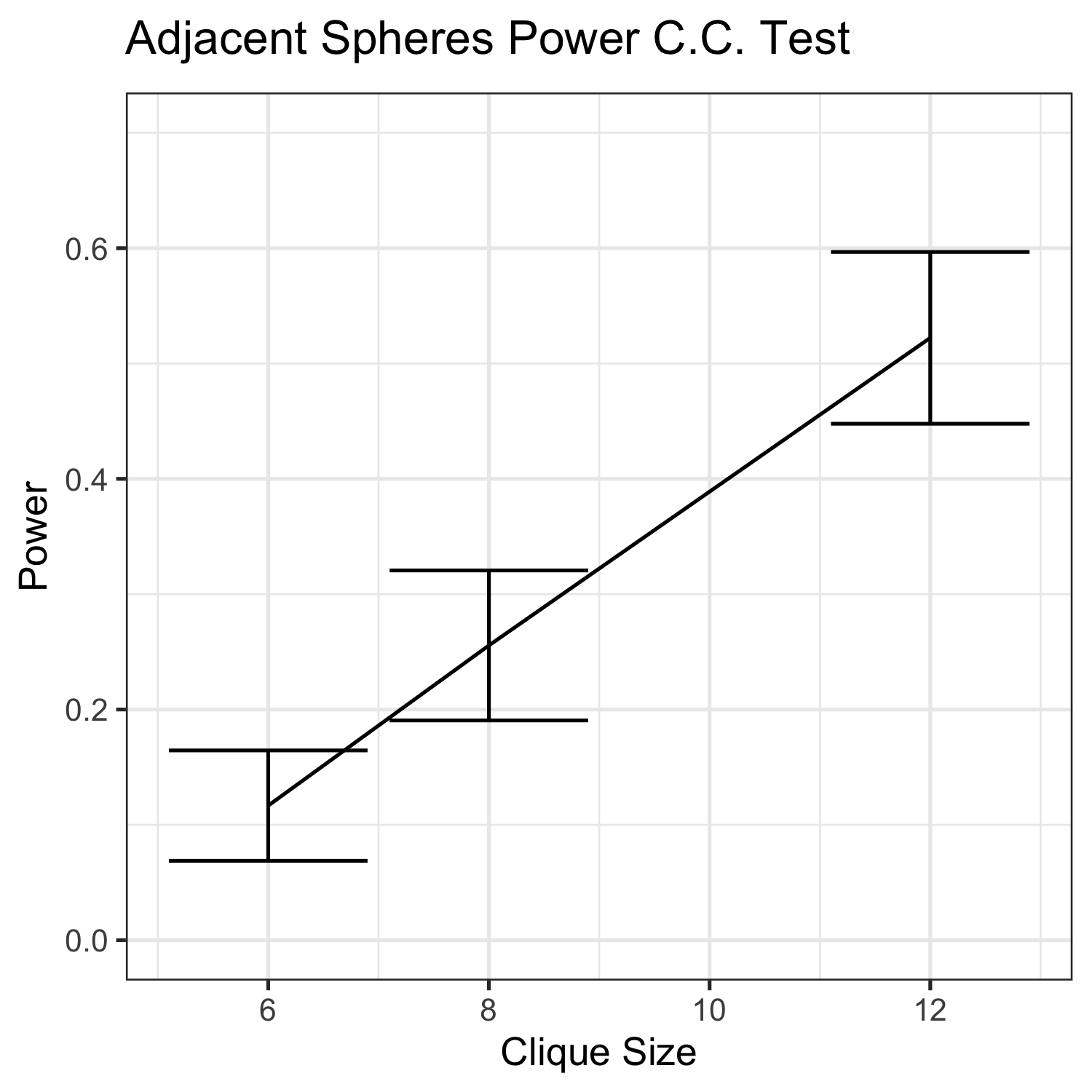}
    \caption{Power of the constant curvature test applied to the adjacent spheres simulation, an example of a non-canonical manifold, which does not have constant curvature.} 
    \label{fig:power_ad_sphere}
\end{figure}

\subsection{Data Analysis: Geometry of Coauthorship} \label{sec: data analysis geometry of coauthorship}
We apply our method to a collection of co-authorship networks in physics introduced in \citet{Leskovec2007GraphEvolution} and available on the \texttt{Stanford Network Analysis Project (SNAP)} repository. These consist of citation networks from High Energy Particle Physics, General Relativity, Astrophysics, and Condensed Matter Physics, with sizes of each of the networks seen in Table~\ref{tab:physics_co-authorship_sizes}. These networks consist of authors as nodes where an edge exists whether any pair of authors has co-authored a paper on ArXiv in any of the specified subject areas between 1993 and 2003. When these data were collected, these were the top 5 most common subject areas in Physics. 

In previous machine learning applications, hyperbolic network embeddings have been successful in tasks such as link-prediction or node in citation networks \citep{Nickel2017PoincareRepresentations, Chami2019HyperbolicNetworks, Chamberlain2017NeuralSpaceb}. This is due to the fact that the hierarchical structure (tree-like) structure generally can be more easily embedded in a negatively curved space.  We wish to answer the question ``under a latent distance model, could the data be generated from a space of constant curvature?"

For each of these networks, we construct an estimate of the distance matrix with random effects, followed by an estimate of the curvature.  For each of these, we use a minimum clique size seen in Table~\ref{tab:physics_application_estimator_params}. We lastly apply our test to see if the difference in curvature is present across the networks.  We estimate the curvature, at the best midpoint set for each of the networks, along with the following p-values for tests of whether a network has constant curvature. 

\begin{table}[htb!] \label{tab:physics_network_summaries}
\centering
 
\begin{tabular}{|l|l|l|}
\hline
Physics Sub-field                      & Number of Nodes (n) & Number of Edges ($|E|$) \\ \hline
Astrophysics                          & 18771               & 396160              \\ \hline
Condensed Matter Physics              & 23133               & 186936              \\ \hline
General Relativity                    & 5241                & 28980               \\ \hline
High Energy Particle Physics          & 12006               & 237010              \\ \hline
High Energy Particle Physics (Theory) & 9875                & 51971               \\ \hline
\end{tabular}
\caption{Physics Co-authorship network sizes. }
\label{tab:physics_co-authorship_sizes}
\end{table}

\begin{table}[htb!] 
\centering
\begin{tabular}{|l|l|l|l|}
\hline
Physics Sub-field                      & \begin{tabular}[c]{@{}l@{}}Min Clique \\ Size ($\ell$)\end{tabular} & \begin{tabular}[c]{@{}l@{}}Number of Cliques \\ of Size $\geq \ell$ (K)\end{tabular} & \begin{tabular}[c]{@{}l@{}}Largest \\ Clique Size\end{tabular} \\ \hline
Astrophysics                          & 19                                                               & 57                                                                           & 57                                                             \\ \hline
Condensed Matter Physics              & 12                                                               & 52                                                                           & 26                                                             \\ \hline
General Relativity                    & 7                                                                & 44                                                                           & 44                                                             \\ \hline
High Energy Particle Physics          & 14                                                               & 42                                                                           & 239                                                            \\ \hline
High Energy Particle Physics (Theory) & 7                                                                & 42                                                                           & 32                                                             \\ \hline
\end{tabular}
\caption{Clique size and number of cliques used to estimate distance matrix. }
\label{tab:physics_application_estimator_params}
\end{table}

\begin{table}[htb!] 
\centering
\begin{tabular}{|l|l|l|}
\hline
Physics Sub-field                      & Curvature Estimate & \begin{tabular}[c]{@{}l@{}}Constant Curve\\ Test p-value\end{tabular} \\ \hline
Astrophysics                          & -0.01378               & 0.030                                                         \\ \hline
Condensed Matter Physics              & 0.107              & 1.000                                                                \\ \hline
General Relativity                    & 0.0989            & 1.000                                                                 \\ \hline
High Energy Particle Physics          & -0.986             & 1.000                                                                 \\ \hline
High Energy Particle Physics (Theory) & 0.1674              & 0.240                                                                 \\ \hline
\end{tabular}
\caption{Curvature estimates from best midpoints and p-values for constant curvature test $J = 3$. }
\label{tab:physics_curve_pvals}
\end{table}

\begin{figure}[htb!] 
    \centering
    \begin{subfigure}[b]{0.25\textwidth}
        \centering
        \includegraphics[width=\textwidth]{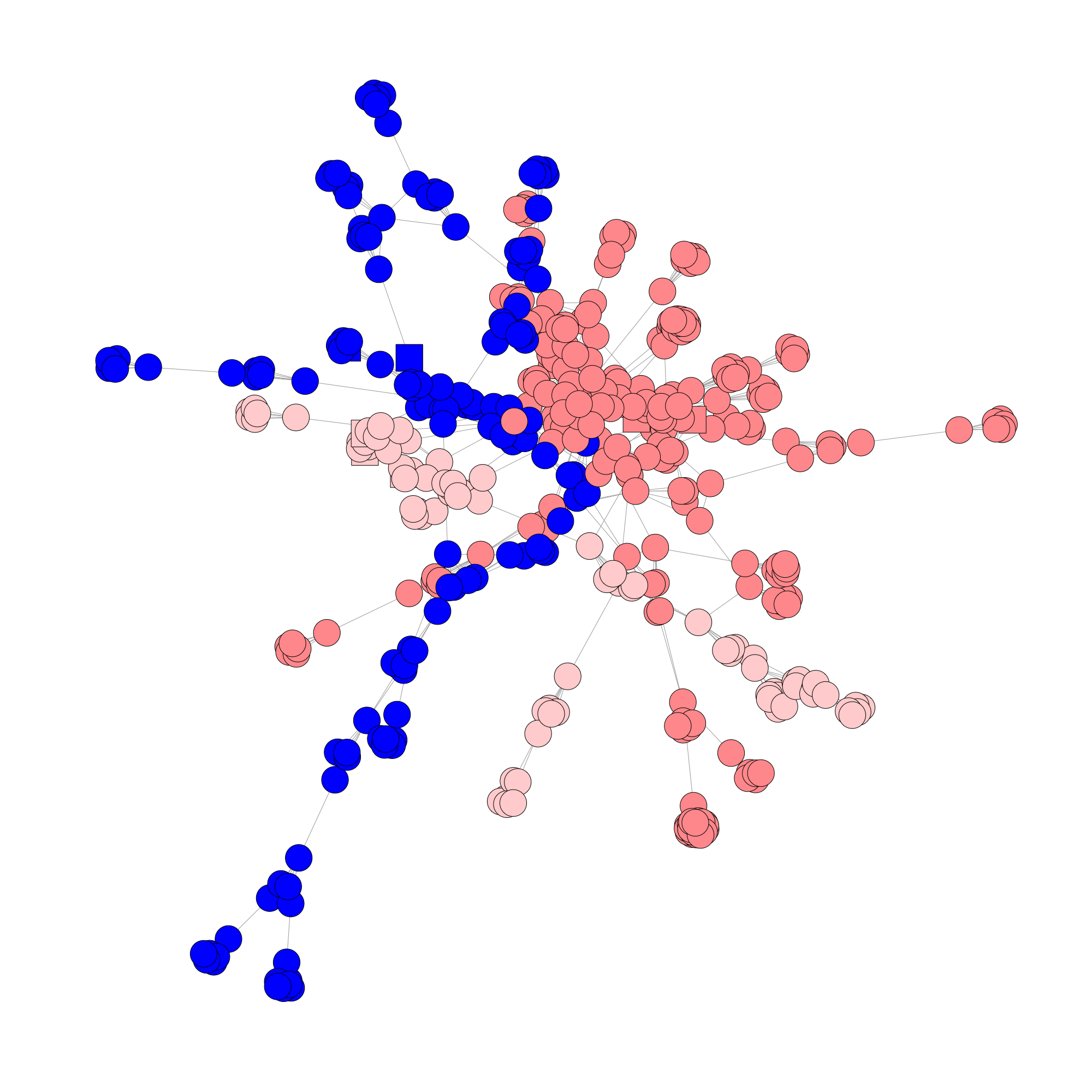}
        \caption{High Energy Particle Physics (Theory) clique curvature labels. }
    \end{subfigure}%
    ~ 
    \begin{subfigure}[b]{0.25\textwidth}
        \centering
        \includegraphics[width=\textwidth]{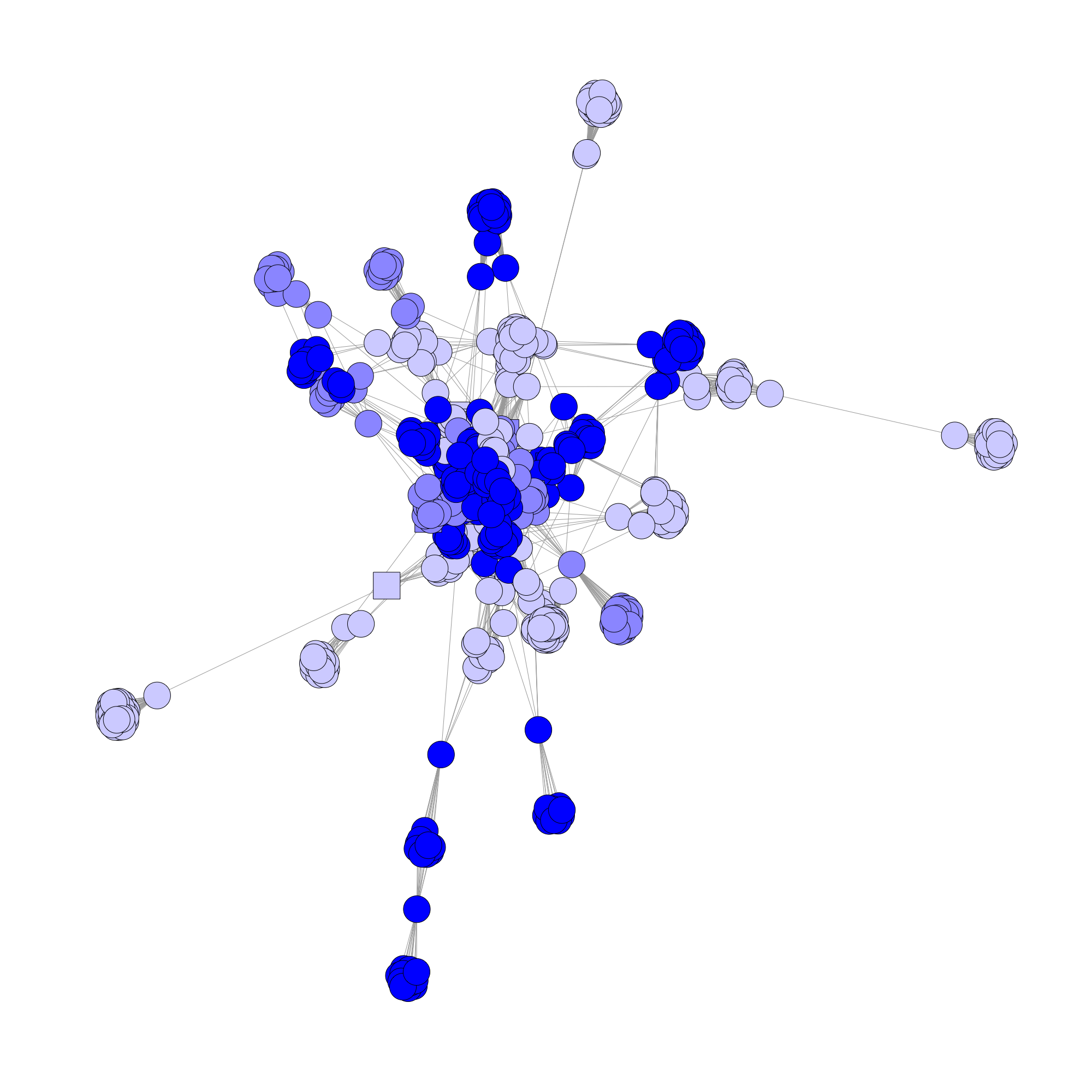}
        \caption{High Energy Particle Physics clique curvature labels. }
        \label{fig:curvature_partition:subfig:HEP_applied}
    \end{subfigure}%
    ~ 
    \begin{subfigure}[b]{0.25\textwidth}
        \centering
        \includegraphics[width=\textwidth]{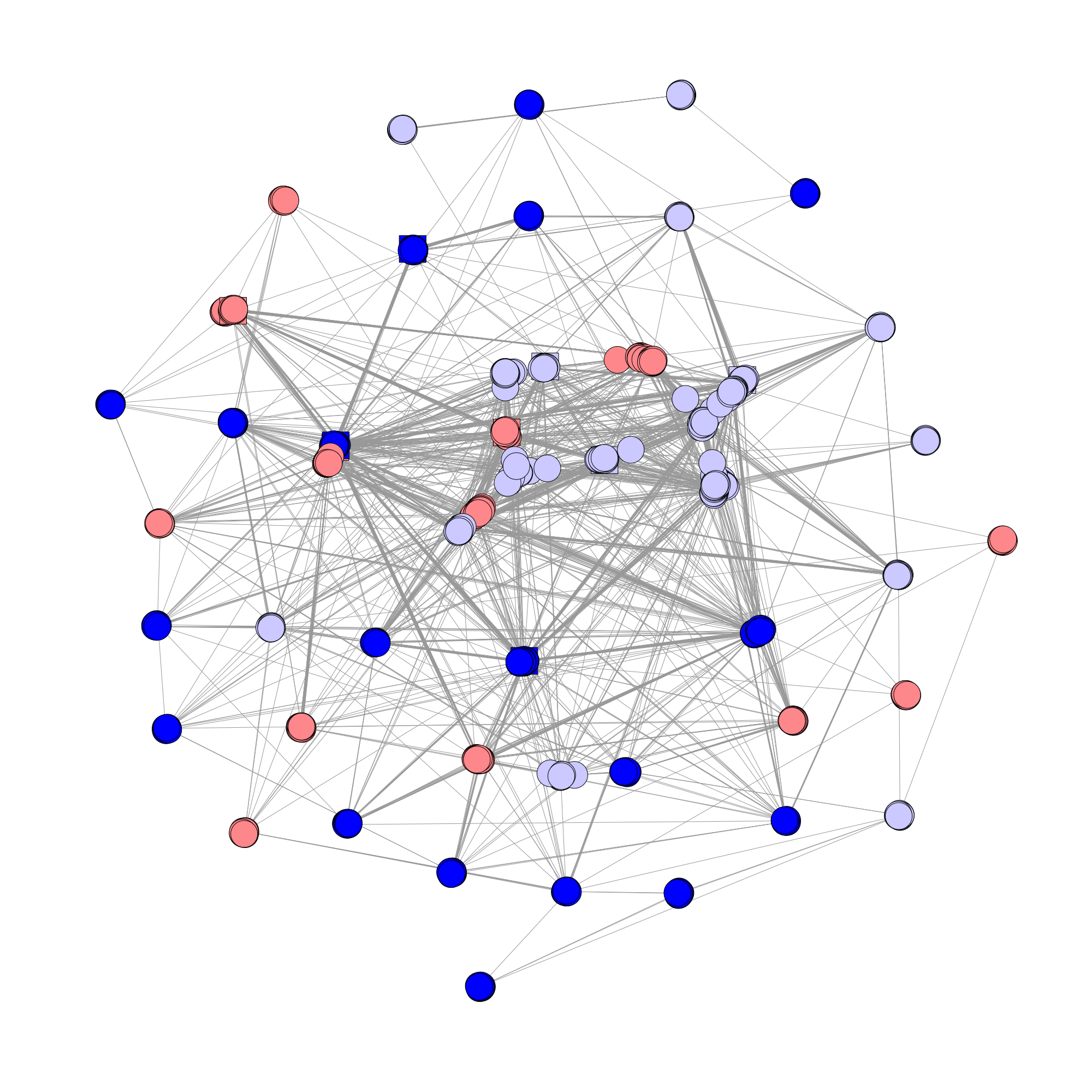}
        \caption{Astrophysics clique curvature labels. }
    \end{subfigure}
    \caption{Cliques colored to nearest labeled curvature value.  Blue, negative and red positive, with $p$ values of constant curvature decreasing from left to right. }
    \label{fig:curvature_partition}
\end{figure}

One concern a reader may have is whether good surrogate midpoints form in real data.  In practice, we find good midpoints can be found in our real datasets. For example, the Fr\'echet mean objective function~\eqref{eq:midpoint_obj} has a value at the true midpoint of $0.5$. In the 5 physics co-authorship networks this was found to be $0.5000003$, $0.5039$, $0.50049$, $0.50012$, $0.50043$, suggesting good surrogate midpoints in practice. 

We see that Astrophysics and General Relativity are estimated to have a large negative curvature, which however, may not be constant in the case of the Astrophysics citation network. In the Astrophysics network, at the best 3 surrogate midpoint sets, we estimate curvature to be $(-0.01378, -\infty, 0.306)$.  The estimate of $-\infty$ comes from the fact there is a minimum distance of $d_{xm}$ given the other 3, $d_{xy},d_{yz},d_{xz}$, and since the midpoints are not exact, sometimes this length can be too short to embed in any of the hyperbolic spaces.  If the estimated distance is below this value we call the estimate $-\infty$, and similarly $\infty$ if it is too large.  This highlights the fact that this network appears to have negatively curved, positively curved and flat regions, and therefore models which reflect only a single curvature, may not capture the individual level behavior of the network very well. In contrast, condensed matter physics and HEP physics seem to have a slight curvature, though both are nearly flat networks. We remark on the large $p$-values for 3 of the networks in Table~\ref{tab:physics_curve_pvals}. In these settings, the second and 3rd midpoint sets proved to have poorer alignment and therefore were quite conservative when constructing the constant curvature test, leading to larger $p$-values.

\section{Application: Multiple Change Point Detection} \label{sec:application_changepoint}
Another natural question one may seek to answer is whether a change in latent curvature has occurred.  This is distinct from a problem of whether the latent positions may change.  Standard multiple change point algorithms (for example that of \citet{harchaoui2010multiple}) are not immediately well-suited to this problem due to the possibility of large outliers, which may occur in our setting, in particular for largely negative values. We apply the method of \citet{Fearnhead2016ChangepointOutliers} which proposes a multiple changepoint detection algorithm under the presence of outliers.  For a particular network, we may measure a collection of estimates of the curvature $\{\hat \kappa_{t}\}_{t = 1}^{T}$.  Under this model, we assume that a network has a constant curvature within a single time point $t$.  We construct an objective function for the changepoint problem
\al{
    L(\theta) &= \sum_{t = 1}^T \tilde \ell( \hat \kappa_{t},\theta_t) 
}

In our application, we let $\tilde \ell$ be the bi-weight loss function, however, more general loss functions such as absolute deviation or Huber loss are also available (see \citet{Fearnhead2016ChangepointOutliers} for a more in depth discussion of robust change point detection). In multiple change point detection algorithms a penalty of $\beta$ for the number of segments included is also applied.  Let $J(\theta)$ denote a function of the number of breaks in the sequence $\theta$.  Then the full loss function is 
\begin{equation}
    L(\theta; \beta) = L(\theta) + \beta J(\theta). 
\end{equation}

Since we often care about understanding where the changes of curvature occur, we can apply a monotone transformation to the estimates of curvature. The function $c\tanh(\cdot/c)$ function which smoothly truncates the extreme values under a monotone transformation.  In all simulations and applications where this is applied, we set $c = 10$.

We next apply this to a simulation setting where we construct a sequence of networks with latent positions evolving according to the following process
\al{
    Z^{(t)}_{i} &= \mathcal{F}(Z^{(t - 1)}_i, \epsilon^{(t)}_{i})
}
where $Z_{(t - 1)i}$ is the location's previous position, $\epsilon^{(t)}_{i}$ is a noise random variable sampled from the true cluster's density, and $\mathcal{F}$ stands for the mean on the sphere (the Fr\'echet mean).  We set three different curvatures $\kappa_1 = 1.0 \quad \kappa_2 = 0.15 \quad \kappa_3 = 1.3 $ to occur changepoints at $t_1 = 18, t_2 = 35$ and with the final time $T = 50$. 

We use the implementation of the changepoint method in the \texttt{robseg} package introduced in \citet{Fearnhead2016ChangepointOutliers} using the default regularization parameters.  We see that as clique size $\ell$ increases, we are able to consistently estimate the true curvature function. We plot the mean of the absolute deviation of the curvature estimate over $200$ simulations and plot the results in Figure~\ref{fig:rmse_changepoint} showing consistency of the true curvature with respect to the mean absolute deviation. 

\begin{figure}[htb!]
    \centering
    \includegraphics[width = 0.45\textwidth]{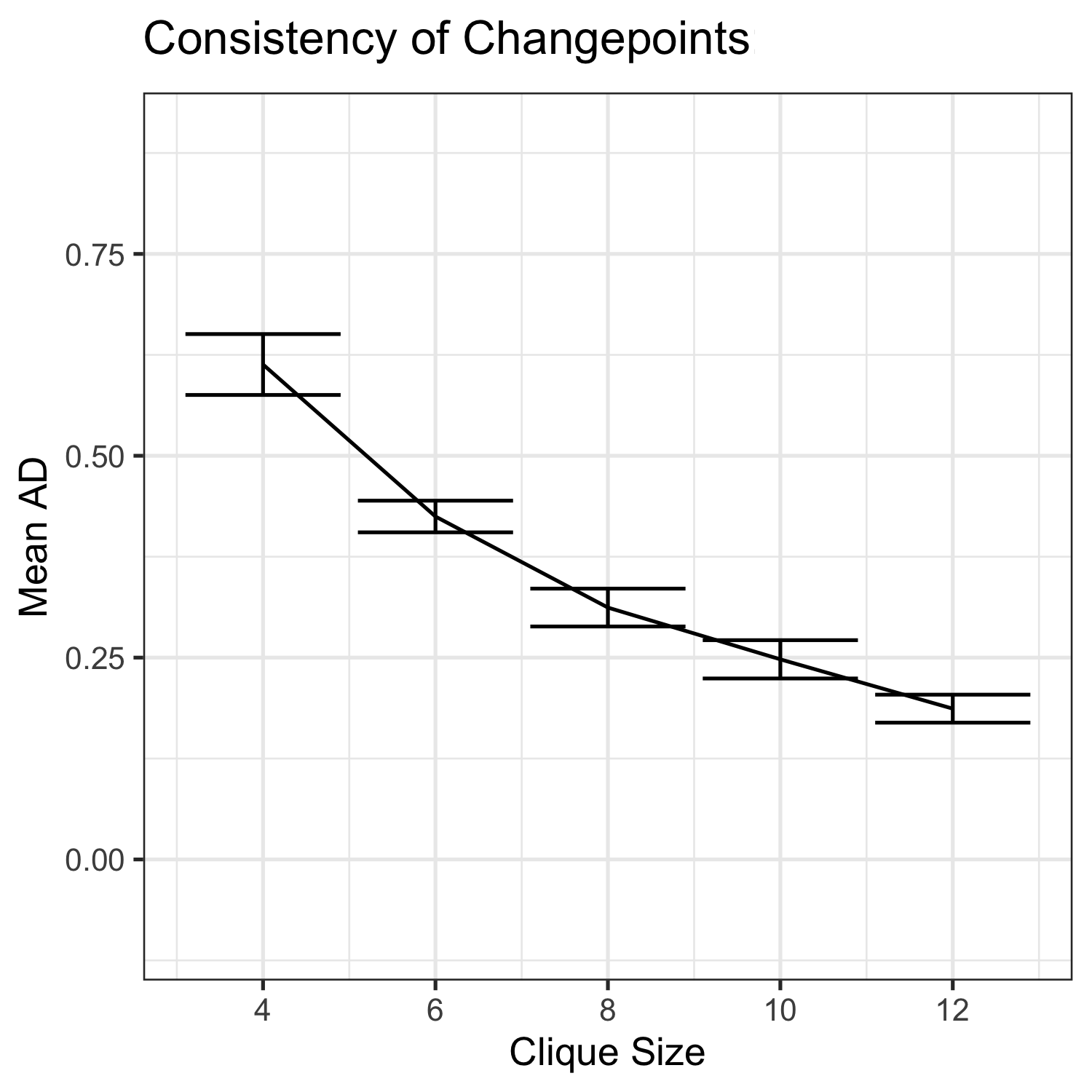}
    \caption{Mean absolute deviation of curvature estimates over time window.}
    \label{fig:rmse_changepoint}
\end{figure}

\subsection{Application: Cybersecurity} \label{sec:application_cybersecurity}

For our curvature changepoint application, we utilize our network curvature estimates on the Los Alamos Unified Network and Host Dataset to demonstrate that variations in curvature can help identify a red team attack; a controlled exercise where a cybersecurity team simulates an infiltration on a device network to test its security.

This dataset encompasses $89$ days of directed communication events among $27436$ devices at the Los Alamos National Laboratory. It records 56 normal operational days followed by a red team attack that spans from day $57$ to day $89$.

Anomaly detection holds significant importance in cybersecurity, and recent studies, such as \citep{Lee2019AnomalyModels}, have highlighted latent distance models as a promising method for detecting changes in node properties. In contrast, our research adopts changepoint methods applied to sequential curvature estimates within this dataset, underscoring the utility of curvature as a comprehensive indicator of network behavior.

Edges are defined in this dataset as messages passed between nodes during a particular time period. In order to maintain enough connections to find cliques, we consider a connection to be formed if any message was passed in the previous $4$ days.  We then compute the curvature values at each time point and take the median over the time-point. We scale each estimated value by $c\tanh(\cdot/c)$ for $c = 10$ in order to limit the influence of extreme negative outliers. We then apply the off the shelf change point algorithm of \citet{Fearnhead2016ChangepointOutliers}. 
We estimate the curvature at each time step.  We consider a minimum clique size of $\ell = 5$. Due to the small number of available cliques, and relative sparsity of the dataset, we restrict the random effects to be $0$ and compute the corresponding distance matrix. 

\begin{figure}[htb!] 
    \centering
    \begin{subfigure}[b]{0.35\textwidth}
        \centering
        \includegraphics[width=\textwidth]{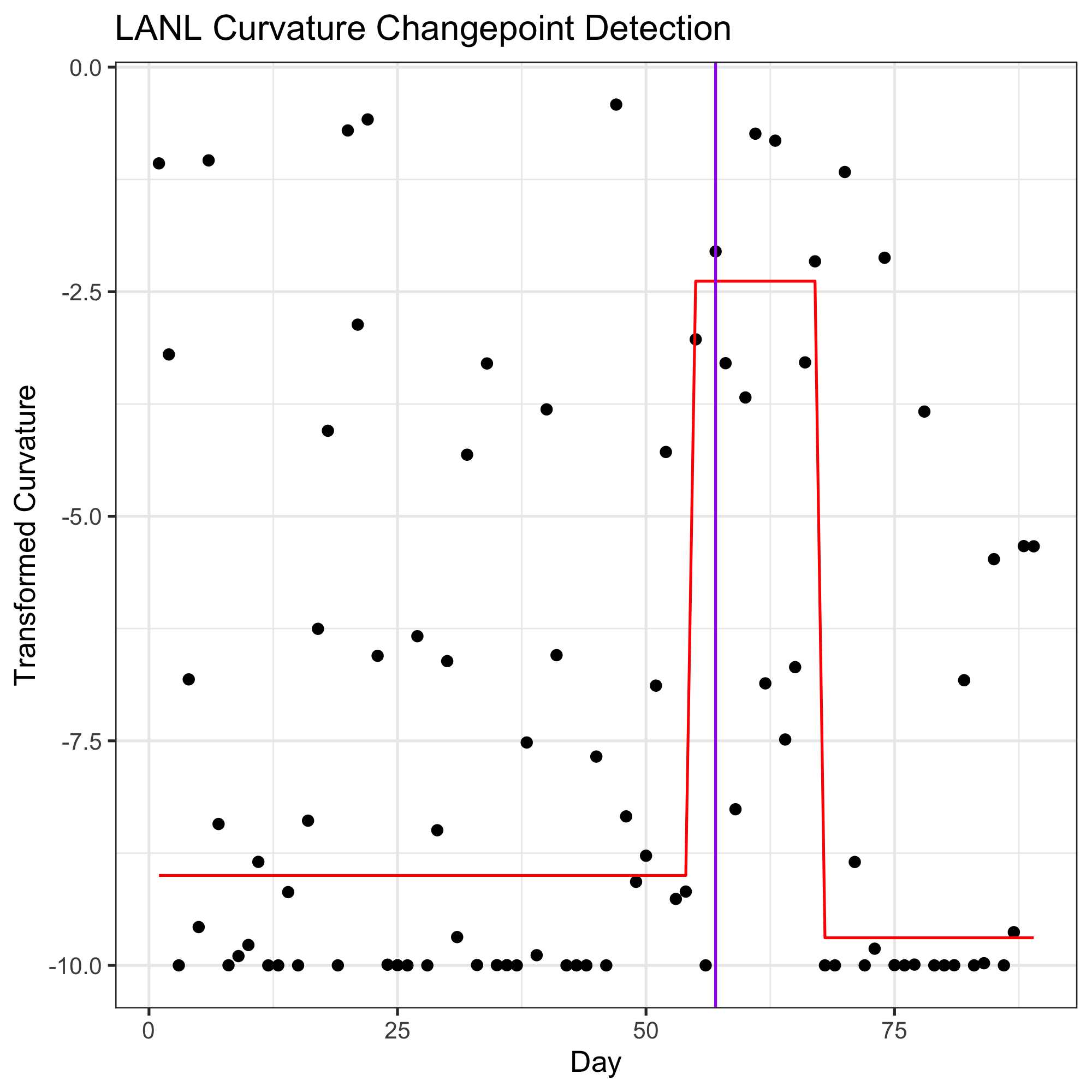}
        \caption{High regularization changes in curvature estimates.}
    \end{subfigure}%
    ~ 
    \begin{subfigure}[b]{0.35\textwidth}
        \centering
        \includegraphics[width=\textwidth]{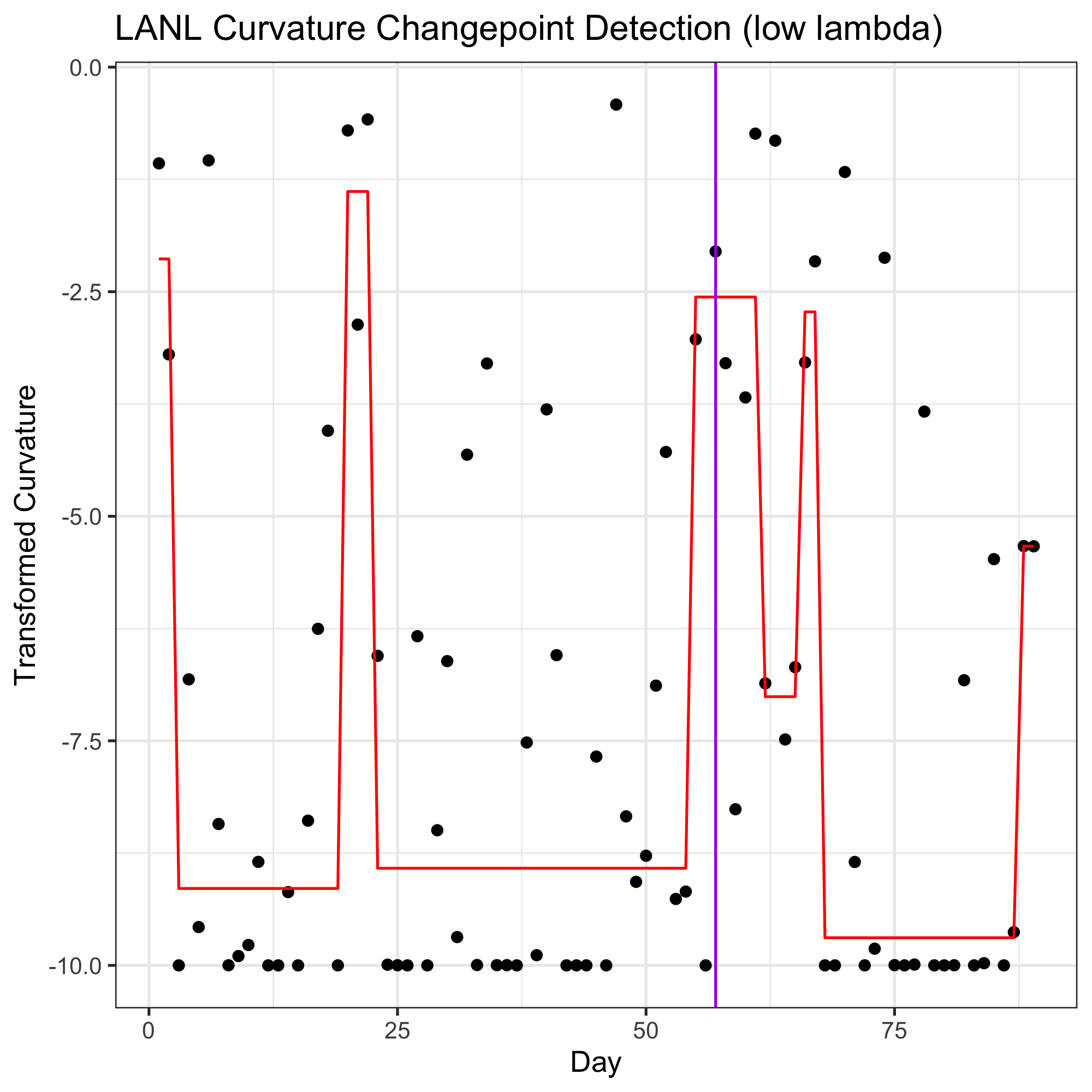}
        \caption{Low regularization changes in curvature estimates. }
    \end{subfigure}
    \caption{Two regularization values for LANL Netflow changepoint dataset. True red team attack time is illustrated in purple. }
    \label{fig:LANL_curvature_time_series}
\end{figure}

\begin{figure}[htb!] 
    \centering
    \begin{subfigure}[b]{0.35\textwidth}
        \centering
        \includegraphics[width=\textwidth]{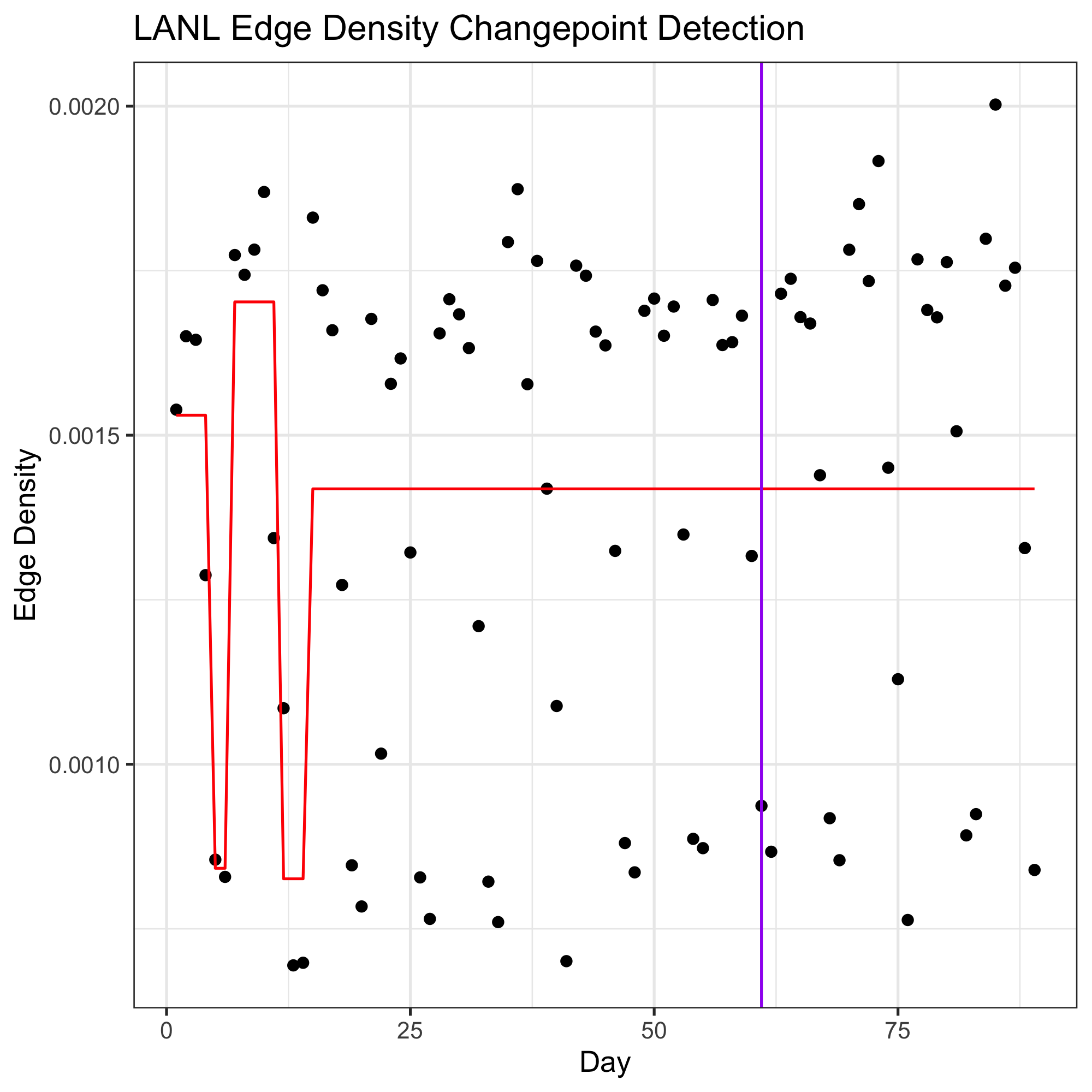}
        \caption{Edge density changepoints.}
    \end{subfigure}%
    ~ 
    \begin{subfigure}[b]{0.35\textwidth}
        \centering
        \includegraphics[width=\textwidth]{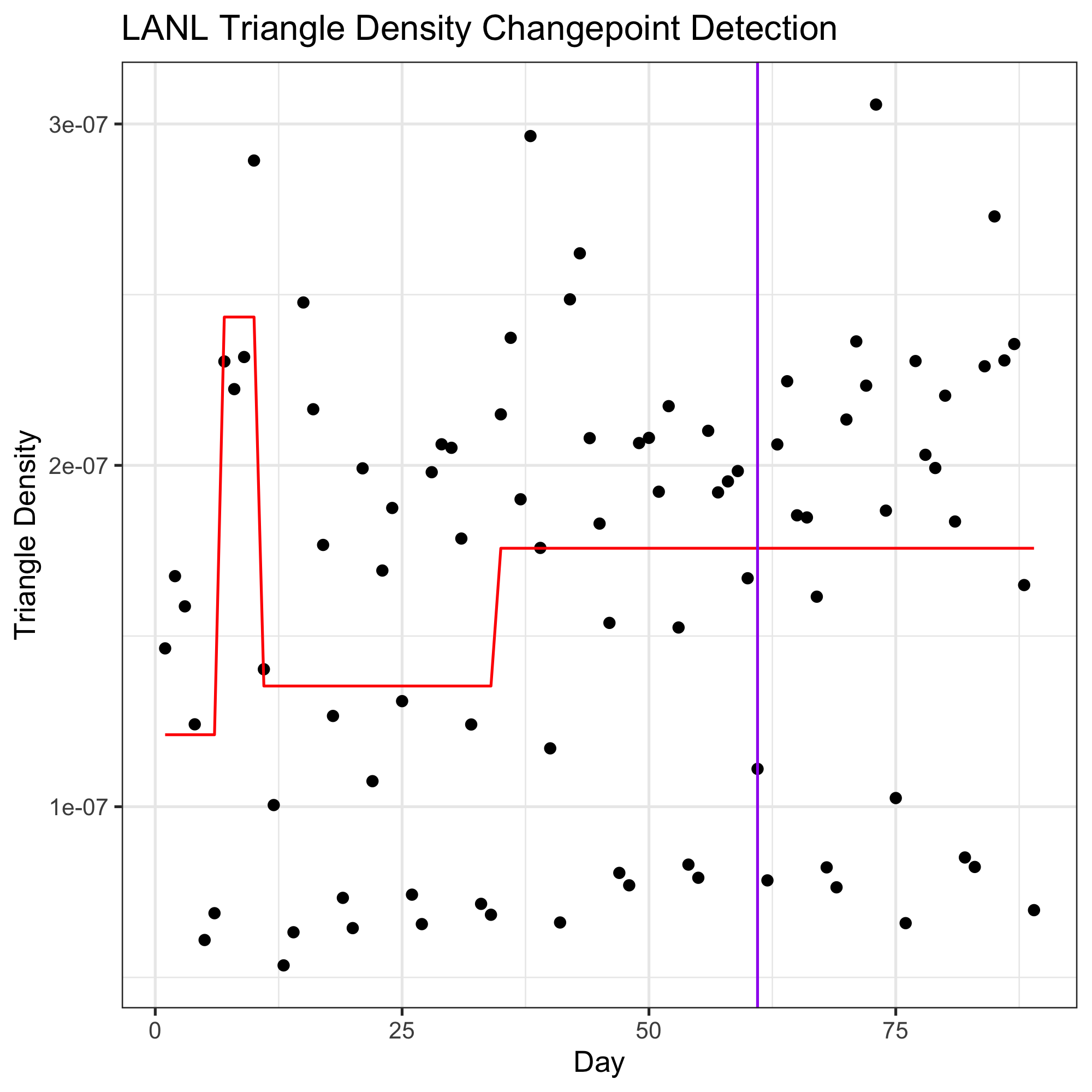}
        \caption{Tri. density changepoints.}
    \end{subfigure}
    \caption{Changepoints of daily LANL measurements using simple graph motifs. Red team attack illustrated in purple. }
    \label{fig:LANL_graph_motifs}
\end{figure}

In Figure~\ref{fig:LANL_curvature_time_series} we show that the most substantial changepoint in curvature occurs at the time of the red team attack.  In contrast, these changes are much less substantial in Figure~\ref{fig:LANL_graph_motifs} when using simple graph motifs from the daily averages.  

Since the time of detection after the event is most important, we wish to investigate the time after detection as a function of the number of events involved (alarm rate).  We show that in Figure~\ref{fig:detection_budget} that our method achieves a much smaller detection delay given any alarm rate. 

\begin{figure}[htb!] 
    \centering
    \includegraphics[width=0.45\textwidth]{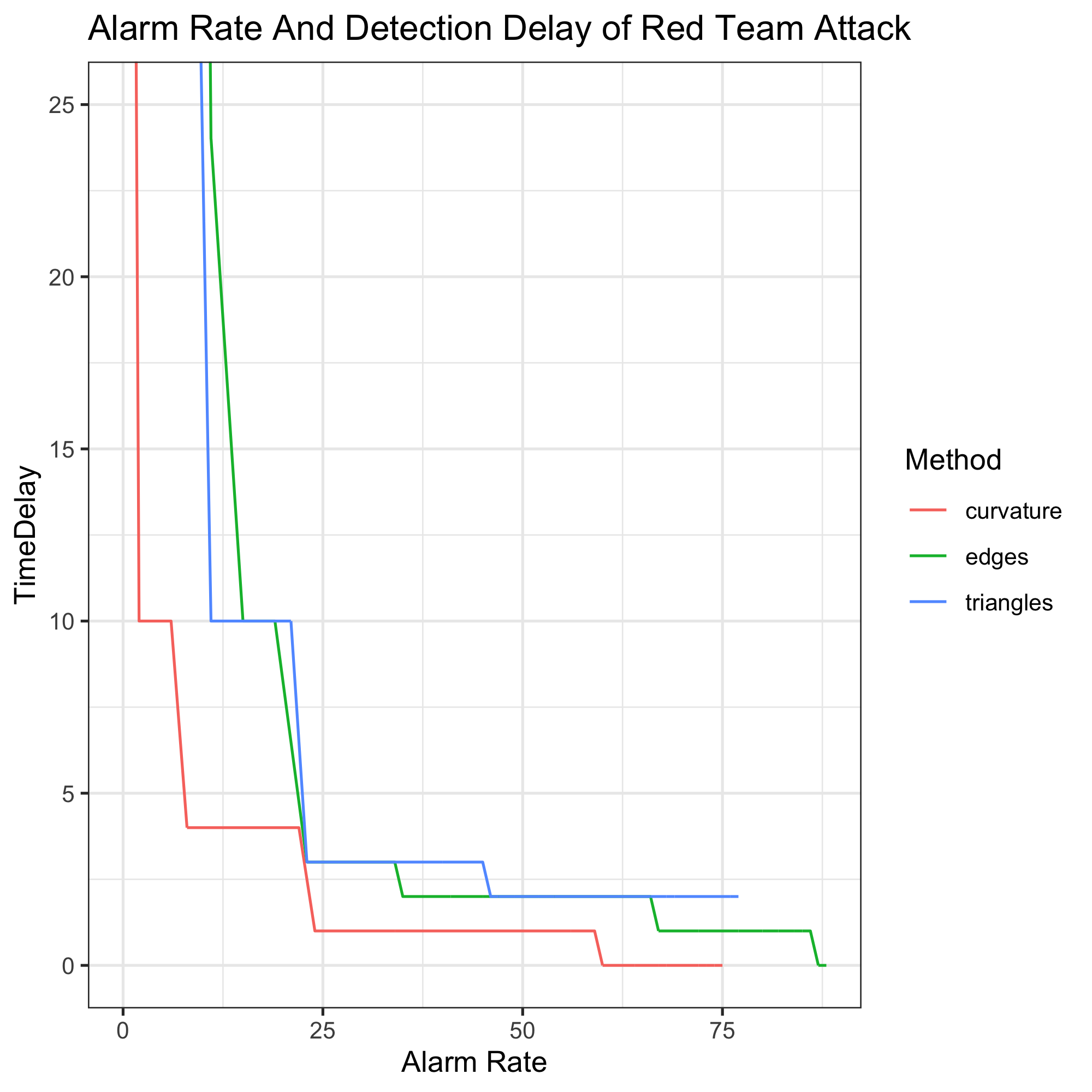}
    \caption{Time to first changepoint in days (TimeDelay) as a function of alarm rate. }
    \label{fig:detection_budget}
\end{figure}

This suggests that models accounting for changes in network curvature may be a promising avenue for the development of specialized models to detect anomalous events in the online setting.

\section{Discussion} \label{sec:discussion}

Riemannian sectional curvature is a fundamental property of a manifold, and we present a novel method to estimate it from a noisy distance matrix.  Though our motivating example involves estimating the distances of a latent distance model from a random network, the curvature estimates (and constant curvature tests) in this paper are more general and can be applied whenever one can either estimate a distance matrix (or bootstrap or subsample their distance matrices).

We develop a test for detecting whether the curvature is constant on the latent manifold.  A natural followup question is what should a practitioner do if they find that their data are not represented well by this model.  One might instead use non-geometric models such as stochastic block models and their variants, however, this also suggests the development of latent distance models which are not restricted to constant curvature. The development of such a model class as something which will scale to large networks is of further interest. One promising approach we plan to investigate in future work is that of product spaces for latent distance models. This geometry has seen considerable success in representation learning tasks such as \citet{gu2018learning} or \citet{zhang2021switch}.

One might question the interpretation of the curvature parameter $\kappa$ when the set of points does not reside within a canonical manifold. Our approach specifically fits a constant curvature manifold for each quartet $(x,y,z,m)$, embedding these four distances. By taking the median, we estimate the median curvature of these interpolating spaces. This concept echoes the manifold learning methods proposed by \citet{Li2017efficient, Li2019geodesic}, which use spherelets to approximate manifolds rather than using locally linear approximations of tangent spaces.

Future extensions may include methods for fitting models in the product space geometry (as is done in \cite{gu2018learning}) or other non-constant curvature spaces. Other methodology may include statistically consistent node-level definitions of network curvature as well as more application-focused development of anomaly detection incorporating curvature into the models.  Additional work may include using our local definition of curvature to understand its role in relation to notions of brokerage in the sociology literature \citep{burt1992structural, buskens2008dynamics}, as well as considering how curvature interacts with other quantities of interest across other sciences. Furthermore, additional applications may use non-parametric distance estimators to identify behaviors that form good midpoints in networks, such as overlapping sub-fields within physics in our applications.

\bibliographystyle{abbrvnat}
\bibliography{references}

\begin{thebibliography}{72}
\providecommand{\natexlab}[1]{#1}
\providecommand{\url}[1]{\texttt{#1}}
\expandafter\ifx\csname urlstyle\endcsname\relax
  \providecommand{\doi}[1]{doi: #1}\else
  \providecommand{\doi}{doi: \begingroup \urlstyle{rm}\Url}\fi

\bibitem[Aliverti and Durante(2019)]{Aliverti2019SpatialClustering}
E.~Aliverti and D.~Durante.
\newblock {Spatial modeling of brain connectivity data via latent distance models with nodes clustering}.
\newblock \emph{Statistical Analysis and Data Mining: The ASA Data Science Journal}, 12\penalty0 (3):\penalty0 185--196, 6 2019.
\newblock ISSN 1932-1872.
\newblock \doi{10.1002/SAM.11412}.
\newblock URL \url{https://onlinelibrary.wiley.com/doi/full/10.1002/sam.11412}.

\bibitem[Andersen(1970)]{andersen1970asymptotic}
E.~B. Andersen.
\newblock Asymptotic properties of conditional maximum-likelihood estimators.
\newblock \emph{Journal of the Royal Statistical Society: Series B (Methodological)}, 32\penalty0 (2):\penalty0 283--301, 1970.

\bibitem[Andersen and Andersen(2000)]{Andersen2000TheAlgorithm}
E.~D. Andersen and K.~D. Andersen.
\newblock The mosek interior point optimizer for linear programming: an implementation of the homogeneous algorithm.
\newblock In \emph{High performance optimization}, pages 197--232. Springer, 2000.

\bibitem[Banerjee et~al.(2013)Banerjee, Chandrasekhar, Duflo, and Jackson]{Banerjee2013TheMicrofinance}
A.~Banerjee, A.~G. Chandrasekhar, E.~Duflo, and M.~O. Jackson.
\newblock {The Diffusion of Microfinance}.
\newblock \emph{Science}, 341\penalty0 (6144), 2013.
\newblock ISSN 10959203.
\newblock \doi{10.1126/SCIENCE.1236498}.

\bibitem[Barkanass et~al.(2022)Barkanass, Jost, and Hancock]{Barkanass2022GeometricNetworks}
V.~Barkanass, U.~Jost, and E.~Hancock.
\newblock {Geometric sampling of networks}.
\newblock \emph{Journal of Complex Networks}, 10\penalty0 (4), 6 2022.
\newblock ISSN 2051-1310.
\newblock \doi{10.1093/COMNET/CNAC014}.
\newblock URL \url{https://academic.oup.com/comnet/article/10/4/cnac014/6644814}.

\bibitem[Bassett et~al.(2018)Bassett, Zurn, and Gold]{Bassett2018OnNeuroscience}
D.~S. Bassett, P.~Zurn, and J.~I. Gold.
\newblock {On the nature and use of models in network neuroscience}.
\newblock \emph{Nature Reviews Neuroscience 2018 19:9}, 19\penalty0 (9):\penalty0 566--578, 7 2018.
\newblock ISSN 1471-0048.
\newblock \doi{10.1038/s41583-018-0038-8}.
\newblock URL \url{https://www.nature.com/articles/s41583-018-0038-8}.

\bibitem[Begelfor and Werman(2005)]{Begelfor2005TheManifolds}
E.~Begelfor and M.~Werman.
\newblock {The World is not always Flat or Learning Curved Manifolds}.
\newblock \emph{School of Engineering and Computer Science, Hebrew University of Jerusalem., Tech. Rep}, 3\penalty0 (8), 2005.

\bibitem[Berger(1962)]{berger1962extension}
M.~Berger.
\newblock An extension of rauch's metric comparison theorem and some applications.
\newblock \emph{Illinois Journal of Mathematics}, 6\penalty0 (4):\penalty0 700--712, 1962.

\bibitem[Blumenthal and Gillam(1943)]{Blumenthal1943DistributionN-Space}
L.~M. Blumenthal and B.~E. Gillam.
\newblock {Distribution of Points in n-Space}.
\newblock \emph{The American Mathematical Monthly}, 50\penalty0 (3):\penalty0 181, 3 1943.
\newblock ISSN 00029890.
\newblock \doi{10.2307/2302400}.

\bibitem[Borgatti et~al.(2009)Borgatti, Mehra, Brass, and Labianca]{Borgatti2009NetworkSciences}
S.~P. Borgatti, A.~Mehra, D.~J. Brass, and G.~Labianca.
\newblock {Network analysis in the social sciences}.
\newblock \emph{Science}, 323\penalty0 (5916):\penalty0 892--895, 2 2009.
\newblock ISSN 00368075.
\newblock \doi{10.1126/SCIENCE.1165821/ASSET/2421E2F8-2DC7-4CAD-8A84-471B43A3C443/ASSETS/GRAPHIC/323_892_F5.JPEG}.
\newblock URL \url{https://www.science.org/doi/10.1126/science.1165821}.

\bibitem[Brauchart et~al.(2015)Brauchart, Reznikov, Saff, Sloan, Wang, and Womersley]{Brauchart2015RandomSeparation}
J.~S. Brauchart, A.~B. Reznikov, E.~B. Saff, I.~H. Sloan, Y.~G. Wang, and R.~S. Womersley.
\newblock {Random Point Sets on the Sphere --- Hole Radii, Covering, and Separation}.
\newblock \emph{Experimental Mathematics}, 27\penalty0 (1):\penalty0 62--81, 12 2015.
\newblock ISSN 1944950X.
\newblock \doi{10.1080/10586458.2016.1226209}.
\newblock URL \url{https://arxiv.org/abs/1512.07470v2}.

\bibitem[Burt(1992)]{burt1992structural}
R.~S. Burt.
\newblock \emph{Structural holes: The social structure of competition}.
\newblock Harvard University Press, Cambridge, MA, 1992.

\bibitem[Buskens and Van~de Rijt(2008)]{buskens2008dynamics}
V.~Buskens and A.~Van~de Rijt.
\newblock Dynamics of networks if everyone strives for structural holes.
\newblock \emph{American Journal of Sociology}, 114\penalty0 (2):\penalty0 371--407, 2008.

\bibitem[Cai et~al.(2013)Cai, Fan, and Jiang]{Cai2013DistributionsSpheres}
T.~Cai, J.~Fan, and T.~Jiang.
\newblock {Distributions of Angles in Random Packing on Spheres}.
\newblock \emph{Journal of Machine Learning Research}, 14:\penalty0 1837--1864, 2013.

\bibitem[Chamberlain et~al.(2017)Chamberlain, Clough, and Deisenroth]{Chamberlain2017NeuralSpaceb}
B.~P. Chamberlain, J.~R. Clough, and M.~P. Deisenroth.
\newblock {Neural Embeddings of Graphs in Hyperbolic Space}.
\newblock In \emph{13th international workshop on mining and learning from graphs held in conjunction with KDD}, 2017.

\bibitem[Chami et~al.(2019)Chami, Ying, R{\'{e}}, and Leskovec]{Chami2019HyperbolicNetworks}
I.~Chami, R.~Ying, C.~R{\'{e}}, and J.~Leskovec.
\newblock {Hyperbolic Graph Convolutional Neural Networks}.
\newblock \emph{Advances in Neural Information Processing Systems}, 32, 10 2019.
\newblock ISSN 10495258.
\newblock \doi{10.48550/arxiv.1910.12933}.
\newblock URL \url{https://arxiv.org/abs/1910.12933v1}.

\bibitem[Chen et~al.(2021)Chen, Kato, and Leng]{chen2021analysis}
M.~Chen, K.~Kato, and C.~Leng.
\newblock Analysis of networks via the sparse $\beta$-model.
\newblock \emph{Journal of the Royal Statistical Society Series B: Statistical Methodology}, 83\penalty0 (5):\penalty0 887--910, 2021.

\bibitem[Csardi et~al.(2006)Csardi, Nepusz, et~al.]{csardi2006igraph}
G.~Csardi, T.~Nepusz, et~al.
\newblock The igraph software package for complex network research.
\newblock \emph{InterJournal, complex systems}, 1695\penalty0 (5):\penalty0 1--9, 2006.

\bibitem[Durante et~al.(2013)Durante, Fern{\'{a}}ndez-S{\'{a}}nchez, and Sempi]{Durante2013ATheorem}
F.~Durante, J.~Fern{\'{a}}ndez-S{\'{a}}nchez, and C.~Sempi.
\newblock {A topological proof of Sklar’s theorem}.
\newblock \emph{Applied Mathematics Letters}, 26\penalty0 (9):\penalty0 945--948, 9 2013.
\newblock ISSN 0893-9659.
\newblock \doi{10.1016/J.AML.2013.04.005}.

\bibitem[Farooq et~al.(2019)Farooq, Chen, Georgiou, Tannenbaum, and Lenglet]{Farooq2019NetworkConnectivity}
H.~Farooq, Y.~Chen, T.~T. Georgiou, A.~Tannenbaum, and C.~Lenglet.
\newblock {Network curvature as a hallmark of brain structural connectivity}.
\newblock \emph{Nature Communications 2019 10:1}, 10\penalty0 (1):\penalty0 1--11, 10 2019.
\newblock ISSN 2041-1723.
\newblock \doi{10.1038/s41467-019-12915-x}.
\newblock URL \url{https://www.nature.com/articles/s41467-019-12915-x}.

\bibitem[Fearnhead and Rigaill(2016)]{Fearnhead2016ChangepointOutliers}
P.~Fearnhead and G.~Rigaill.
\newblock {Changepoint Detection in the Presence of Outliers}.
\newblock \emph{Journal of the American Statistical Association}, 114\penalty0 (525):\penalty0 169--183, 9 2016.
\newblock ISSN 1537274X.
\newblock \doi{10.48550/arxiv.1609.07363}.
\newblock URL \url{https://arxiv.org/abs/1609.07363v2}.

\bibitem[Fosdick et~al.(2016)Fosdick, McCormick, Murphy, Ng, and Westling]{Fosdick2016MultiresolutionModels}
B.~K. Fosdick, T.~H. McCormick, T.~B. Murphy, T.~L.~J. Ng, and T.~Westling.
\newblock {Multiresolution network models}.
\newblock \emph{Journal of Computational and Graphical Statistics}, 28\penalty0 (1):\penalty0 185--196, 8 2016.
\newblock ISSN 15372715.
\newblock \doi{10.48550/arxiv.1608.07618}.
\newblock URL \url{https://arxiv.org/abs/1608.07618v5}.

\bibitem[Frechet(1957)]{Frechet1957SurProbabilite}
M.~Frechet.
\newblock {Sur la distance de deux lois de probabilite}.
\newblock \emph{CR. Acad Sci. Paris}, 244, 1957.

\bibitem[Fu et~al.(2020)Fu, Narasimhan, and Boyd]{Fu2020CVXR:Optimization}
A.~Fu, B.~Narasimhan, and S.~Boyd.
\newblock {CVXR: An R package for disciplined convex optimization}.
\newblock \emph{Journal of Statistical Software}, 94\penalty0 (14):\penalty0 1--34, 11 2020.
\newblock ISSN 15487660.
\newblock \doi{10.18637/jss.v094.i14}.
\newblock URL \url{https://CRAN.R-project.}

\bibitem[Gallot et~al.(2004)Gallot, Hulin, and Lafontaine]{Gallot2004RiemannianGeometry}
S.~Gallot, D.~Hulin, and J.~Lafontaine.
\newblock \emph{{Riemannian Geometry}}.
\newblock Universitext. Springer Berlin Heidelberg, Berlin, Heidelberg, 2004.
\newblock ISBN 978-3-540-20493-0.
\newblock \doi{10.1007/978-3-642-18855-8}.
\newblock URL \url{http://link.springer.com/10.1007/978-3-642-18855-8}.

\bibitem[Gilbert and Jain(2017)]{gilbert2017if}
A.~C. Gilbert and L.~Jain.
\newblock If it ain't broke, don't fix it: Sparse metric repair.
\newblock In \emph{2017 55th Annual Allerton Conference on Communication, Control, and Computing (Allerton)}, pages 612--619. IEEE, 2017.

\bibitem[Grimmett and McDiarmid(1975)]{grimmett1975colouring}
G.~R. Grimmett and C.~J. McDiarmid.
\newblock On colouring random graphs.
\newblock In \emph{Mathematical Proceedings of the Cambridge Philosophical Society}, volume~77, pages 313--324. Cambridge University Press, 1975.

\bibitem[Gromov(2007)]{2007MetricSpaces}
M.~Gromov, editor.
\newblock \emph{{Metric Structures for Riemannian and NonRiemannian Spaces}}.
\newblock Birkh{\"{a}}user Boston, 2007.
\newblock \doi{10.1007/978-0-8176-4583-0}.

\bibitem[Gu et~al.(2018)Gu, Sala, Gunel, and R{\'e}]{gu2018learning}
A.~Gu, F.~Sala, B.~Gunel, and C.~R{\'e}.
\newblock Learning mixed-curvature representations in product spaces.
\newblock In \emph{International Conference on Learning Representations}, 2018.

\bibitem[Handcock et~al.(2007)Handcock, Raftery, and Tantrum]{Handcock2007Model-basedNetworks}
M.~S. Handcock, A.~E. Raftery, and J.~M. Tantrum.
\newblock {Model-based clustering for social networks}.
\newblock \emph{Journal of the Royal Statistical Society: Series A (Statistics in Society)}, 170\penalty0 (2):\penalty0 301--354, 3 2007.
\newblock ISSN 1467-985X.
\newblock \doi{10.1111/J.1467-985X.2007.00471.X}.
\newblock URL \url{https://onlinelibrary.wiley.com/doi/full/10.1111/j.1467-985X.2007.00471.x https://onlinelibrary.wiley.com/doi/abs/10.1111/j.1467-985X.2007.00471.x https://rss.onlinelibrary.wiley.com/doi/10.1111/j.1467-985X.2007.00471.x}.

\bibitem[Harchaoui and L{\'e}vy-Leduc(2010)]{harchaoui2010multiple}
Z.~Harchaoui and C.~L{\'e}vy-Leduc.
\newblock Multiple change-point estimation with a total variation penalty.
\newblock \emph{Journal of the American Statistical Association}, 105\penalty0 (492):\penalty0 1480--1493, 2010.

\bibitem[Havel and W{\"u}thrich(1985)]{havel1985evaluation}
T.~F. Havel and K.~W{\"u}thrich.
\newblock An evaluation of the combined use of nuclear magnetic resonance and distance geometry for the determination of protein conformations in solution.
\newblock \emph{Journal of molecular biology}, 182\penalty0 (2):\penalty0 281--294, 1985.

\bibitem[Hoeffding(1940)]{Hoeffding1940MasstabvariarteKorrelationstheorie.}
W.~Hoeffding.
\newblock {Masstabvariarte Korrelationstheorie.}
\newblock \emph{Schrijl Math. Inst. Univ. Berlin}, 5\penalty0 (6), 1940.

\bibitem[Hoff(2007)]{hoff2007modeling}
P.~Hoff.
\newblock Modeling homophily and stochastic equivalence in symmetric relational data.
\newblock \emph{Advances in neural information processing systems}, 20, 2007.

\bibitem[Hoff et~al.(2002)Hoff, Raftery, and Handcock]{Hoff2002latent}
P.~D. Hoff, A.~E. Raftery, and M.~S. Handcock.
\newblock Latent space approaches to social network analysis.
\newblock \emph{Journal of the american Statistical association}, 97\penalty0 (460):\penalty0 1090--1098, 2002.

\bibitem[Hopf(1926)]{hopf1926clifford}
H.~Hopf.
\newblock Zum clifford-kieinschen raumproblem.
\newblock \emph{Mathematische Annalen}, 95:\penalty0 313--339, 1926.

\bibitem[Killing(1891)]{Killing1891UeberRaumformen}
W.~Killing.
\newblock {Ueber die Clifford-Klein'schen Raumformen}.
\newblock \emph{Mathematische Annalen 1891 39:2}, 39\penalty0 (2):\penalty0 257--278, 6 1891.
\newblock ISSN 1432-1807.
\newblock \doi{10.1007/BF01206655}.
\newblock URL \url{https://link.springer.com/article/10.1007/BF01206655}.

\bibitem[Kim et~al.(2018)Kim, Lee, Xue, and Niu]{kim2018review}
B.~Kim, K.~H. Lee, L.~Xue, and X.~Niu.
\newblock A review of dynamic network models with latent variables.
\newblock \emph{Statistics surveys}, 12:\penalty0 105, 2018.

\bibitem[Klingenberg(1995)]{Klingenberg1995RiemannianGeometry}
W.~Klingenberg.
\newblock \emph{{Riemannian Geometry}}.
\newblock Walter de Gruyter, 1 edition, 12 1995.
\newblock \doi{10.1515/9783110905120}.

\bibitem[Kusner et~al.(2015)Kusner, Sun, Kolkin, and Weinberger]{kusner2015word}
M.~Kusner, Y.~Sun, N.~Kolkin, and K.~Weinberger.
\newblock From word embeddings to document distances.
\newblock In \emph{International conference on machine learning}, pages 957--966. PMLR, 2015.

\bibitem[Le~Cam(1956)]{le1956asymptotic}
L.~Le~Cam.
\newblock On the asymptotic theory of estimation and testing hypotheses.
\newblock In \emph{Proceedings of the Third Berkeley Symposium on Mathematical Statistics and Probability, Volume 1: Contributions to the Theory of Statistics}, volume~3, pages 129--157. University of California Press, 1956.

\bibitem[Leal et~al.(2018)Leal, Restrepo, Stadler, and Jost]{Leal2018Forman-RicciHypergraphs}
W.~Leal, G.~Restrepo, P.~F. Stadler, and J.~Jost.
\newblock {Forman-Ricci Curvature for Hypergraphs}.
\newblock \emph{Advances in Complex Systems}, 24\penalty0 (1), 11 2018.
\newblock \doi{10.13140/RG.2.2.27347.84001}.
\newblock URL \url{http://arxiv.org/abs/1811.07825 http://dx.doi.org/10.13140/RG.2.2.27347.84001}.

\bibitem[Lee et~al.(2019)Lee, Mccormick, Neil, Cole, Microsoft, and Cui]{Lee2019AnomalyModels}
W.~Lee, T.~H. Mccormick, J.~Neil, M.~Cole, S.~Microsoft, and Y.~Cui.
\newblock {Anomaly Detection in Large Scale Networks with Latent Space Models}.
\newblock \emph{Technometrics}, pages 1--23, 11 2019.
\newblock ISSN 0040-1706.
\newblock \doi{10.1080/00401706.2021.1952900}.
\newblock URL \url{https://arxiv.org/abs/1911.05522v2}.

\bibitem[Leskovec et~al.(2007)Leskovec, Kleinberg, and Faloutsos]{Leskovec2007GraphEvolution}
J.~Leskovec, J.~Kleinberg, and C.~Faloutsos.
\newblock {Graph evolution}.
\newblock \emph{ACM Transactions on Knowledge Discovery from Data (TKDD)}, 1\penalty0 (1), 3 2007.
\newblock ISSN 15564681.
\newblock \doi{10.1145/1217299.1217301}.
\newblock URL \url{https://dl.acm.org/doi/abs/10.1145/1217299.1217301}.

\bibitem[Li and Dunson(2019)]{Li2019geodesic}
D.~Li and D.~B. Dunson.
\newblock Geodesic distance estimation with spherelets.
\newblock \emph{arXiv preprint arXiv:1907.00296}, 2019.

\bibitem[Li et~al.(2017)Li, Mukhopadhyay, and Dunson]{Li2017efficient}
D.~Li, M.~Mukhopadhyay, and D.~B. Dunson.
\newblock Efficient manifold and subspace approximations with spherelets.
\newblock \emph{arXiv preprint arXiv:1706.08263}, 2017.

\bibitem[Lok et~al.(2021)Lok, Ng, Murphy, Westling, Mccormick, and Fosdick]{Lok2021ModelingBlockmodel}
T.~Lok, J.~Ng, T.~B. Murphy, T.~Westling, T.~H. Mccormick, and B.~Fosdick.
\newblock {Modeling the social media relationships of Irish politicians using a generalized latent space stochastic blockmodel}.
\newblock \emph{Annals of Applied Statistics}, 15\penalty0 (4):\penalty0 1923--1944, 12 2021.
\newblock ISSN 1932-6157.
\newblock \doi{10.1214/21-AOAS1483}.

\bibitem[Lubold et~al.(2023)Lubold, Chandrasekhar, and McCormick]{lubold2023identifying}
S.~Lubold, A.~G. Chandrasekhar, and T.~H. McCormick.
\newblock Identifying the latent space geometry of network models through analysis of curvature.
\newblock \emph{Journal of the Royal Statistical Society Series B: Statistical Methodology}, 2023.

\bibitem[MacDonald et~al.(2022)MacDonald, Levina, and Zhu]{Macdonald2022latent}
P.~W. MacDonald, E.~Levina, and J.~Zhu.
\newblock Latent space models for multiplex networks with shared structure.
\newblock \emph{Biometrika}, 109\penalty0 (3):\penalty0 683--706, 2022.

\bibitem[Nagano et~al.(2019)Nagano, Yamaguchi, Fujita, and Koyama]{Nagano2019ALearning}
Y.~Nagano, S.~Yamaguchi, Y.~Fujita, and M.~Koyama.
\newblock {A Wrapped Normal Distribution on Hyperbolic Space for Gradient-Based Learning}.
\newblock \emph{36th International Conference on Machine Learning, ICML 2019}, 2019-June:\penalty0 8242--8251, 2 2019.
\newblock \doi{10.48550/arxiv.1902.02992}.
\newblock URL \url{https://arxiv.org/abs/1902.02992v2}.

\bibitem[Ni et~al.(2019)Ni, Lin, Luo, and Gao]{Ni2019CommunityFlow}
C.~C. Ni, Y.~Y. Lin, F.~Luo, and J.~Gao.
\newblock {Community Detection on Networks with Ricci Flow}.
\newblock \emph{Scientific Reports 2019 9:1}, 9\penalty0 (1):\penalty0 1--12, 7 2019.
\newblock ISSN 2045-2322.
\newblock \doi{10.1038/s41598-019-46380-9}.
\newblock URL \url{https://www.nature.com/articles/s41598-019-46380-9}.

\bibitem[Nickel and Kiela(2017)]{Nickel2017PoincareRepresentations}
M.~Nickel and D.~Kiela.
\newblock {Poincar{\'{e}} Embeddings for Learning Hierarchical Representations}.
\newblock \emph{Advances in Neural Information Processing Systems}, 30, 2017.

\bibitem[Ollivier(2007)]{Ollivier2007RicciSpaces}
Y.~Ollivier.
\newblock {Ricci curvature of Markov chains on metric spaces}.
\newblock \emph{Journal of Functional Analysis}, 256\penalty0 (3):\penalty0 810--864, 1 2007.
\newblock ISSN 00221236.
\newblock \doi{10.48550/arxiv.math/0701886}.
\newblock URL \url{https://arxiv.org/abs/math/0701886v4}.

\bibitem[Papadopoulos et~al.(2018)Papadopoulos, Porter, Daniels, and Bassett]{Papadopoulos2018NetworkGrains}
L.~Papadopoulos, M.~A. Porter, K.~E. Daniels, and D.~S. Bassett.
\newblock {Network analysis of particles and grains}.
\newblock \emph{Journal of Complex Networks}, 6\penalty0 (4):\penalty0 485--565, 8 2018.
\newblock ISSN 20511329.
\newblock \doi{10.1093/COMNET/CNY005}.
\newblock URL \url{https://academic.oup.com/comnet/article/6/4/485/4959635}.

\bibitem[Pennec(1999)]{Pennec1999ProbabilitiesMeasurements.}
X.~Pennec.
\newblock {Probabilities and statistics on Riemannian manifolds: Basic tools for geometric measurements.}
\newblock In \emph{International Workshop on Nonlinear Signal and Image Processing}, pages 194--198, Antalya, Turkey, 6 1999.
\newblock URL \url{http://www-sop.inria.fr/epidaure/personnel/pennec/pennec.html}.

\bibitem[Politis and Romano(1994)]{Politis1994LargeAssumptions}
D.~N. Politis and J.~P. Romano.
\newblock {Large Sample Confidence Regions Based on Subsamples under Minimal Assumptions}.
\newblock \emph{https://doi.org/10.1214/aos/1176325770}, 22\penalty0 (4):\penalty0 2031--2050, 12 1994.
\newblock ISSN 0090-5364.
\newblock \doi{10.1214/AOS/1176325770}.

\bibitem[Salter-Townshend and McCormick(2017)]{Salter-Townshend2017LatentData}
M.~Salter-Townshend and T.~H. McCormick.
\newblock {Latent space models for multiview network data}.
\newblock \emph{Annals of Applied Statistics}, 11\penalty0 (3):\penalty0 1217--1244, 9 2017.
\newblock ISSN 1932-6157.
\newblock \doi{10.1214/16-AOAS955}.
\newblock URL \url{https://projecteuclid.org/journals/annals-of-applied-statistics/volume-11/issue-3/Latent-space-models-for-multiview-network-data/10.1214/16-AOAS955.full}.

\bibitem[Samal et~al.(2021)Samal, Pharasi, Ramaia, Kannan, Saucan, Jost, and Chakraborti]{Samal2021NetworkInstability}
A.~Samal, H.~K. Pharasi, S.~J. Ramaia, H.~Kannan, E.~Saucan, J.~Jost, and A.~Chakraborti.
\newblock {Network geometry and market instability}.
\newblock \emph{Royal Society Open Science}, 8\penalty0 (2), 2 2021.
\newblock ISSN 20545703.
\newblock \doi{10.1098/RSOS.201734}.
\newblock URL \url{https://royalsocietypublishing.org/doi/10.1098/rsos.201734}.

\bibitem[Sandhu et~al.(2015)Sandhu, Georgiou, Reznik, Zhu, Kolesov, Senbabaoglu, and Tannenbaum]{Sandhu2015GraphNetworks}
R.~Sandhu, T.~Georgiou, E.~Reznik, L.~Zhu, I.~Kolesov, Y.~Senbabaoglu, and A.~Tannenbaum.
\newblock {Graph Curvature for Differentiating Cancer Networks}.
\newblock \emph{Nature Scientific Reports}, 5\penalty0 (1):\penalty0 1--13, 7 2015.
\newblock ISSN 2045-2322.
\newblock \doi{10.1038/srep12323}.
\newblock URL \url{https://www.nature.com/articles/srep12323}.

\bibitem[Sandhu et~al.(2016)Sandhu, Georgiou, and Tannenbaum]{Sandhu2016RicciRisk}
R.~S. Sandhu, T.~T. Georgiou, and A.~R. Tannenbaum.
\newblock {Ricci curvature: An economic indicator for market fragility and systemic risk}.
\newblock \emph{Science Advances}, 2\penalty0 (5), 5 2016.
\newblock ISSN 23752548.
\newblock \doi{10.1126/SCIADV.1501495/SUPPL_FILE/1501495_SM.PDF}.
\newblock URL \url{https://www.science.org/doi/10.1126/sciadv.1501495}.

\bibitem[Saucan et~al.(2020)Saucan, Samal, and Jost]{Saucan2020ANetworks}
E.~Saucan, A.~Samal, and J.~Jost.
\newblock {A Simple Differential Geometry for Complex Networks}.
\newblock \emph{Network Science}, 9\penalty0 (S1):\penalty0 S106--S133, 4 2020.
\newblock ISSN 20501250.
\newblock \doi{10.48550/arxiv.2004.11112}.
\newblock URL \url{https://arxiv.org/abs/2004.11112v2}.

\bibitem[Schoenberg(1935)]{Schoenberg1935RemarksHilbert}
I.~J. Schoenberg.
\newblock {Remarks to Maurice Frechet's Article ``Sur La Definition Axiomatique D'Une Classe D'Espace Distances Vectoriellement Applicable Sur L'Espace De Hilbert}.
\newblock \emph{The Annals of Mathematics}, 36\penalty0 (3):\penalty0 724, 7 1935.
\newblock ISSN 0003486X.
\newblock \doi{10.2307/1968654}.

\bibitem[Sia et~al.(2019)Sia, Jonckheere, and Bogdan]{Sia2019Ollivier-RicciNetworks}
J.~Sia, E.~Jonckheere, and P.~Bogdan.
\newblock {Ollivier-Ricci Curvature-Based Method to Community Detection in Complex Networks}.
\newblock \emph{Nature, Scientific Reports 2019 9:1}, 9\penalty0 (1):\penalty0 1--12, 7 2019.
\newblock ISSN 2045-2322.
\newblock \doi{10.1038/s41598-019-46079-x}.
\newblock URL \url{https://www.nature.com/articles/s41598-019-46079-x}.

\bibitem[Sklar(1959)]{SKLAR1959FonctionsMarges}
A.~Sklar.
\newblock {Fonctions de repartition a n dimensions et leurs marges}.
\newblock \emph{Publications de l'Institut de statistique de l'Universit{\'{e}} de Paris}, 8:\penalty0 229--231, 1959.
\newblock URL \url{https://ci.nii.ac.jp/naid/10011938360}.

\bibitem[Smith et~al.(2017)Smith, Asta, and Calder]{Smith2017TheData}
A.~L. Smith, D.~M. Asta, and C.~A. Calder.
\newblock {The Geometry of Continuous Latent Space Models for Network Data}.
\newblock \emph{Statistical Science}, 34\penalty0 (3):\penalty0 428--453, 12 2017.
\newblock ISSN 21688745.
\newblock \doi{10.48550/arxiv.1712.08641}.
\newblock URL \url{https://arxiv.org/abs/1712.08641v2}.

\bibitem[Sweet and Adhikari(2020)]{Sweet2020AInfluence}
T.~Sweet and S.~Adhikari.
\newblock {A Latent Space Network Model for Social Influence}.
\newblock \emph{Psychometrika}, 85\penalty0 (2):\penalty0 251--274, 6 2020.
\newblock ISSN 18600980.
\newblock \doi{10.1007/S11336-020-09700-X/FIGURES/11}.
\newblock URL \url{https://link.springer.com/article/10.1007/s11336-020-09700-x}.

\bibitem[Tenenbaum et~al.(2000)Tenenbaum, Silva, and Langford]{tenenbaum2000global}
J.~B. Tenenbaum, V.~d. Silva, and J.~C. Langford.
\newblock A global geometric framework for nonlinear dimensionality reduction.
\newblock \emph{science}, 290\penalty0 (5500):\penalty0 2319--2323, 2000.

\bibitem[Torgerson(1952)]{torgerson1952multidimensional}
W.~S. Torgerson.
\newblock Multidimensional scaling: I. theory and method.
\newblock \emph{Psychometrika}, 17\penalty0 (4):\penalty0 401--419, 1952.

\bibitem[van~der Hoorn et~al.(2020)van~der Hoorn, Cunningham, Lippner, Trugenberger, and Krioukov]{vanderHoorn2020Ollivier-RicciGraphs}
P.~van~der Hoorn, W.~J. Cunningham, G.~Lippner, C.~Trugenberger, and D.~Krioukov.
\newblock {Ollivier-Ricci curvature convergence in random geometric graphs}.
\newblock \emph{Physical Review Research}, 3\penalty0 (1), 8 2020.
\newblock \doi{10.1103/PhysRevResearch.3.013211}.
\newblock URL \url{http://arxiv.org/abs/2008.01209 http://dx.doi.org/10.1103/PhysRevResearch.3.013211}.

\bibitem[van~der Vaart(1998)]{vanderVaart1998AsymptoticStatistics}
A.~van~der Vaart.
\newblock \emph{{Asymptotic Statistics}}.
\newblock Cambridge University Press, 10 1998.
\newblock \doi{10.1017/cbo9780511802256}.
\newblock URL \url{/core/books/asymptotic-statistics/A3C7DAD3F7E66A1FA60E9C8FE132EE1D}.

\bibitem[Volz et~al.(2011)Volz, Miller, Galvani, and Meyers]{Volz2011EffectsDynamics}
E.~M. Volz, J.~C. Miller, A.~Galvani, and L.~Meyers.
\newblock {Effects of Heterogeneous and Clustered Contact Patterns on Infectious Disease Dynamics}.
\newblock \emph{PLOS Computational Biology}, 7\penalty0 (6):\penalty0 e1002042, 2011.
\newblock ISSN 1553-7358.
\newblock \doi{10.1371/JOURNAL.PCBI.1002042}.
\newblock URL \url{https://journals.plos.org/ploscompbiol/article?id=10.1371/journal.pcbi.1002042}.

\bibitem[Zhang et~al.(2021)Zhang, Tay, Jiang, Juan, and Zhang]{zhang2021switch}
S.~Zhang, Y.~Tay, W.~Jiang, D.-c. Juan, and C.~Zhang.
\newblock Switch spaces: learning product spaces with sparse gating.
\newblock \emph{arXiv preprint arXiv:2102.08688}, 2021.

\end{thebibliography}


\appendix

\section{Proofs of Theorems}
\subsection{Proof of Lemma~\ref{lem:submanifold_of_dim_2}} \label{sec:proof_submanifold_of_dim_2}

In order to prove this lemma, we first introduce a useful result. 
\begin{theorem}[Theorem 1.10.15 in \citet{Klingenberg1995RiemannianGeometry}] \label{thm:isometry}
    If $f:(\bar M, \bar g) \to (\bar M, \bar g)$ is an isometry on a Riemannian manifold, then the fixed point set of $f$ forms a totally geodesic submanifold. 
\end{theorem}

Using this theorem, if we can construct an isometry, a bijective isomorphism between two metric spaces (in this case, the submanifold within the canonical manifold which contains the triangle, as well as the canonical manifold of dimension 2), then the fixed points of the set will form a totally geodesic submanifold which will be useful for constructing our fixed point equation.  

\begin{proof}
    We will prove this by first considering constructing a change of basis to parameterize the sub-manifold using the first $3$ coordinates, then we will construct the isometry between the manifold of dimension $p \geq 2$ and that of dimension $2$.

    Consider  a set of 3 points  $(x,y,z) \in \Sp^p(\kappa)$ which are not co-linear.  
    We can construct an orthogonal matrix (rotation matrix) $Q \in \R^{p+1 \times p+1}$ which allows us to construct a rotational isometry. 
    
    WLOG, in our coordinate system, we place $x$ at the origin, i.e. $x = (1,0, \dots, 0)$. We construct an orthogonal rotation matrix $Q$ which rotates the coordinates into the first $3$ indices, with the rest of them being $0$. 
    Let $\tilde q_i$ denote the $i^{th}$ un-normalized column vector of $Q$, and $q_i = \tilde q_i/ \norm{\tilde q_i}_2$.  We the define the following first $3$ basis functions as the normalized projections  of the components of $y$ and then $z$:
    \al{
        \tilde q_1 &= [1,0,\dots, 0]^\top \\
        \tilde q_2 &= y - (y^\top q_1) q_1 \\
        \tilde q_3 &= z -  (z^\top q_2) q_2 -  (z^\top q_1) q_1.
    }
    Since $x,y,z$ are not colinear, then $q_1, q_2, q_3$ are independent and are the normalization's of $\tilde q_1, \tilde q_2, \tilde q_3$ respectively and are orthogonal.  The remaining columns of $Q$ are the completion of the orthonormal basis with any remaining orthogonal basis vectors of $\R^p$. Thus we have constructed an orthogonal matrix $Q$ and $Q^\top Q = I = QQ^\top$.  
    This matrix $Q$ can be used to construct an isometry $f_1(\cdot): \Sp^p \mapsto \Sp^p$ where 
    \al{
        f_1(x) &= Q^\top x \\
        d(f_1(x),f_1(y)) &= \frac{1}{\sqrt{\kappa}}\acos\left( f_1(x)^\top f_1(y) \right) \\
        &= \frac{1}{\sqrt{\kappa}}\acos\left( x^\top Q Q^\top y \right) \\
        &= \frac{1}{\sqrt{\kappa}}\acos\left( x^\top y \right) \\
        &= d(x,y).
    }
    Hence, this rotation is an isometry. Therefore, we can equivalently use the parameterization where the points \( x, y, z \) have non-zero coordinates only in the first three indices.
    
    Next, let $f_2(x) = (x_0,x_1,x_2,-x_3,-x_4, \dots, -x_p)$ denote an second mapping.  Under this transformation
    \al{
        f_2(x)^\top f_2(y) &= \sum_{i = 0}^2(x_i)(y_i) + \sum_{i = 3}^p (-x_i)(-y_i) \\
        &= x^T y \\
        \implies d(f_2(x),f_2(y)) &= d(x,y) 
    }
    and thus $f_2$ is also an isometry. Since the composition of isometries is also an isometry. By the composition of this rotation and sign flip of coordinates, we can construct the corresponding totally geodesic submanifold $\mathcal{M}^2(\kappa)$ as follows.  Lastly the totally geodesic submanifold of dimension $2$ can be constructed using the set of points satisfying $\{Q v\}$ for $ v = [v_0,v_1,v_2, 0,0,\dots, 0]^\top$ for any $v_0, v_1, v_2$ such that $v_0^2 + v_1^2 + v_2^2 = 1$, which is simply a reparameterization of $\Sp^2(\kappa)$ mapped into $\Sp^p(\kappa)$. 
    By construction of this rotation matrix, $x,y,z$ are all fixed points of this isometry ($f_2 \circ  f_1$). When translating between coordinate positions and the distances, one must use the same curvature value $\kappa$, and thus the totally geodesic manifold of dimension $2$ exists with the same curvature, and thus contains its midpoints.
    
    The proofs for $\E^p$ and $\Hy^p$ follows this argument identically and thus the proof is complete. 
    
\end{proof}

\subsection{Proof of Theorem~\ref{thm:midpoint_curve_equation}} \label{sec:proof_midpoint_curve_equation}

\begin{proof}
    
    In order to develop our identifying equation for the curvature using triangle distances, we note that any three points $(x, y, z)$ can be embedded isometrically in a totally geodesic submanifold of dimension 2 when $\M^p(\kappa)$ is a canonical manifold, as stated in Lemma~\ref{lem:submanifold_of_dim_2}. Since this is a totally geodesic submanifold, the distance to the midpoint parameterized by coordinates in $\M^2(\kappa)$ will be the same as in $\M^p(\kappa)$. We continue with the proof by embedding the triangle $xyz$ in a canonical manifold of dimension 2. This will provide an implicit equation for the curvature, determined by the side lengths of the triangle as well as the length of the triangle median. In each case, we use the representations of the canonical manifolds outlined in Section~\ref{sec: geometric environment}.
    
    \textbf{Case 1: Spherical.} For convenience of derivation, we derive an implicit equation in a coordinate system where the midpoint of points $y$ and $z$ is placed at the origin ($m = (1,0,0)$ in $\mathbb{S}^2$) and the line between $y$ and $z$ define the axis $(0,1,0)$.  Given the distance $d_{yz}$ we place point $y$ at the point $(y_0, y_1, 0)$ where $y_0 = \cos(\sqrt{\kappa}\frac{d_{yz}}{2})$.  Next we place point $z$ at $z = (z_0, z_1,0)$ where $z_0 = y_0$ and $z_1 = -y_1$.  Given this parameterization of the manifold embedding of $y$, $z$ and the distances,  $d_{xy}, d_{xz}$, we solve for the coordinates of $x$.
    
    In general, $x = (x_0,x_1,x_2)$. Since the midpoint is placed at $(1,0,0)$, then $x_0 = \cos(\sqrt{\kappa}\frac{d_{xm}}{2})$.  Next, $x_1$ can be solved for by considering its distance to $y$: 
    \al{
        \cos(\sqrt{\kappa}d_{xy}) &= x_0y_0 + x_1y_1 \\
        &= \cos(\sqrt{\kappa}d_{xm})\cos(\sqrt{\kappa}d_{yz}/2) + x_1\sin(\sqrt{\kappa}d_{yz}/2).
    }
    Similarly the distance to $z$ can be computed as: 
    \al{
        \cos(\sqrt{\kappa}d_{xz}) &= x_0z_0 + x_1z_1 \\
        &= \cos(\sqrt{\kappa}d_{xm})\cos(\sqrt{\kappa}d_{yz}/2) - x_1\sin(\sqrt{\kappa}d_{yz}/2).
    }
    Solving for $x_1$ leads to the expression 
    \al{
        \frac{\cos(\sqrt{\kappa}d_{xy}) - \cos(\sqrt{\kappa}d_{xm})\cos(\sqrt{\kappa}d_{yz}/2)}{\cos(\sqrt{\kappa}d_{xz}) - \cos(\sqrt{\kappa}d_{xm})\cos(\sqrt{\kappa}d_{yz}/2)} &= - 1 
    }
    which can be further rearranged as follows: 
    \al{
        \cos(\sqrt{\kappa}d_{xy}) - \cos(\sqrt{\kappa}d_{xm})\cos(\sqrt{\kappa}d_{yz}/2)&= - (\cos(\sqrt{\kappa}d_{xz}) - \cos(\sqrt{\kappa}d_{xm})\cos(\sqrt{\kappa}d_{yz}/2)) \\
        \implies  \cos(\sqrt{\kappa}d_{xy}) + \cos(\sqrt{\kappa}d_{xz}) &= 2\cos(\sqrt{\kappa}d_{xm})\cos(\sqrt{\kappa}d_{yz}/2)\\
        \implies \text{Sec}(\frac{d_{yz}}{2}\sqrt{\kappa})(\cos(\sqrt{\kappa}d_{xy}) + \cos(\sqrt{\kappa}d_{xz})) &= 2\cos(\sqrt{\kappa}d_{xm}).
    }
    This finally leads to the expression: 
    \al{
         2\cos(d_{xm}\sqrt{\kappa}) - \text{Sec}(\frac{d_{yz}}{2}\sqrt{\kappa})(\cos(d_{xy}\sqrt{\kappa}) + \cos(d_{xz}\sqrt{\kappa})) &= 0. 
    }
    We then normalize this by the curvature value $\frac{1}{\kappa}$,  which allows for $g(\kappa,d)$ to be a continuous function of $\kappa$ from the hyperbolic $\kappa < 0$ to spherical space $\kappa > 0$.  

    In the spherical case, we also require that the distances themselves satisfy $d_{pq} \leq \frac{\pi}{\sqrt{\kappa}}$, due to the fact that this corresponds to the maximum possible distance on the sphere, however this restriction is not present in the $\E^p$ and $\Hy^p$.
    
    \textbf{Case 2: Hyperbolic.} The proof for deriving this method in the hyperbolic case follows an identical method to the spherical case and is left out for brevity. The resultant estimating function is of the following form 
    \al{
        &\frac{1}{-\kappa} \bigg( 2\cosh(d_{xm}\sqrt{-\kappa}) \\
        &- \text{Sech}(\frac{d_{yz}}{2}\sqrt{-\kappa})(\cosh(d_{xy}\sqrt{-\kappa}) + \cosh(d_{xz}\sqrt{-\kappa})) \bigg)= 0
    }
    \text{Remark:}
    Under the limit as $\kappa \to 0$ we find that $\lim_{\kappa \to 0} g(\kappa,\vect{d}^{\triangleline}) =  \frac{1}{2}d_{xy}^2 + \frac{1}{2}d_{xz}^2 - \frac{1}{4}d_{yz}^2 -d_{xm}^2$ which gives exactly the parallelogram law in Euclidean space.  This also highlights the necessity that the term $\frac{1}{\kappa}$ plays in maintaining a smooth equation as a function of $\kappa$ through $0$. 
\end{proof}

\subsection{Proof of Theorem~\ref{thm:asymptotic_normality}} \label{sec:asymptotic_normality}
\begin{proof}
    We will use the implicit function theorem to construct a function for which we can later apply a delta method.  In order to do so, we must ensure that $\frac{d}{d \kappa}g(\kappa, \vect{d}^{\triangleline}) \not = 0$ for $\kappa$ such that $g(\kappa, \vect{d}^{\triangleline}) = 0$.  We first note that $g(\kappa, \vect{d}^{\triangleline})$ is continuously differentiable when $\sqrt{\kappa} < \frac{\acos(\pi)}{d_{yz}}$.  This boundary corresponds to the maximum distance allowed on a sphere, as given by two anti-polar points.  Therefore to apply the implicit function theorem, what remains is 
    \al{
        \frac{\partial}{\partial \kappa} g(\kappa, \vect{d}^{\triangleline}) &\not = 0. 
    }
    which will hold by a brief application of (B3). As we have derived in Theorem~\ref{thm:midpoint_curve_equation}, $d_{xm}(\kappa; \vect{d}^{\triangle}(x,y,z))$, the length of the triangle median as a function of the triangle lengths $\vect{d}^{\triangle}(x,y,z)$ is a differentiable, continuous, increasing function of $\kappa$.

    By definition, this function satisfies
    $$g(\kappa, d_{xy}, d_{xz}, d_{yz}, d_{xm}(\kappa; \vect{d}^{\triangle}(x,y,z))) = 0$$

    Let $d_{xm}$ denote the fixed value of the triangle median length at value of the true curvature $\kappa$, and let $d_{xm}(\kappa; \vect{d}^{\triangle}(x,y,z)))$ denote the median length as a function of the triangle side length. We can now write the estimating function $g$ as a function of the equivalent exact midpoint $d_{xm}(\kappa) := d_{xm}(\kappa; \vect{d}^{\triangle}(x,y,z)) $\footnote{We drop triangle lengths dependence $\vect{d}^{\triangle}(x,y,z)$ for brevity}
    \al{
        g(\kappa, d_{xy}, d_{xz}, d_{yz}, d_{xm}) &= \frac{2\cos(d_{xm}\sqrt{\kappa})}{\kappa} - \frac{2\cos(d_{xm}(\kappa)\sqrt{\kappa})}{\kappa}
    }
    We compute the derivative as a function of $\kappa$

    \al{
        &\frac{d}{d\kappa'}g(\kappa'; d_{xy}, d_{xz}, d_{yz}, d_{xm}) \bigg|_{\kappa' = \kappa} \\
        &= 2\frac{ - d_{xm}\sqrt{\kappa}\sin(d_{xm}\sqrt{\kappa}) + 2 \cos(d_{xm}\sqrt{\kappa})}{2\kappa^2} \\
        &+ 2\frac{(d_{xm}(\kappa)\sqrt{\kappa} + \kappa^{3/2}d'_{xm}(\kappa))\sin(d_{xm}\sqrt{\kappa}) - 2 \cos(d_{xm}(\kappa)\sqrt{\kappa})}{2\kappa^2} 
    }
    Since $g(\kappa,\vect{d}^{\triangleline}) = 0 \implies d_{xm} = d_{xm}(\kappa)$, 
     \al{
        &\frac{d}{d\kappa'}g(\kappa'; d_{xy}, d_{xz}, d_{yz}, d_{xm}) \bigg|_{\kappa' = \kappa} \\ 
        &= 2d'_{xm}(\kappa)\frac{\sin(d_{xm}(\kappa)\sqrt{\kappa})}{\sqrt{\kappa}}
    }
    Since $\kappa \leq \frac{\pi^2}{d_{xm}^2}$ then if $d'_{xm}(\kappa) > 0$ then $\frac{d}{d\kappa} g(\kappa, \vect{d}^{\triangleline}) > 0$ and hence $g(\kappa, \vect{d}^{\triangleline})$ has positive derivative $\frac{d}{d\kappa} g(\kappa, \vect{d}^{\triangleline}) > 0$ for $\kappa: g(\kappa; \vect{d}^{\triangleline}) = 0$, where we use the notation $g(\kappa; \vect{d}^{\triangleline})$ when we are specifying the $g$ as a function of $\kappa$ for a fixed set of distances. 
    
    Next, by the implicit function theorem, there exists an open neighborhood $\mathcal{\tilde D} \in \R^4$ and an ``implicit" function $\tilde f(d)$ such that: $\kappa = \tilde f(d)$ and $g(\tilde f(d), d) = 0$ for $d \in \mathcal{\tilde D}$.  Furthermore the gradient satisfies: 
    \al{
        \nabla_d \tilde f(d) &= - \left(\frac{\partial g(\kappa, \vect{d}^{\triangleline})}{\partial \kappa}\right)^{-1}\left[ \nabla_d g(\kappa, \vect{d}^{\triangleline})\right]. 
    }
    Therefore by the delta method we arrive at the asymptotic distribution of $\hat \kappa$
    \al{
        \sqrt{r(n)}(\hat \kappa - \kappa)&= \sqrt{r(n)}(\tilde f(\hat d) - \tilde f(d)) \to_{D} N(0, \nabla  \tilde f(d)^\top  \Sigma \nabla  \tilde f(d))
    }
    where $\to_{D}$ refers to convergence in distribution.  Lastly, by the implicit function theorem: 
    \al{
        \nabla  \tilde f(d)^\top  \Sigma \nabla  \tilde f(d) &= \left(\frac{\partial g(\kappa, \vect{d}^{\triangleline})}{\partial \kappa}^2\right)^{-1}\left[ \nabla_d g(\kappa, \vect{d}^{\triangleline})^\top \Sigma \nabla_d g(\kappa, \vect{d}^{\triangleline})\right]. 
    }
    If instead we only have consistency, rather than asymptotic normality of $\hat d$, $\norm{\hat d - d}_2 = o_P(r(n)^{-1/2})$, then we can use the continuous mapping theorem instead and we have $|\hat \kappa - \kappa| =  o_P(r(n)^{-1/2})$.  
    \textbf{Remark:} This form is very similar to the usual standard asymptotic normality proofs for $Z$ estimators as in \citet{vanderVaart1998AsymptoticStatistics}.  However, the main difference is in the fact that we are not averaging the estimating function, but rather plugging in a distance estimate which is asymptotically normal, meaning that we develop this delta method argument for a plug-in estimator. 
    
    We find that in practice, condition (B3) holds quite generally, unless the points $x,y,z$ are co linear.  In fact, in the Euclidean case, we can derive this exactly according to the Taylor series expansion at $\kappa = 0$
    
    \al{
        g(\kappa, \vect{d}^{\triangleline}) &=  \left(\frac{d_{xy}^2}{2} + \frac{d_{xz}^2}{2} - d_{xm}^2 - \frac{d_{yz}^2}{4}\right) + \left(\frac{d_{xm}^2}{12} - \frac{5d_{yz}^2}{192} + \frac{1}{16}d_{yz}^2(d_{xy}^2 + d_{xz}^2) - \frac{1}{24}(d_{xy}^2 + d_{xz}^2) \right)\kappa + \OO(\kappa^2). 
    }    
    Clearly at $\kappa = 0 $ the solution reduces to the parallelogram law $(\frac{d_{xy}^2}{2} + \frac{d_{xz}^2}{2} - d_{xm}^2 - \frac{d_{yz}^2}{4} = 0 )$.  When we substitute in the corresponding solution at $\kappa = 0$
    \al{
        \frac{\partial }{\partial \kappa }g(\kappa, \vect{d}^{\triangleline})\Bigg|_{\kappa = 0} &= -\frac{1}{48}\left(d_{yz}^4 - 2d_{yz}^2(d_{xz}^2 +(d_{xy}^2) + (d_{xz}^2 - d_{xy}^2)^2 \right) \\
        &= -\frac{1}{48}\bigg((d_{yz} + d_{xz} +d_{xy})(d_{yz} - d_{xz} - d_{xy})\\
        &\times(d_{yz} - d_{xz} +d_{xy})(d_{yz} + d_{xz} - d_{xy}) \bigg). 
    }
    As long as the triangle inequality is satisfied strictly for $x,y,z$, then $\frac{\partial }{\partial \kappa }g(\kappa; \vect{d}^{\triangleline})\Bigg|_{\kappa = 0} > 0$.  However, if the triangle inequality is not strict, i.e. $x,y,z$ are co-linear, and $d_{ij} = d_{ik} + d_{kj}$ for $(i,j.k)$ being some permutation of $(x,y,z)$, then $\frac{\partial }{\partial \kappa }g(\kappa; \vect{d}^{\triangleline})\Bigg|_{\kappa = 0} = 0$. For $\kappa \not = 0$, we do not have simple closed-form expressions for $\frac{\partial }{\partial \kappa }g(\kappa; \vect{d}^{\triangleline})\Bigg|_{\kappa = 0}$ and thus we leave (B3) as an assumption, however, we believe this pattern to hold in the other canonical manifolds. 

\end{proof}

\begin{lemma} \label{lem:non_decreasing_midpoint_length}
    If $x,y,z \in \M^p(\kappa)$,  let $m$ denote the midpoint between $y$ and $z$.  Then let $d_{xm}(\kappa; \vect{d}^{\triangle}(x,y,z))$ denote the distance to the midpoint $m$ from $x$ as a function of the curvature of the latent space $\kappa$. \\

    Then $ d_{xm}(\kappa; \vect{d}^{\triangle}(x,y,z))$ is a non-decreasing function in $\kappa$.  
\end{lemma}
    \begin{proof}
        The proof follows from Toponogov's triangle theorem and the negative curvature extension \citep{berger1962extension} which we include immediately following this proof in Theorem~\ref{thm:toponogov}.  Abbreviated for our purposes in Theorem~\ref{thm:toponogov}, we simply let $p = z, x = r, q = m$ and compare two manifolds $\kappa \leq \kappa'$. 

        Then it immediately follows $$d_{xm}(\kappa; \vect{d}^{\triangle}(x,y,z)) \leq d_{xm}(\kappa'; \vect{d}^{\triangle}(x,y,z)).$$
    \end{proof}

\begin{theorem}[Toponogov's Triangle Comparison Theorem (Theorem 3 of \citet{berger1962extension}] \label{thm:toponogov}
    Let $\mathcal{M}$ be a complete Riemannian manifold, and let $\M(\kappa)$ be the simply connected manifold of constant curvature $\kappa \in \R$.  Let $\triangle_{zyx}$ denote a geodesic triangle with points $z,y,x \in \M$.  Let $\tilde \triangle_{zyx}$ is the geodesic triangle with side lengths $d_{zy} = \tilde d_{zy}$ and $d_{zx} = \tilde d_{yx}$ on $\M(\kappa)$. If $\underline{K}(\mathcal{M})$ is the minimum Riemannian sectional curvature on the manifold, and if $\kappa \leq \underline{K}(\mathcal{M})$, then 
    $$ \tilde d_{yx} \leq  d_{qr}.$$
\end{theorem}

This theorem suggests that assumption (B3) is therefore very mild and that we will always have a non-decreasing $d_{xm}$ function of $\kappa$ but that we only have to assume that the increasingness is strict.  In practice, we observe this is only not strict when the points $x,y,z$ are co-linear, and therefore the distance to the midpoint does not change as a function of the curvature since, effectively, these points lie along a single geodesic.

\subsection{Proof of Theorem~\ref{thm:curve_bounds}}
\begin{proof}
    Recall that $\vect{d}^{\triangle}(x,y,z)$ denotes the vector of the distances of the triangle with vertices $(x,y,z)$, and let $\vect{d}^{\triangle}(m',y,z)$ be the distances in triangle $(m', y, z)$.  We also measure the distance to the surrogate midpoint, $d_{xm'}$.  Though we do not have access to the distance to the true midpoint $d_{xm}$, this is going to be a function of the curvature of the space, and the distances of the triangle $(x,y,z)$ and the median length $d_{xm}(\kappa; \vect{d}^{\triangle}(x,y,z))$.  
    For any curvature $\kappa$ and a surrogate midpoint $m'$, so that $m',x,y,z, \in \M^p(\kappa)$, we can upper and lower bound the triangle median length.  By the triangle inequality these upper and lower bounds are  
    \al{
        d_{xm}(\kappa; \vect{d}^{\triangle}(x,y,z)) &\leq d^{+}_{xm}(\kappa) := d_{xm'} + d_{m' m}(\kappa; \vect{d}(m',y,z)) \\
        d_{xm}(\kappa; \vect{d}^{\triangle}(x,y,z)) &\geq d^{-}_{xm}(\kappa) := d_{xm'} - d_{m' m}(\kappa; \vect{d}^{\triangle}(m',y,z))
    }
    By Lemma~\ref{lem:non_decreasing_midpoint_length} each $d_{xm}$ function is monotone in $\kappa$.  The upper and lower bounds $\kappa^{\pm}$ can then be computed by plugging in this value to $g$, i.e. $g(\kappa,\vect{d}^{\triangleline, \pm}(\kappa)) = 0$ and solving for $\kappa$. By this monotonicity, the number of solutions in $\kappa$ to either equation will typitcally be $\{0,1\}$.  
    

    If there is a solution to $g(\kappa,\vect{d}^{\triangleline, \pm}(\kappa)) = 0$, then we have found and upper and lower bounds. If there are no solutions, then the upper and lower bounds are $\pm \infty$. In the event that there are colinear points, then this can lead to a situation where $(\kappa,\vect{d}^{\triangleline, \pm}(\kappa)) = 0$ has multiple solutions in $\kappa$.  In this case, we can take the minimum of the upper bounds and maximum of the lower bounds respectively.  
\end{proof}

\subsection{Proof of Theorem~\ref{thm:midpoint_convergence} } \label{sec:proof_midpoint_convergence}

Before proving Theorem~\ref{thm:midpoint_convergence}, we first introduce a useful theorem for the proof. 
\begin{theorem}\label{thm:copula_theorem}
    Consider a set of continuous random variables $X_i \in \R$ which are marginally identical but have an arbitrary dependency.  Let $\mathbb{\hat  M}[X]$ denote the sample median of the set of these random variables. Then 
    \begin{equation}
        P(\mathbb{\hat  M}[X] \leq t) \leq 2 P(X \leq t)
    \end{equation}
\end{theorem}

We find the proof in the supplementary materials in Section \ref{sec:proof_copula_theorem}.  This bound is only useful up to $t = \mathbb{ M}[X]$ at which point the upper bound is $1$.  The implication here is that we only need most of the midpoints between any two points to be reasonably separated. We later illustrate how we will find good midpoints and so we can verify their existence in any observed dataset. 

\begin{proof}
    Let $\Xi(Z_{\{1,\dots, i\}})$ denote the set of distances from one midpoint to another as a function of their endpoints:  $Z_{j}, Z_{k}, Z_{l}, Z_{r}$, i.e. $\Xi(Z_{\{1,\dots, i\}}) = \{d(m(Z_j,Z_k), m(Z_l,Z_r))\}_{(j,k,l,r) \in \{1,\dots, i\}}$  where $m(Z_j,Z_k)$ refers to the midpoint of $Z_j$ and $Z_k$. 
    Denote the set of events $M_i(t) = \mathbb{\hat M}[\Xi(Z_{\{1,\dots, i\}})] > t$ the median of these distances after sampling $i$ endpoint positions is at least $t$, and let $G := G(t) = \cap_{i = 1}^K M_i(t)$ be the intersection of a growing sequence of these medians all satisfy this bound.  We can upper bound the probability that $\Phi(\D_K) \leq t$, that the smallest distance to the midpoint of the pair of endpoints by the following induction argument:
    \al{
        &P(\Phi(\D_K) \leq t) \\
        &= P(\Phi(\D_K) \leq t|\Phi(\D_{K - 1}) \leq t, G)P(\Phi(\D_{K - 1}) \leq t|G)P(G) \\
        &+ P(\Phi(\D_K) \leq t|\Phi(\D_{K - 1}) > t, G)P(\Phi(\D_{K - 1}) > t|G)P(G) \\
        &+ P(\Phi(\D_K) \leq t| G^c)P(G^c) \\
        &\leq \underbrace{P(\Phi(\D_K) \leq t|\Phi(\D_{K - 1}) \leq t, G)}_{(I)}P(\Phi(\D_{K - 1}) \leq t|G) + \underbrace{P(G^c)}_{(II)}. 
    }
    We first consider $(I)$.  Since we condition on $G$, at each step at least half of the midpoints are at a distance at least $t$ away from one another. By the locally Euclidean assumption (A2) and the continuity of the latent distribution on a convex region (C1), then since there are at least $\binom{K}{2}/2$ midpoints at a distance at least $t$ away from one another: 
    \al{
        (I) &= P(\Phi(\D_K) \leq t|\Phi(\D_{K - 1}) \leq t, G) \\
        &\leq \max\{0,1 - \binom{K}{2}\alpha\frac{C_p}{2}t^p\} \\
        &\leq \exp(-\binom{K}{2}\alpha\frac{C_p}{2}t^p)
    }
    By recursion, we can bound this over a growing sequence of samples of midpoints $k \in \{1,2,\dots, K\}$
    \al{
        \prod_{k = 1}^K P(\Phi(\D_K) \leq t|\Phi(\D_{K - 1}) > t, G) &\leq \exp(-\sum_{k = 1}^K\binom{k}{2}\alpha\frac{C_p}{2}t^p) \\
        &= \exp(-\binom{K + 1}{3}\alpha\frac{C_p}{2}t^p)\\
        &\leq \exp(-(K + 1)^3\alpha\frac{C_p}{12}t^p)
    }
    
    Secondly, let us consider term $(II)$. We will also demonstrate that this term is generally negligible compared to the first term. By a union bound arguement: 
    \al{
        P(G^c) &= P(\cup_{k = 1}^K \{\mathbb{\hat M}[\Xi(Z_{\{1,\dots, k\}})] \leq t\}) \\
        &\leq \sum_{k = 1}^K P( \mathbb{\hat M}[\Xi(Z_{\{1,\dots, k\}})] \leq t) \\
        &\leq K(P( \mathbb{\hat M}[\Xi(Z_{\{1,\dots, k\}})] \leq t)) 
    }

    We note that this in fact describes a set of marginally identical, yet correlated random variables, allowing up to leverage Theorem~\ref{thm:copula_theorem}. 
    \al{
        P( \{\mathbb{\hat M}[\Xi(Z_{K})] \leq t\}) &\leq 2P( \Xi(Z_{K}) \leq t\}) \\
        &= 2P( \cup_{m' \not = m}d(m,m') \leq t\}) \\
        &\leq 2\sum_{m' \not = m}P( d(m,m') \leq t\})
    }
    And by assumption (C3) in our theorem $P( d(m,m') \leq t) \leq \alpha_m C_p t^{p}$. Since there are $K\binom{K}{2}$ terms in this bounds, then we add a factor of $K\binom{K}{2}$ to construct the bound for (II). Hence we combine these bounds to obtain: 
    \al{
         P(\Phi(\D_K) \leq t) &\leq \exp(-(K + 1)^3\alpha_m\frac{C_p}{12}t^p) + 2\alpha_m C_pK\binom{K}{2}t^p 
    }
    Letting $t = \frac{C_2}{K^{3/p}}$ will bound $\lim_{K\to \infty} P(\Phi(\D_K) \leq t_K)$ by any constant and therefore 
    \begin{equation*}
        \Phi(\D_K) = \OO_P(K^{-3/p}).
    \end{equation*}
\end{proof}

\subsection{Proof of Lemma~\ref{lem:clique_localization_rate}} \label{sec:proof_clique_localization_rate}

The proof draws on a similar structure to that of the consistency result of \citet{lubold2023identifying}, we extend this to illustrate the rate of convergence.  

\begin{proof}
    Let $\epsilon_\delta := \{\max_{i,j} d(Z_i,Z_j) < \delta\}$. We wish to show that: 
    \al{
         \lim_{\ell \to \infty}\frac{P(\epsilon_{\delta}|C_\ell)}{P(\epsilon^C_{\delta}|C_\ell)} &\to \infty
    }
    for $\delta = O(\exp(-\tilde \mu_d \ell/p))$. 

    Firstly we derive the probability of a clique forming, conditional on  $\epsilon_\delta$.  For convenience, we denote $\bar{\nu} = \frac{1}{\ell} \sum_{i = 1}^\ell \nu_i$ and $\bar{d}(\vect{z} = \binom{\ell}{2}^{-1} \sum_{i  < j}^\ell d(Z_i, Z_j)$
    \al{
        P(C_\ell|\epsilon_{\delta}) 
        &= \int \int_{\epsilon_{\delta}}\exp(-\binom{\ell}{2}\bar d (z))\exp\left(-\binom{\ell}{2}\bar \nu\right)dP_{Z|\epsilon_{\delta}}(z)dP_{\vect{\nu}}(\nu)\\ 
        &= \Psi(P_{\nu})\int_{\epsilon_{\delta}}\exp(-\binom{\ell}{2}\bar d (z))dP_{Z|\epsilon_{\delta}}(z) \\
        \text{where }  &\Psi(P_{\nu}) = \int \exp\left(-\binom{\ell}{2}\bar \nu\right)dP_{\vect{\nu}}(\nu). 
    }
    Here we factor out the dependence on the random effects into the multiplicative term $\Psi(P_{\nu})$
    This can be developed analogously for $P(C_\ell|\epsilon_{\delta}^c)$: 
    \al{
        P(C_\ell|\epsilon_{\delta}^c) 
        &= \int \int_{\epsilon_{\delta}}\exp(-\binom{\ell}{2}\bar d (z))\exp\left(-\binom{\ell}{2}\bar \nu\right)dP_{Z|\epsilon_{\delta}}(z)dP_{\vect{\nu}}(\nu)\\ 
        &= \Psi(P_{\nu})\int_{\epsilon_{\delta}^c}\exp(-\binom{\ell}{2}\bar d (z))dP_{Z|\epsilon_{\delta}^c}(z) 
    }
    This allows us to ignore the dependence on the random effects when taking the ratio 
    \al{
        \frac{P(\epsilon_{\delta}|C_\ell)}{P(\epsilon^C_{\delta}|C_\ell)} &= \frac{P(C_\ell|\epsilon_{\delta})}{P(C_\ell|\epsilon^C_{\delta})}\frac{P(\epsilon_{\delta})}{P(\epsilon_{\delta}^c)}. 
    }
    
    We now proceed with the remainder of the proof.  Consider the ball of radius $\frac{\delta}{2}$ at a point $q$ $B(\delta, q)$. This describes one set of points for which all points have a maximal distance of $\delta$. Thus $\left\{\{Z_{i}\}_{i  = 1}^\ell \in B(\delta, q) \right\}\subset \epsilon_{\delta}$ for any $q$.  Since $f$ is continuous, then for some $\tilde \delta$ if $f(q)$ has positive probability $f(q) > 0 $ for some $q$ then for all $\delta \leq \tilde \delta$, $f(x) \geq c_2; \forall x \in B(q,\delta)$ for some constant $c_2$.  Therefore $P(\epsilon_{\delta}) \geq P(\{Z_{i}\}_{i  = 1}^\ell \in B(\delta, q)) \geq  \left(c_2\left(\frac{\delta}{2}\right)^p\right)^\ell$.  

    Next we exploit Lemma A.1 in \cite{lubold2023identifying} states that if $Z_i$ are drawn $iid$ from a latent distribution with funute mean $\mu_d$, then
    \al{
        P(C_\ell|\epsilon_{\delta}^c) &\leq \exp\left(-\binom{\ell}{2} \mu_d\right)
    }
    Putting this all together, for $P(\epsilon_{\delta}) \leq \frac{1}{2}$
    \al{
        \frac{P(\epsilon_{\delta}|C_\ell)}{P(\epsilon^C_{\delta}|C_\ell)} &= \frac{P(C_\ell|\epsilon_{\delta})}{P(C_\ell|\epsilon^C_{\delta})} \frac{P(\epsilon_{\delta})}{1 - P(\epsilon_{\delta})} \\
        &\geq \frac{P(C_\ell|\epsilon_{\delta})}{P(C_\ell|\epsilon^C_{\delta})} \frac{1}{2}P(\epsilon_{\delta}) \\
        &\geq \frac{1}{2}\frac{P(C_\ell|\epsilon_{\delta})}{P(C_\ell|\epsilon^C_{\delta})} c_2^{\ell}\left(\frac{\delta}{2}\right)^{p\ell} \\
        &\geq \frac{1}{2}\frac{\exp(-\binom{\ell}{2}\delta)}{P(C_\ell|\epsilon^C_{\delta})} c_2^{\ell}\left(\frac{\delta}{2}\right)^{p\ell} \\
        &\geq \frac{1}{2}\frac{\exp(-\binom{\ell}{2}\delta)}{\exp(-\binom{\ell}{2}\mu_d)} c_2^{\ell}\left(\frac{\delta}{2}\right)^{p\ell} \\
        &= \frac{1}{2}\exp\left( -\binom{\ell}{2}\delta + \binom{\ell}{2}\mu_d + \ell\log(c_2) + \log(\delta/2)p\ell \right) 
    }
    
    Therefore letting $\delta = 2\exp(-\tilde \mu_{d}\frac{\ell - 1}{p})$ for $\tilde \mu_d < \mu_d$.  Then
    \al{
        \lim_{\ell \to \infty} \frac{P(\epsilon_{\delta}|C_\ell)}{P(\epsilon^C_{\delta}|C_\ell)} \to \infty
    }
    and therefore the proof is complete. 
\end{proof}

\subsection{Proof of Lemma~\ref{lem:rand_eff_clique}}
At a high level, the proof structure is nearly identical to Lemma~\ref{lem:clique_localization_rate}, however, we swap the roles of $\nu$ and $d(Z_i,Z_j)$. 
\begin{proof} 
    
    Let $\varepsilon_{\delta} = \left\{ \min_{i \in \{1,2,\dots, \ell\}} \nu_i \geq -\delta \right\}$
    
    If $f_{\nu}(\nu)$ is continuous around $0$ then for small enough $\delta, F_{\nu}(\nu > - \delta) \geq c_3 \delta$ for some $c_3$. 
    As in the proof of Lemma~\ref{lem:clique_localization_rate}, we similarly take the ratio and integrate out $\Psi(P_{\bar{d}})$ since 
    \al{
        P(C_\ell|\varepsilon_{\delta}) 
        &= \int_{\varepsilon_{\delta}}\int \exp(-\binom{\ell}{2}\bar d (z))\exp\left(-\binom{\ell}{2}\bar \nu\right)dP_{Z}(z)dP_{\vect{\nu}|\varepsilon_{\delta}}(\nu) \\
        &= \Psi(P_{\bar{d}}) \int_{\varepsilon_{\delta}}\exp\left(-\binom{\ell}{2}\bar \nu\right)dP_{\vect{\nu}|\varepsilon_{\delta}}(\nu)  \\
        \text{where } \Psi(P_{\bar{d}}) &= \int \exp(-\binom{\ell}{2}\bar d (\vect{z})) dP_{Z}(\vect{z})
    }
    
    Then consider the set $[-\delta, 0]$. Then  $P(\varepsilon_{\delta}) \geq \left(c_3\left(\delta\right)\right)^\ell$.  
    
    Next we consider the probability of a clique, conditional on the fact that the locations did not occur in set $\varepsilon^{c}_{\delta}$. 
    \al{
        P(C_\ell|\varepsilon^{c}_{\delta}) &\leq \int \exp\left( 2\binom{\ell}{2}\bar{\nu}\right)dP_{\vect{\nu}}(\vect{\nu}).
    }
    Let $\mu_{\nu} = \E[\nu] < 0$. 
    This is due to the fact $\varepsilon_{\delta}$ denotes the event $\{\min_{i} \nu_i  < - \delta\}$ it is reasonable that $P(C_\ell|\epsilon^c_{\ell}) \geq P(C_\ell)$ since we are excluding the events most likely to create a clique.  
    Furthermore by the same argument as in Lemma A.1 in \cite{lubold2023identifying} we can show
    \al{
        P(C_\ell|\varepsilon^{c}_{\delta}) &\leq \exp(-2\binom{\ell}{2} |\mu_{\nu} )
    }

    This can be achieved as follows: 
    \al{
        P(C_\ell|\varepsilon^{c}_{\delta}) &= E \left[ \prod_{i < j}\exp(\nu_i + \nu_j) \bigg |\varepsilon_\delta^c \right] \\
        &\leq E \left[ \prod_{i < j} \exp(\nu_i + \nu_j)  \right] \\
        &\leq \prod_{i < j}\left(E \left[ \exp(\binom{\ell}{2}(\nu_i + \nu_j))  \right]\right)^{1/\binom{\ell}{2}} \\
        &\leq \exp(\binom{\ell}{2}2\mu_\nu) \prod_{i < j} \left(E \left[ \exp(\binom{\ell}{2}(\eta_i + \eta_j))  \right]\right)^{1/\binom{\ell}{2}} \\
        &\leq \exp(\binom{\ell}{2}2\mu_\nu) \times 1 \\
        &= \exp(- \binom{\ell}{2} |\mu_\nu|)
    }
   due to the fact that the probability of connection is higher within $\varepsilon_{\delta}$ than outside of $\varepsilon^{c}_{\delta}$.  The remaining steps follow from an application of Holder's generalized inequality and $x_{i} = \mu_\nu + \eta_i$ for some noise $\eta_i$. 

   Returning to the computation of the ratio of probabilities,
    \al{
        \frac{P(\varepsilon_{\delta}|C_\ell)}{P(\varepsilon^C_{\delta}|C_\ell)} &= \frac{P(C_\ell|\varepsilon_{\delta})}{P(C_\ell|\varepsilon^C_{\delta})} \frac{P(\varepsilon_{\delta})}{1 - P(\varepsilon_{\delta})} \\
        &\geq \frac{1}{2} \frac{P(C_\ell|\varepsilon_{\delta})}{P(C_\ell|\varepsilon^C_{\delta})} P(\varepsilon_{\delta}) \\
        &\geq \frac{1}{2}\frac{P(C_\ell|\varepsilon_{\delta})}{P(C_\ell|\varepsilon^C_{\delta})} c_3^{\ell}\left(\delta\right)^{\ell} \\
        &\geq \frac{1}{2}\frac{\exp(-2\binom{\ell}{2}\delta)}{P(C_\ell|\varepsilon^C_{\delta})} c_3^{\ell}\left(\delta\right)^{\ell} \\
        &\geq \frac{1}{2}\exp\left( -2\binom{\ell}{2}\delta + 2\binom{\ell}{2}|\mu_\nu| + \ell\log(c_3) + \log(\delta)\ell \right) 
    }
    
    therefore we can let $\delta = \exp(- |\tilde \mu_\nu|2(\ell - 1) )$ for any $|\tilde \mu_\nu| < |\mu_\nu|$
    and therefore the proof is complete. 
    
\end{proof}

\subsection{Proof of Theorem~\ref{thm:asymptotic_distance}} \label{sec:proof_asymptotic_distance}

In order to prove the derive the asymptotic distribution of the distance estimator, we introduce a useful theorem.  This theorem will illustrate the rate of estimation of the random effects.

\begin{lemma} \label{lem:randeff_consistency}
    Let $W \subset \{1,2,\dots, n\}$ denote a subset of indices which form a clique ($C_\ell$).  Let $d_i$ denote the degree of node $i \in W$ where $|W| = \ell$ and assume that the points in the clique have a common latent position.  Denote the estimator of $\gamma_\mathpzc{W}$.  Let $\mu_{\nu} := E[\nu]$.  Then for any $\mu_{\nu} < \tilde \mu_{\nu} < 0 $
    \al{
        \hat \gamma_\mathpzc{W} &= \log\left( \frac{1}{\ell}\sum_{i \in W} \frac{d_i}{\max_{j \in W} d_{j}}\right).
    }
    Then if (E1) in Lemma~\ref{lem:rand_eff_clique} holds, then define $\tilde \mu_{\nu}$ as in Lemma~\ref{lem:rand_eff_clique}
    \begin{equation}
        \hat \gamma_\mathpzc{W} - \gamma_\mathpzc{W} = \OO_P\left(\max\left\{\exp(\tilde \mu_{\nu} \ell), \frac{1}{\sqrt{n}}\right\}\right)
    \end{equation}

\end{lemma}

Therefore, estimation of $\gamma_\mathpzc{W}$ within a clique can occur at an exponential rate, meaning that this estimation will be negligible compared to the average cross-clique probabilities.

\begin{proof}
    The proof here is a straightforward application of the plug in estimator of $\hat p_{\mathpzc{X}\mathpzc{Y}}$ and the set of random effects $\hat \gamma_{X/Y}$. Then by the Lindeberg-Feller CLT:
    \al{
        \sqrt{\ell^2 \sigma_\ell}(\hat p_{\mathpzc{X}\mathpzc{Y}} - p_{\mathpzc{X}\mathpzc{Y}}) &\to_d N\left(0, 1\right)
    }
    We first consider the localization of nodes within a clique, specifically
    \al{
         p_{\mathpzc{X}\mathpzc{Y}} &= \frac{1}{\ell^2}\sum_{x \in \mathpzc{X}}\sum_{y \in \mathpzc{Y}} \exp(\nu_x + \nu_y - d_{xy}) \\
         &= \frac{1}{\ell^2}\sum_{x \in \mathpzc{X}}\sum_{y \in \mathpzc{Y}} \exp(\nu_x + \nu_y - d_{\mathpzc{X}\mathpzc{Y}} + o_P(-\tilde \mu_d \ell)) \text{ by Lemma~\ref{lem:clique_localization_rate}} \\
         \implies d_{\mathpzc{X}\mathpzc{Y}} &= -\log(p_{\mathpzc{X}\mathpzc{Y}}) + \gamma_{X} + \gamma_{Y} + o_P(-\tilde \mu_d \ell)
    }
    We now study our distance estimator as per equation~\eqref{eq: distance estimate}
    \al{
        \hat d_{\mathpzc{X}\mathpzc{Y}} &= -\log(\hat p_{\mathpzc{X}\mathpzc{Y}} ) + \hat \gamma_{X} +  \hat \gamma_{Y} + o_P(\tilde \mu_d \ell) + o_P(-\tilde \mu_d \ell)\\
        \sqrt{\ell^2}\hat d_{\mathpzc{X}\mathpzc{Y}} &= -\sqrt{\ell^2}\log(\hat p_{\mathpzc{X}\mathpzc{Y}} ) + \sqrt{\ell^2}\hat \gamma_{X} +  \sqrt{\ell^2}\hat \gamma_{Y} \\
        &= -\sqrt{\ell^2}\log(\hat p_{\mathpzc{X}\mathpzc{Y}} ) + \gamma_{X} + \gamma_{Y} + o_P(\ell n^{-1/2}) +  o_P((\ell + 1) \exp(-\tilde \mu_d \ell)) +  o_P(\ell (-\tilde \mu_\nu \ell)) \text{ By Lemma~\ref{lem:rand_eff_clique}} \\
        &= -\sqrt{\ell^2}\log(\hat p_{\mathpzc{X}\mathpzc{Y}} ) + \gamma_{X} + \gamma_{Y} + o_P(1). 
    }
    And hence the higher order terms are asymptotically negligible compared to the estimation of $p_{\mathpzc{X}\mathpzc{Y}}$. Therefore by the delta method we can derive the asymptotic distribution of the distance estimator
    \al{
        \sqrt{\ell^2 \frac{\sigma_{\ell}}{p_{\mathpzc{X}\mathpzc{Y}}}}\left(\hat d_{\mathpzc{X}\mathpzc{Y}} - d_{\mathpzc{X}\mathpzc{Y}} \right) &\to_d N\left(0, 1\right) 
    }
    
\end{proof}

\subsection{Proof of Lemma~\ref{lem:randeff_consistency}}
\begin{proof}
    Consider a clique $\mathpzc{X}$ of size $\ell$.  We note that by Lemma~\ref{lem:clique_localization_rate} that the nodes within the clique are exponentially close to one another, and $d_{xx'} = o_P(\exp(-\tilde \mu_{d} \ell))$.
    Therefore we consider the ratio of the connection probability to any node in the network $P(A_{xi} = 1|\mathpzc{X})/P(A_{x'i} = 1|\mathpzc{X})$, where 
    \al{
        P(A_{xi} = 1|\mathpzc{X}) &= \int \exp(\nu_x + \nu_i - d(Z_x, Z_i)) dF_\nu(\nu_i)dF_Z(Z_i) \\
        &= \exp(\nu_x) \int  \exp(\nu_i - d(Z_x, Z_i)) dF_\nu(\nu_i)dF_Z(Z_i) \\
        &= \exp(\nu_x) \int  \exp(\nu_i - d(Z_{x'}, Z_i) + o_P(\exp(-\tilde \mu_d \ell))) dF_\nu(\nu_i)dF_Z(Z_i) \\
        \implies P(A_{xi} = 1|\mathpzc{X})/P(A_{x'i} = 1|\mathpzc{X}) &= \exp(\nu_x - \nu_{x'}) \exp(o_P(\exp(-\tilde \mu_d \ell)))) \\
        &=  \exp(\nu_x - \nu_{x'}) + o_P(\exp(-\tilde \mu_d \ell)))
    }
    Next we note that the empirical ratio of the connection probabilities can be computed using a ratio of degrees 
    \al{
        \frac{\widehat{P}(A_{xi} = 1|\mathpzc{X})}{\widehat{P}(A_{x'i} = 1|\mathpzc{X})} &= \frac{d_x/n}{d_{x'}/n} \\
    }
    By a simple application of Hoeffding's inequality, we see that 
    \al{
        |\widehat{P}(A_{xi} = 1|\mathpzc{X}) - P(A_{xi} = 1|\mathpzc{X})| &= O_P(1/\sqrt{n})
    }
    And therefore
    \al{
        \frac{\widehat{P}(A_{xi} = 1|\mathpzc{X})}{\widehat{P}(A_{x'i} = 1|\mathpzc{X})}  &= \frac{{P}(A_{xi} = 1|\mathpzc{X})}{{P}(A_{x'i} = 1|\mathpzc{X})} + O_P(1/\sqrt{n})
    }
    Lastly we apply this to the estimation of $\gamma_\mathpzc{X}$ where. 
    \al{
        \gamma_\mathpzc{X} &= \log\left(\frac{1}{\ell}\sum_{x \in \mathpzc{X}}\exp(\nu_i) \right)
    }
    By Lemma~\ref{lem:randeff_consistency} we note that the random effects converge to $0$ within a clique, therefore setting the smallest magnitude estimate to $0$ we finally observe that the final rate for estimating the random effects is as follows: 
    \al{
        \widehat{\gamma}_{\mathpzc{X}} - {\gamma}_{\mathpzc{X}} &= O_P(1/\sqrt{n}) + o_P(\exp(-\tilde \mu_d \ell))) + o_P(\exp(-|\tilde \mu_\nu| \ell))).
    }
\end{proof}

\subsection{Proof of Theorem~\ref{thm:copula_theorem}} \label{sec:proof_copula_theorem}

In order to prove Theorem~\ref{thm:copula_theorem} we first introduce a series of lemmas that will be useful for proof.  The proofs of these are contained in the subsequent subsections of the appendix.

Let $X_i$ be a random variable with marginal distribution function $F$, and let $X = (X_1, X_2, \dots, X_n)$ denote a random vector for which all marginal distributions are $F$, that may in general be correlated.  Let $H$ denote the joint distribution. By Sklar's Theorem~\citep{SKLAR1959FonctionsMarges, Durante2013ATheorem} we can express the joint distribution as
\al{
    H(\vect{x}) &= P(X_1 \leq x_1, X_2 \leq x_2 , \dots, X_n \leq x_n) \\
    &= \C(F(x_1), F(x_2), \dots, F(x_n))
}
where $\C$ is a copula for the joint distribution. A copula is simply a multivariate distribution function for a set of uniform random variables which explains their dependence.  

\begin{lemma} \label{lem:med_uniform_marginal}
    For a set of random variables $X_i$ with equal marginal distribution, with arbitrary correlation structure denoted by a copula $\C$, then $\mathbb{\hat  M}(X) = F^{-1}(\mathbb{\hat  M}(U))$ where $U \sim \C$ is a multivariate uniform random variable $U = (U_1, U_2, \dots, U_n)$. 
\end{lemma}

This simply relates the distribution of the median of correlated random variables to the median in the copula representation.

\begin{lemma}\label{lem:union_cop}
    Consider for a fixed $z$, a set of intervals $A_I = \cap_{i \in I} \{U_i \leq z\}$ where $I$ is a subset of $\{1,2,\dots, n\}$.  Consider a sequence of sets $A_J$ $J \in \mathcal{J}$ such that any pair of $J, J'$ contain at least $1$ overlapping index. Then 
    \begin{equation}
        \C(\cup_{J \in \mathcal{J}} A_J) \leq \W(\cup_{J \in \mathcal{J}} A_J)
    \end{equation}
    where $\W$ is the perfectly collinear measure corresponding to the upper Fr\'echet Hoeffding Bound \citep{Frechet1957SurProbabilite, Hoeffding1940MasstabvariarteKorrelationstheorie.}. 
\end{lemma}

This lemma is useful since it can describe an upper bound on the union of sets of intervals $A_J$. It is used in an intermediary step in the derivation of Lemma~\ref{lem:union_cop2}.

\begin{lemma}\label{lem:union_cop2}
    For any copula, $\C$
    \begin{equation}
        \C(\mathbb{\hat M}(U) \leq z) \leq 2 z
    \end{equation}
\end{lemma}

This lemma is useful in that it provides a worst-case bound on the median of uniform random variables.  From here the proof is short.

\begin{proof}
    The proof immediately is a result of Lemma~\ref{lem:med_uniform_marginal} and \ref{lem:union_cop2}. 
    \al{
        P(\mathbb{\hat M}(X) \leq t) &= \C( F^{-1}(\mathbb{\hat M} (U)) \leq t) \text{ Lemma~\ref{lem:med_uniform_marginal}}\\
        &= \C( \mathbb{\hat M} (U) \leq F(t)) \\
        &\leq 2 F(t) \text{ Lemma~\ref{lem:union_cop2}}\\
        &= 2P(X_1 \leq t)
    }
\end{proof}
This concludes the proof of Theorem~\ref{thm:copula_theorem}.  We subsequently prove the results of Lemmas~\ref{lem:med_uniform_marginal}, \ref{lem:union_cop}, \ref{lem:union_cop2}.

\subsection{Proof of Lemma~\ref{lem:med_uniform_marginal}}
\begin{proof}
    We can express $X = F^{-1}(U)$ where the vector $U$ is sampled from the copula $\C$. Then if $n$ is odd, with the subscript $(i)$ denoting the $i^{th}$ order statistic, $X_{((n + 1)/2)}$ is the median.  In this case, clearly $X_{(i)} = F^{-1}(U_{(i)})$. Now consider the case if $n$ is even.  If $z$ is any median of $U$ then $z \in [U_{(n/2)}, U_{(n/2 + 1)}]$ then 
    \al{
        z &\in [U_{(n/2)}, U_{(n/2 + 1)}] \\
        \iff F^{-1}(z) &\in [F^{-1}(U_{(n/2)}), F^{-1}(U_{(n/2 + 1)})] \\
        \iff F^{-1}(z) &\in [X_{(n/2)}, X_{(n/2 + 1)}]
    }
    Hence medians of the uniform distribution generated by the corresponding copula are mapped to the median of the observed variables. 
\end{proof}

\subsection{Proof of Lemma~\ref{lem:union_cop}}
\begin{proof}
    We prove this by induction. Consider the base case where there are two sets of intervals $A_{J_1}, A_{J_2}$. 
    Suppose there exists some copula $\Q$ such that $\Q(A_{J_1} \cup A_{J_2}) > \W(A_{J_1} \cup A_{J_2})$.  Firstly under the perfectly correlated copula $\W$ 
    \al{
        \W(A_{J_1} \cup A_{J_2}) &= \W(A_{J_1}) + \W(A_{J_2}) - \W(A_{J_1} \cap A_{J_2}) \\
        &= \min\{1,z\} + \min\{1,z\} - \min\{1,z\} \\
        &= z + z - z \\
        &= z
    }
    Then since $\{A_{J_1} \cup A_{J_2}\} \subset \{U_k \leq z\}$ for some index $k$.  Then 
    \al{
        \Q(A_{J_1} \cup A_{J_2}) &\leq \Q(U_k \leq z)
    }
    However, then this implies that the marginal distribution of $U_k$ under $\Q$ is not uniform, i.e. $\Q(U_k \leq z) > z$ and thus $\Q$ is not a copula. Therefore it must be the case that $\Q(A_1 \cup A_{J_2}) \leq \W(A_{J_1} \cup A_{J_2})$ holds in the base case.
    
    We next prove the induction step.  Suppose the following holds for any copula $\Q$ and a sequence of n sets of intervals $\mathcal{J}_n$. 
    \al{
        \Q(\cup_{J \in \mathcal{J}_n} A_j) \leq \W(\cup_{J \in \mathcal{J}_n} A_j)
    }
    for some $n$. Denote $B = \cup_{J \in \mathcal{J}_n} A_J$ then we must show for a new set of intervals $A_{n + 1}$
    \al{
        \Q(B \cup A_{n + 1}) \leq \W(B \cup A_{n + 1}).
    }
    Again, we prove this by contradiction. Suppose that there exists $\Q$ such that $\Q(B \cup A_{n + 1}) > \W(B \cup A_{n + 1})$
    \al{
        \W(B \cup A_{n + 1}) &= \W(B) + \W(A_{n + 1}) - \W(B \cap A_{n + 1})
    }
    Clearly by definition of $\W:$ $\W(B) = z$, $\W(A_{n + 1}) = z$ and $\W(B \cap A_{n + 1}) = z$. 
    Then 
    \al{
        \Q(B \cup A_{n + 1}) &= \Q(B)  + \Q(A_{n + 1}) - \Q(B \cap A_{n + 1}) \\
         &\leq \W(B) + \Q(A_{n + 1}) - \Q(B \cap A_{n + 1}) \\
        &\leq z + \underbrace{\Q(A_{n + 1}) - \Q(B \cap A_{n + 1})}_{\geq 0}
    }
    which generates a contradiction hence 
    \al{
        \Q(B \cup A_{n + 1}) \leq \W(B \cup A_{n + 1})
    }
\end{proof}
\subsection{Proof of Lemma~\ref{lem:union_cop2}}
\begin{proof}
    Define the events $A_J$ as in Lemma~\ref{lem:union_cop} of size $\lceil{n/2\rceil}$. Then the event $\{\mathbb{\hat M}(U) \leq z\}$ is equal to the union of all such $A_J$ as the median will be equivalent to the case when at least half of all uniforms are below $z$.  We can partition $\{A_J\}_{J \in \mathcal{J}}$ into two sets $S_1, S_2$ for which any two pairs of $A_J$ in a set have at least one overlapping index. This can be done by taking all the sets $\{A_J\}_{J \in \mathcal{J}}$ which suggest $\{U_1 \leq z\}$ and placing them into set $S_1$. Then all other sets must be placed in set $S_2$.  Since there are $(n - 1)$ possible remaining indices $\{2,3,\dots, n\}$ available for $S_2$ then any two events must have an overlapping index by the pigeonhole principle. Therefore 
    \al{
        \C(\mathbb{\hat M}(U) \leq z) &= \C(\cup A_i) \\
        &= \C(S_1 \cup S_2) \\
        &\leq \C(S_1) + \C(S_2) \text{ Union Bound}\\
        &\leq \W(S_1) + \W(S_2) \text{ Lemma~\ref{lem:union_cop}}\\
        &= 2z. 
    }
\end{proof}

\subsection{A corollary of Lemma~\ref{lem:union_cop2}}
A corollary immediately follows from Lemma~\ref{lem:union_cop2}.  This bounds the deviation of an arbitrary set of correlated uniform random variables by the distribution of its marginal. 

\begin{corollary}
    For any copula $\C$. 
    \begin{equation}
        \C(|\mathbb{\hat M}(U) - 1/2| > \epsilon) \leq 2 P(|U - 1/2| > \epsilon)
    \end{equation}
\end{corollary}
\begin{proof}
    $P(|U - 1/2| > \epsilon)$ is simply the marginal distribution of a uniform distribution. 
    \al{
        P(|U - 1/2| > \epsilon) &= \max\{1 - 2\epsilon, 0\} 
    }
    Next we note that
    \al{
        \{|\mathbb{\hat M}(U) - 1/2| > \epsilon\} &= \{\mathbb{\hat M}(U) < 1/2 - \epsilon\} \cap \{ \mathbb{\hat M}(U) > 1/2 + \epsilon\} \\
    }
    Note that we can define $1 - V = U$ where the measure of $V$, $\tilde \C$ is also a copula, since this is a joint distribution of marginally uniform variables.  
    \al{
        \{ \mathbb{\hat M}(U) > 1/2 + \epsilon\} &= \{\mathbb{\hat M}(V) < 1/2 - \epsilon\}
    }
    hence by a union bound. 
    \al{
        \C(|\mathbb{\hat M}(U) - 1/2| > \epsilon) &\leq \C(\mathbb{\hat M}(U) < 1/2 - \epsilon) + \tilde \C(\mathbb{\hat M}(V) < 1/2 - \epsilon) \\
        &= 2(1/2 - \epsilon) + 2(1/2 - \epsilon) 
    }
    then since the copula is non-negative 
    \al{
        \C(|\mathbb{\hat M}(U) - 1/2| > \epsilon) &\leq 2 \max\{1 - 2\epsilon, 0\} \\
        &= 2P(|U - 1/2| > \epsilon)
    }
\end{proof}

\section{Additional Computational Details}

\subsection{Newton Method for $\hat \kappa$}

Given a set of distances $\hat d$ we can estimate the curvature using a newton method.  Firstly, we compute the derivative of $g(\kappa,d)$ with respect to $\kappa$. 

\al{
    \frac{\partial}{\partial \kappa} g(\kappa, \vect{d}^{\triangleline}) &= (1/(4 \kappa^2)) \bigg(8 \cos(d_{xm} \sqrt{\kappa}) - 
  4 \cos(d_{xz} \sqrt{\kappa}) \text{sec}(d_{yz} \frac{\sqrt{\kappa}}{2}) \\
  &+ 
  4 d_{xm} \sqrt{\kappa} \sin(d_{xm} \sqrt{\kappa}) \\
  &- 
  2 d_{xy} \sqrt{\kappa} \text{sec}(d_{yz} \frac{\sqrt{\kappa}}{2}) \sin(d_{xy} \sqrt{\kappa}) \\
  &- 
  2 d_{xz} \sqrt{\kappa} \text{sec}(d_{yz} \frac{\sqrt{\kappa}}{2}) \sin(d_{xz} \sqrt{\kappa}) \\
  &+ 
  d_{yz} \sqrt{\kappa}
    \cos(d_{xz} \sqrt{\kappa}) \text{sec}(d_{yz} \frac{\sqrt{\kappa}}{2}) \text{tan}(d_{yz} \frac{\sqrt{\kappa}}{2}) \\
    &+ 
  \cos(d_{xy} \sqrt{\kappa}) \text{sec}((d_{yz} \sqrt{\kappa})/
    2) (-4 + d_{yz} \sqrt{\kappa} \text{tan}(d_{yz} \frac{\sqrt{\kappa}}{2})\bigg)
}

This allows us to construct a newton method for estimating the root $\hat \kappa$

\begin{equation*}
    \hat \kappa_{(m + 1)} = \hat \kappa_{(m)} - \frac{g(\hat \kappa_{(m)}, \hat d)}{\frac{\partial}{\partial \kappa} g(\hat \kappa_{(m)}, \hat d)}
\end{equation*}

\subsection{Distance Estimation} \label{sec:distance_estimation}

We recall the problem of estimating a distance matrix from a set of cliques $\mathcal{C}$.  Though this problem is convex, due to the $O(K^3)$ restrictions in the problem, it the problem is often slow to reach a solution in \texttt{CVXR}.  Instead, we solve this problem using a successive second order approximation.  

\al{
    &f(D):= \sum_{\mathpzc{X}, \mathpzc{Y} \in \mathcal{C}, i \in \mathpzc{X}, j \in Y} \bigg(A_{ij}\left( \nu_i + \nu_j - d_{xy}\right) \\
    &+ (1 - A_{ij})\log\left(1 - \exp\left( \nu_i + \nu_j - d_{xy}\right)\right) \bigg)\\
    &\approx \sum_{\mathpzc{X}, \mathpzc{Y} \in \mathcal{C}, i \in \mathpzc{X}, j \in Y} \bigg( A_{ij}\left( \nu_i + \nu_j - D_{0,x,y}\right) \\
    &+ (1 - A_{ij})\log\left(1 - \exp\left( \nu_i + \nu_j - D_{0,x,y}\right)\right) \\
    &+  \left( (A_{ij} -1)\frac{\exp\left( \nu_i + \nu_j - D_{0,x,y}\right)}{1 - \exp\left( \nu_i + \nu_j - D_{0,x,y}\right)} -A_{ij}\right)\left( d_{xy} - D_{0,x,y}\right) \\
    &+   (A_{ij} - 1)\frac{\exp\left( \nu_i + \nu_j - D_{0,x,y}\right)}{\left(1 - \exp\left( \nu_i + \nu_j - D_{0,x,y}\right)\right)^2}\frac{\left( d_{xy} - D_{0,x,y}\right)^2}{2} \bigg) \\
    &:= \tilde g(D,D_0)
}

Hence to compute the global solution $\vect{\hat D}$,  we can iteratively solve the following optimization problem
\al{
    \vect{\hat D}_{t +1} &= \text{argsup}_{D \in \in \R^{K \times K}} \tilde g(D,\vect{\hat D}_{t}) \\
     D_{ij} &\geq 0 \quad \text{ for all } i,j \\
     \text{Diag}(D) &= 0 \\
     \text{tr}(E_{s}^\top D) &\geq 0 \forall s \in \mathcal{S}
}

This process can be further sped up by choosing a good initialization matrix. We can do this by using the unconstrained maximum likelihood estimate $\vect{\hat D}_U$, which is very fast to compute but does not enforce triangle inequality restrictions. This can be computed analogously as in Theorem~\ref{thm:asymptotic_distance}.  
Though many of the distances $\vect{\hat D}_U$ may not satisfy the triangle inequality, we can trim the distances so that $\vect{\hat D}_U$ form a distance matrix, and use this as the starting point. 
The Floyd-Warshall Algorithm is a possible option for constructing a distance matrix from a noisy matrix which might not have a distance structure. A natural extension to this in our context is seen in Algorithm~\ref{alg:FWA}.

\begin{algorithm}[!t]
\caption{Adapted Floyd-Warshall Algorithm}\label{alg:FWA}
\begin{algorithmic}[1]
\Require $D \in \R^{K \times K}_{\geq 0}$
\Ensure $ D = D^\top$
\State Trim entries below 0: $  D[D < 0 ] \gets 0$
\For{$k \in \{1,2,\dots, n\}$}
{
    \For{$j \in \{1,2,\dots, n\}$}
    {
        \For{$i \in \{1,2,\dots, n\}$}
        {
            \If{$D_{ij} >  (D_{ik} + D_{kj})$}
            
                \State $D_{ij} \gets  (D_{ik} + D_{kj})$
            
            \EndIf
        }
        \EndFor
    }
    \EndFor
}
\EndFor
\end{algorithmic}
\end{algorithm}

One can draw similarities here to the problem of sparse metric repair. Metric repair seeks to adjust the fewest entries in a noisy distance matrix so that it still preserves the properties of being a metric (positivity, triangle inequality). \citet{gilbert2017if} illustrated this Floyd-Warshall algorithm to be a solution to a special type known as decrease only metric repair.

\section{Additional Discussion on the Latent Distance Model} \label{sec: additional discussion on the latent distance model}
\subsection{Other Link Functions} \label{sec: Other Link Functions}

Another common link function is the logistic link where the generative model for the network is as follows: 
\al{
    \nu_i &\sim F_{\nu}, \quad \nu_i \leq 0 \\
    Z_i &\sim F_Z, \quad Z_i \in \M^p \\
    P(A_{ij} = 1) &= \logit(\nu_i + \nu_j + \varphi - d(Z_i,Z_j)). 
}
We can consistently estimate the node level parameters up to a constant shift using conditional maximum likelihood as in the semiparametric Rasch model \citep{andersen1970asymptotic}. The parameter $\varphi$ controls the global sparsity. Similar to before, we note that $\nu_i$ terms are likely to be very large in a cliques then we can set the largest parameter in each group to be nearly zero.

Other link functions may be used but will likely all need specific methods to estimate $\nu_i$ parameters within each clique. However, we have shown how it can be developed in these two canonical cases. 

\subsection{Alternative Estimators of the Distance Matrix} \label{sec: Alternative Estimators of the Distance Matrix}
Additional methods for estimating the distance matrix can be developed. A promising direction is to utilize structured sparsity in the distance matrix. 

Our model exhibits numerous similarities to the $\beta$-model of network formation. In this framework, each node in the network possesses a gregariousness parameter $\beta$. An extension of this is the sparse beta model \cite{chen2021analysis}. In our context, $\nu$ functions analogously to $\beta$:

$$ P(A_{ij} = 1 | \nu) = \Lambda(\mu + \nu_i + \nu_j)$$

where $\|\nu\|_0 \leq s$ for some $s \ll n$. This parameterization facilitates a diminishing value of $\mu$ and is suitable for sparsely growing networks.

Furthermore, a distance matrix can be incorporated as follows:

$$ P(A_{ij} = 1 | \nu) = \Lambda(\mu + \nu_i + \nu_j - d_{ij})$$

where $\|D\|_{0} \leq s_D \ll n$, $D \in \mathcal{D}$. Here, $\mathcal{D}$ represents the convex region of matrices constrained to be distances. If sparsity is not maintained, the interior dimension of $\mathcal{D}$ for a set of $n$ points is $\binom{n}{2}$. Consequently, restricting the analysis solely to the distance matrix still results in $\binom{n}{2}$ observations, making estimation infeasible without additional shared structure, i.e. through sparsity of the distance matrix. 

Additionally, this formulation would be able to adjust to the sparsity level of the network data by asymptotically letting $\mu \to -\infty$. While a lasso-type procedure could be considered for estimating this distance matrix, a complete discussion of such methodologies is beyond the scope of this paper and is designated as future work.

\subsection{Rates of Clique Formation} \label{sec: rates of clique formation}
Here we clarify the notion of the likelihood of forming cliques of a given size. A famous result by \cite{grimmett1975colouring} states that the largest clique within an Erdos-Reyni random graph, $CN(n)$ 
$$ \frac{CN(n)}{\log(n)}\to_{a.s.}\frac{2}{\log(1/p)}$$
where $\to_{a.s.}$ indicates almost sure convergence.

The behaviour governing the formation of the cliques is determined by the concentration of the latent positions in the latent space.  Since this requires points that are exceedingly close together, the curvature of the space does not come into play here, but rather in the tendency of connections across cliques.  In the case of a hyperbolic space, this tends to generate tree-like structures between clusters of cliques, as seen in Figure~\ref{fig:curvature_partition:subfig:HEP_applied}.  Hence a hyperbolic space itself does not prevent the formation of cliques, so long as there as there is a reasonable concentration of the positions in the space. 

However, in our case, due to the fact that we assume that there is some continuous distribution of latent positions and gregariousness parameters, then we can illustrate a polynomial growth in the size of cliques. This results in probabilities of connection approaching arbitrarily close to $1$ rather than being bounded away from $1$ by a fixed probability.  Let $\mathfrak{S}(\ell)$ denote the combinations of nodes of size $\ell$ and $S \in \mathfrak{S}(\ell)$. 

Let $W_S$ denote the event that the nodes in $S$ form a clique. Our goal is to express the expected number of cliques of size $\ell$ generated from the model $\mathfrak{W} = \sum_{S \in \mathfrak{S}(\ell)} W_S$. 

\begin{theorem} \label{thm: lower_bound_clique_expectation}
    Let $\mathfrak{W}_{n,\ell}$ be the number of cliques formed of size $\ell$ formed from a network of size $n$.  If assumptions (D1) in \ref{lem:clique_localization_rate} and (E1) in \ref{lem:rand_eff_clique} hold, and $\ell = \OO(n^{1/(p + 2) - \epsilon'})$ for any $\epsilon' > 0$. Then the expected number of cliques of size $\ell$ diverges i.e.
    \begin{equation} \label{thm: num_cliques_diverge}
        \lim_{n \to \infty}\E[\mathfrak{W}_{n,\ell}] \to  \infty
    \end{equation}
\end{theorem}
\begin{proof}
We first lower bound $\E[W_S]$ then utilize the linearity of expectation in order to compute this lower bound. 

\al{
    \E[W_S] &= \int \prod_{i < j}\exp(\nu_i + \nu_j - d(Z_i,Z_j)) dF_{\vect{\nu}}(\vect{\nu}) dF_{\vect{Z}}(\vect{Z}) \\
     &= \int \exp( (\ell - 1)\sum_{i = 1}^\ell \nu_i - \sum_{i < j}^\ell d(Z_i,Z_j)) dF_{\vect{\nu}}(\vect{\nu}) dF_{\vect{Z}}(\vect{Z})
}
Next, using assumptions (D1) in \ref{lem:clique_localization_rate} and (E1) in \ref{lem:rand_eff_clique} we focus on the concentration of the latent positions. 

Suppose that $\epsilon_\delta = \{\max_{i = 1}^\ell d(Z_i, Z_j) \leq \delta\}$ and $\varepsilon_{\delta} = \{\min_{i = 1}^\ell \nu_i > -\delta\}$ then $P(\epsilon_\delta) \geq (c_2\delta/2)^{p\ell}$ and $P(\varepsilon_\delta) \geq (c_3\delta)^\ell$. 

Therefore: 
\al{
    \E[W_S] &\geq \exp(- 2\ell(\ell - 1)\delta ) (c_2\delta/2)^{p\ell}(c_3\delta)^\ell
}
and we let $\delta = \frac{\delta'}{\ell}$ for some $\delta' > 0$
\al{
    \E[W_S] &\geq \exp(- 2(\ell - 1)\delta' ) (c_2\delta'/(2\ell))^{p\ell}(c_3\delta'/\ell)^\ell \\
    &=  \exp(- 2(\ell - 1)\delta' - (p + 1)\ell\log(\ell))(c_2\delta'/(2))^{p\ell}(c_3\delta')^\ell
}
The dominant term here is the $\exp(- (p + 1)\ell\log(\ell))$. Next summing over $\mathfrak{S}(\ell)$. 
\al{
    \E[\mathfrak{W}] &= \sum_{S \in \mathfrak{S}(\ell)} \E[W_S] \\
    &\geq \binom{n}{\ell} \exp(- 2(\ell - 1)\delta' - (p + 1)\ell\log(\ell))(c_2\delta'/(2))^{p\ell}(c_3\delta')^\ell. 
}
We can lower bound the binomial coefficient $\binom{n}{\ell} \geq \left(\frac{n}{\ell}\right)^\ell$. Therefore, we consider the relationship between $n$ and $\ell$ such that the expected number of cliques of size $\ell$ grows to infinity.  We see that letting $\ell = n^{1/(p + 2) + \epsilon'}$ for any $\epsilon' > 0$ ensures that this lower bound diverges. 

Then for large $n$
\al{
    \E[\mathfrak{W}] &\gtrsim \left(\frac{n}{\ell} \right)^\ell \exp(- (p + 1)\ell\log(\ell)) \\
    &= \left(\frac{\ell^{(p + 2) + \epsilon'}}{\ell} \right)^\ell \exp(- (p + 1)\ell\log(\ell)) \\
    &= \exp\left( (p + 1 + \epsilon')\ell\log(\ell)  - (p + 1)\ell\log(\ell)\right) \\
    &\to_{\ell \to \infty } \infty
}
and hence we expect the number of cliques of this size to diverge to infinity. 
\end{proof}

\section{Riemannian Geometry Definitions} \label{sec: riemannian geometry definitions}

In this section, we review some definitions of the sectional and scalar curvature as well as the volume elements.  A Riemannian manifold $\mathcal{M} = (M,g)$ is a smooth manifold $M$ equipped with a Riemannian inner product $g_q$ on the tangent space $T_q(\M)$ at any point $q \in \mathcal{M}$, $g_q(u,v): T_q(\M) \times T_q(\M) \mapsto \R $.

This inner product can be used to define the Riemann curvature tensor at a point $q \in \M$ $R_q(u,v)w$, which takes $3$ vectors $u,v,w$ and returns an element of the tangent space
\al{
    R_q(u,v)w&: \quad  T_q(\M)  \times T_q(\M) \times T_q(\M) \mapsto T_q(\M) \\
    R_q(u,v)w &:= [\nabla_{u}, \nabla_{v}]w- \nabla_{[u,v]} w  
}
where $[u,v]$ is the lie bracket of vector fields and $[\nabla_{u}, \nabla_{v}]$ is the commuter of differential operators. The Riemann curvature tensor can be used to define our main quantity of interest, the sectional curvature at a point $\kappa_q(u,v): \times T_q(\M) \times T_q(\M) \mapsto \R$. The sectional curvature takes two linearly independent elements of the tangent space and maps them to the real line. 

$$ \kappa_q(u,v) := \frac{g_q(R_q(u,v)v, u)}{g_q(u,u)g_q(v,v) - g_q(u,v)^2}.$$

The sectional curvature is independent of the coordinate system used, but depends only on the linear subspace spanned by $u,v$.  Furthermore, in the canonical manifolds $\kappa_q(u,v) = \kappa$ by construction. 

From the sectional curvature, we can define the scalar curvature $S(m)$, 
$$  S(q) := \sum_{i \not = j } \kappa(e_i,e_j)$$
where $\{e_i\}_{i = 1}^p$ form an orthonormal frame for $T_q(\M)$.  We can think of the scalar curvature as an ``average of sectional curvatures" across the manifold. \\

Next we consider the distance induced by the metric tensor. Given a smooth curve \(\gamma: [a, b] \to \M^p\) with \(\gamma(a) = q_1\) and \(\gamma(b) = q_2\), the \textbf{length} \(L(\gamma)\) of \(\gamma\) is defined by:
\[
L(\gamma) = \int_a^b \sqrt{g_{\gamma(t)}\left(\dot{\gamma}(t), \dot{\gamma}(t)\right)} \, dt,
\]
where \(\dot{\gamma}(t)\) is the tangent vector to the curve \(\gamma\) at time \(t\).

The Riemannian distance \(d(q_1, q_2)\) between two points \(q_1, q_2 \in \M^p\) is defined as:
\[
d(q_1, q_2) = \inf_{\gamma} L(\gamma),
\]
where the infimum is taken over all smooth curves \(\gamma: [a, b] \to M\) such that \(\gamma(a) = q_1\) and \(\gamma(b) = q_2\).

We lastly define a volume form (also known as the Levi-Civita Tensor) via the Riemannian inner product $g$.  If $\omega$ is a local oriented coordinate system near a point $q$ then 
\al{
    dV &:= \sqrt{|\text{det}(g)|} d\omega. 
}
where $g_{q}$ is the metric tensor evaluated on the basis coordinate system $\omega$.  For further details on these quantities, see \citet{Klingenberg1995RiemannianGeometry}. 

From this definition of a volume, we can define probability density functions on the manifold.  A density function $f$ corresponding to a measure $F$ with support on the manifold can be defined as follows. For a set $\mathcal{X} \subset \M$.

$$ P(X \in \mathcal{X}) = \int_{x \in \mathcal{X}} f(x) dV_x$$

See \citet{Pennec1999ProbabilitiesMeasurements.} for further introduction for defining probabilities on the manifold.

\section{Assumptions on $\M$} \label{sec:manifold_assumptions}

Here we verify that the Algebraic Midpoint properties, as well as locally Euclidean properties are satisfied for a complete simply connected smooth Riemannian manifolds. 

If the algebraic midpoint property is satisfied for any complete metric space $\Mf$ then by Theorem 1.8 of \citet{2007MetricSpaces}, then $\Mf$ is a \textbf{path metric space}. The authors follow up in discussion a list of examples of path metric spaces, which include Riemannian manifolds with boundary. 

Secondly, if $\M^p$ has scalar curvature at point $q$, $S(q)$ then 
\al{
    \frac{\text{Vol}(B_{\M^p}(\epsilon, q))}{\text{Vol}(B_{\E^p}(\epsilon, q))} &= 1 - \frac{S(q)}{6(p + 2) }\epsilon^2 + o(\epsilon^3)
}
by Theorem 3.98 of \citet{Gallot2004RiemannianGeometry}. Since this holds, then for a latent metric which is generated by distances on a Riemannian manifold, the locally Euclidean volume property will hold.

\section{Graph Statistics From Simulations}\label{sec:graph_statistics}

Columns denote the scale factor used in the simulations.

\begin{table}[!h] \centering 

  \caption{$\kappa = -2$ graph statistics summary. } 
  \label{tab:kappa_m2} 
\begin{tabular}{@{\extracolsep{5pt}} ccccc} 
\\[-1.8ex]\hline 
\hline \\[-1.8ex] 
Scale ($\rho$) & 0.7 & 1 & 2  \\ 
\hline \\[-1.8ex] 
Edge.fraction & 0.016 (0.001) & 0.012 (0.001) & 0.007 (0.001)  \\ 
Max.Degree & 341.3 (23.596) & 412.29 (26.364) & 656.63 (33.954) \\ 
Mean.Degree & 56.726 (5.279) & 58.538 (5.351) & 73.975 (5.448)  \\ 
Distinct.Cliques$\geq \ell$ & 78.665 (6.955) & 53.675 (5.427) & 43.515 (4.41)  \\ 
Max.Clique.Size & 24.215 (4.985) & 28.725 (5.494) & NaN (NA)  \\ 
Mean.Degree.Centrality & 0.155 (0.009) & 0.134 (0.007) & 0.107 (0.005) \\ 
\hline \\[-1.8ex] 
\end{tabular} 
\end{table}

\begin{table}[!h] \centering 
  \caption{$\kappa = -1$ graph statistics summary. } 
  \label{tab:kappa_m1}  
\begin{tabular}{@{\extracolsep{5pt}} ccccc} 
\\[-1.8ex]\hline 
\hline \\[-1.8ex] 
 Scale ($\rho$) & 0.7 & 1 & 2 \\ 
\hline \\[-1.8ex] 
Edge.fraction & 0.017 (0.002) & 0.012 (0.001) & 0.008 (0.001) \\ 
Max.Degree & 343.69 (25.21) & 416.825 (28.796) & 661.12 (32.257) \\ 
Mean.Degree & 58.943 (5.354) & 61.168 (5.636) & 78.239 (5.6)  \\ 
Distinct.Cliques$\geq \ell$ & 80.14 (6.488) & 54.965 (5.404) & 42.32 (5.027)  \\ 
Max.Clique.Size & 24.005 (5.456) & 28.735 (6) & NaN (NA)  \\ 
Mean.Degree.Centrality & 0.159 (0.008) & 0.137 (0.007) & 0.112 (0.005) \\ 
\hline \\[-1.8ex] 
\end{tabular} 
\end{table}

\begin{table}[!h] \centering 
  \caption{$\kappa = -0.5$ graph statistics summary. } 
  \label{tab:kappa_m0.5}  
\begin{tabular}{@{\extracolsep{5pt}} ccccc} 
\\[-1.8ex]\hline 
\hline \\[-1.8ex] 
 Scale ($\rho$) & 0.7 & 1 & 2  \\ 
\hline \\[-1.8ex] 
Edge.fraction & 0.017 (0.001) & 0.013 (0.001) & 0.008 (0.001)  \\ 
Max.Degree & 340.785 (21.648) & 413.115 (24.818) & 656.07 (39.433)  \\ 
Mean.Degree & 59.658 (4.776) & 62.67 (4.989) & 79.873 (6.566)  \\ 
Distinct.Cliques$\geq \ell$ & 80.205 (6.493) & 55.155 (5.288) & 39.95 (5.289)  \\ 
Max.Clique.Size & 23.62 (5.216) & 29.52 (6.073) & NaN (NA)  \\ 
Mean.Degree.Centrality & 0.161 (0.009) & 0.141 (0.007) & 0.114 (0.005)  \\ 
\hline \\[-1.8ex] 
\end{tabular} 
\end{table}

\begin{table}[!h] \centering 
  \caption{$\kappa = 0$ graph statistics summary. } 
  \label{tab:kappa_0}  
\begin{tabular}{@{\extracolsep{5pt}} ccccc} 
\\[-1.8ex]\hline 
\hline \\[-1.8ex] 
 Scale ($\rho$) & 0.7 & 1 & 2  \\ 
\hline \\[-1.8ex] 
Edge.fraction & 0.012 (0.001) & 0.008 (0) & 0.005 (0)  \\ 
Max.Degree & 205.38 (16.619) & 242.88 (18.178) & 371.805 (24.336)  \\ 
Mean.Degree & 40.815 (2.626) & 41.759 (2.415) & 52.465 (2.813)  \\ 
Distinct.Cliques$\geq \ell$ & 70.2 (13.19) & 46.84 (6.057) & 41.54 (4.576)  \\ 
Max.Clique.Size & 23.735 (5.244) & 28.49 (5.421) & 41.63 (7.913)  \\ 
Mean.Degree.Centrality & 0.169 (0.017) & 0.149 (0.013) & 0.125 (0.009)  \\ 
\hline \\[-1.8ex] 
\end{tabular} 
\end{table}

\begin{table}[!h] \centering 
  \caption{$\kappa = 0.5$ graph statistics summary. } 
  \label{tab:kappa_0.5}  
\begin{tabular}{@{\extracolsep{5pt}} ccccc} 
\\[-1.8ex]\hline 
\hline \\[-1.8ex] 
 Scale ($\rho$) & 0.7 & 1 & 2  \\ 
\hline \\[-1.8ex] 
Edge.fraction & 0.017 (0.001) & 0.013 (0.001) & 0.008 (0) \\ 
Max.Degree & 267.08 (17.432) & 318.615 (18.062) & 488.545 (25.046)  \\ 
Mean.Degree & 61.311 (2.803) & 63.539 (2.96) & 80.373 (3.454)  \\ 
Distinct.Cliques$\geq \ell$ & 87.945 (12.778) & 55.385 (7.343) & 39.535 (5.591)  \\ 
Max.Clique.Size & 23.595 (4.844) & 28.45 (6.158) & 39.925 (7.861) \\ 
Mean.Degree.Centrality & 0.21 (0.016) & 0.183 (0.012) & 0.153 (0.008)  \\ 
\hline \\[-1.8ex] 
\end{tabular} 
\end{table}

\begin{table}[!h] \centering 
  \caption{$\kappa = 1$ graph statistics summary. } 
  \label{tab:kappa_1}  
\begin{tabular}{@{\extracolsep{5pt}} ccccc} 
\\[-1.8ex]\hline 
\hline \\[-1.8ex] 
 Scale ($\rho$) & 0.7 & 1 & 2  \\ 
\hline \\[-1.8ex] 
Edge.fraction & 0.025 (0.001) & 0.018 (0.001) & 0.012 (0)  \\ 
Max.Degree & 351.89 (16.965) & 421.36 (18.32) & 652.265 (26.131)  \\ 
Mean.Degree & 88.331 (3.89) & 91.901 (3.805) & 117.569 (4.089) \\ 
Distinct.Cliques$\geq \ell$ & 91.71 (30.245) & 64.23 (17.006) & 41.765 (4.761)  \\ 
Max.Clique.Size & 24.07 (5.334) & 29.29 (6.259) & 40.52 (7.969)  \\ 
Mean.Degree.Centrality & 0.241 (0.012) & 0.21 (0.01) & 0.174 (0.007)  \\ 
\hline \\[-1.8ex] 
\end{tabular} 
\end{table} 

\section{Values of Tuning Parameter $C_{\triangle}$ Used in Simulations And Applications} \label{sec: values of C_delta}

Here we provide details on the choice of $C_{\triangle}$ used in various simulations. 

\begin{table}[H]
\centering
\begin{tabular}{|l|c|}
\hline Section (Figure) & $C_{\triangle}$ \\ \hline 
\ref{sec: Consistency of Simulation Curvature Estimates} (Figure~\ref{fig:Estimator_Latent_GMM_consistency})       & 1.5     \\
\ref{sec:simulations_test_type1_error} (Figure~\ref{fig:Estimator_Type_I_error})      & 1.2     \\
\ref{sec: multiplex network constant curvature tests} (Figure~\ref{fig:power_multiview})      & 1.2     \\
\ref{sec: Noncanonical Manifolds} (Figure~\ref{fig:power_ad_sphere}) & 1.3 \\
 \ref{sec:application_cybersecurity} (Figure~\ref{fig:rmse_changepoint})               & 1.4 \\
\ref{sec:application_cybersecurity} (Figure~\ref{fig:LANL_curvature_time_series})      & 1.4  \\
\hline 
\end{tabular}
\end{table}

\section{Additional Miscellanea} \label{sec:}
\subsection{Embeddings and Graph Distances} \label{sec: Embeddings and Graph Distances}
Our focus on this paper is the estimation of curvature of latent spaces, however, one may consider a similar problem of embedding the graph distances ($D_G$), (i.e. shortest path distances) as in \cite{gu2018learning}.  This is a distinct, but related problem that we can study using our estimating equation for curvature.  Firstly, we note that using a set of graph distances will generally allow for the formation of many good quality midpoints, as any chain of $3$ node forms a midpoint set.  Our method may be useful here for identifying nodes who's distances may not be preserved well using a standard embedding in a space of constant curvature.  We leave this possible extension of our method as future work.


\subsection{Smoothness of The Estimating Equation} \label{sec: Smoothness of the Estimating Equation}
Note: We moved this section to the main text.
To highlight the smoothness of our estimating function we plot a set of examples. For a unit equilateral triangle, we compute the corresponding midpoint distance $d_{xm}$ for each curvature space with $\kappa \in \{-2,-1,0,1,2\}$.  We see in Figure~\ref{fig:example_ee} that our proposed estimating equation is differentiable around the solution with non-zero derivative, allowing one to identify the curvature from the $\kappa: g(\kappa,d) = 0$. 

\begin{figure}
    \centering
     \includegraphics[width = 0.4\textwidth]{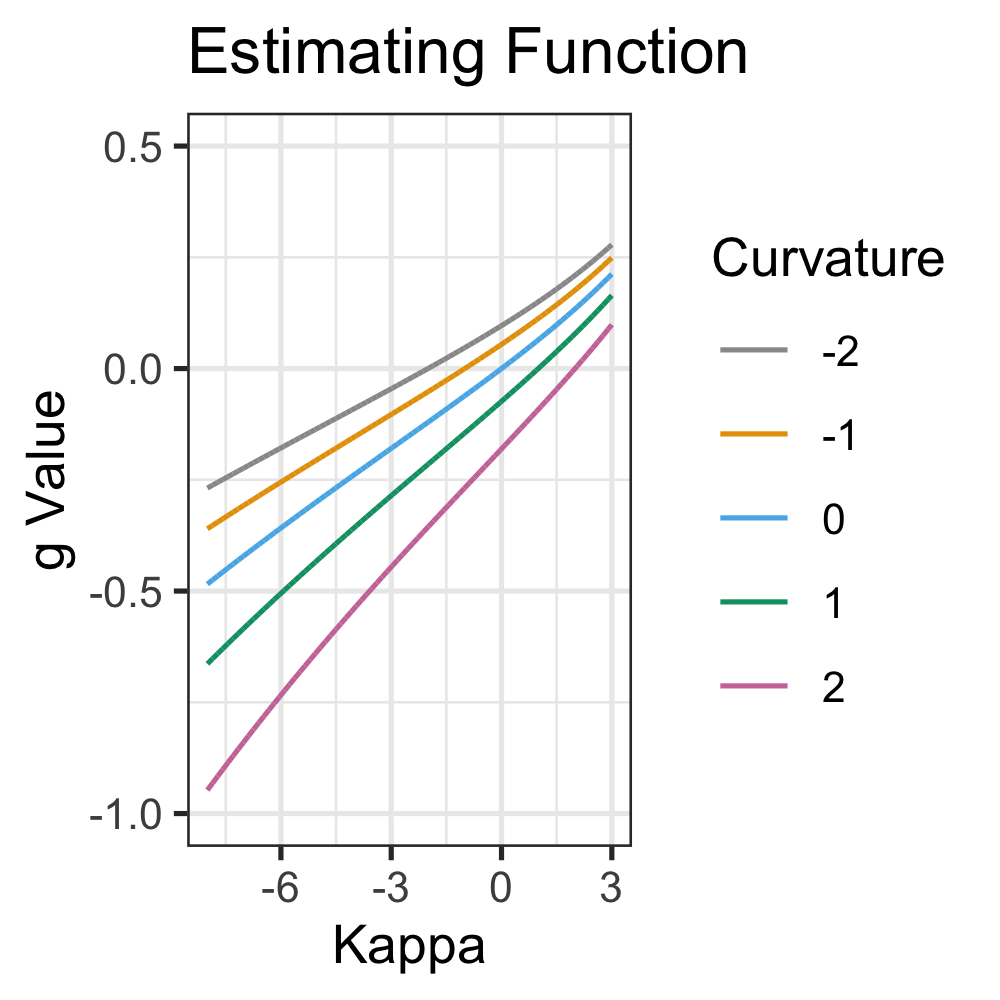}
    \caption{Example of estimating functions $g$ as a function of $\kappa$. }
    \label{fig:example_ee}
\end{figure}


\end{document}